\documentclass[11pt]{article}
\usepackage{amsmath,amsthm,amssymb}
\usepackage{fullpage}
\usepackage{enumerate}

\usepackage{tabularx}
\usepackage{tabulary}
\usepackage{multirow}

\usepackage{caption,subcaption}

\usepackage{tikz}
\usetikzlibrary{arrows.meta}
\tikzset{>={Latex[width=2mm,length=2mm]}}

\usetikzlibrary{decorations.pathreplacing}

\usepackage{tkz-graph}
\tikzstyle{vertex}=[circle, draw, inner sep=0pt, minimum size=6pt]

\newcommand{\IGNORE}[1]{}
\newcommand{\noi}{{\noindent}}

\newcommand{\ds}{\ensuremath{\displaystyle}}
\newcommand{\wt}{\ensuremath{\widetilde}}
\newcommand{\wh}{\ensuremath{\widehat}}

\newcommand{\opt}{\textsl{opt}}
\newcommand{\optsol}{\textsl{OPT}}
\newcommand{\ALGO}{\textsl{ALGO}}

     \newcommand\ubpcomp[3]{{#1}^{\{#2,#3\}}}
     \newcommand\sbpcomp[3]{{#1}^{\oslash}}

     \newcommand\vbpcomp[3]{{#1}}

\newcommand{\nbp}{\ensuremath{\#\textup{bp}}}

\newcommand{\ncomp}{\ensuremath{\#\textup{comp}}}

\newcommand{\pstart}{a_0}
\newcommand{\pend}{z_0}
\newcommand{\auxpstart}{a_{\star}}
\newcommand{\auxpend}{z_{\star}}

\newcommand{\earP}{\widehat{P}}
\newcommand{\dblearP}{\widetilde{P_{\star}}}

\newcommand{\LB}{{\mathit{l}}{\mathit{b}}}

\newtheorem{theorem}{Theorem} % [section]
\newtheorem{proposition}[theorem]{Proposition}
\newtheorem{lemma}[theorem]{Lemma}
\newtheorem{fact}[theorem]{Fact}
\newtheorem{defn}{Definition} % [section]

\newcommand{\authremark}[1]{\medskip\noi\textbf{Remark}:{#1}\newline\medskip}

\newcommand{\floor}[1]{\lfloor{#1}\rfloor}
\newcommand{\ceiling}[1]{\lceil{#1}\rceil}

\newcommand{\DTWO}{{\textup{D2}\/}}
\newcommand{\dtwocost}{\tau}
\newcommand{\ndtwocost}{\widehat{\tau}}

\newcommand{\cost}{\textup{cost}}
\newcommand{\credits}{\textup{credits}}

\def\DEBUG{true}

\ifdefined\DEBUG
 
 \newcommand{\jos}[1]{\textcolor{orange}{\fbox{#1}}}
 \newcommand{\vishnu}[1]{\textcolor{blue}{#1}}

  \def\rem#1{{\marginpar{\raggedright\scriptsize #1}}}
  \newcommand{\josrt}[1]{\rem{\textcolor{orange}{$\bullet$ #1}}}
  \newcommand{\vishnurt}[1]{\rem{\textcolor{blue}{$\bullet$ #1}}}
\else
  \newcommand{\jos}[1]{}
  \newcommand{\vishnu}[1]{}
  \newcommand{\josrt}[1]{}
  \newcommand{\vishnurt}[1]{}
\fi

\begin{document}

\title
{
The Matching Augmentation Problem:
\\
A $\frac74$-Approximation Algorithm
}
\author{
J.Cheriyan	\thanks{University of Waterloo, Waterloo, Canada}
\and
J.Dippel	\thanks{University of Waterloo, Waterloo, Canada}
\and
F.Grandoni	\thanks{IDSIA, USI-SUPSI, Lugano, Switzerland}
\and
A.Khan		\thanks{Technical University of Munich, Garching, Germany
	}
\and
V.V.Narayan	\thanks{McGill University, Montreal, Canada}
}

  \date{December 4, 2018}

\maketitle

\vspace{-0.25in}

\begin{abstract}
We present a $\frac74$ approximation algorithm for the matching augmentation
problem (MAP): given a multi-graph with edges of cost either zero or one
such that the edges of cost~zero form a matching, find a 2-edge connected
spanning subgraph (2-ECSS) of minimum cost.

We first present a reduction of any given MAP instance to a collection of
well-structured MAP instances such that the approximation guarantee is
preserved. Then we present a $\frac74$ approximation algorithm for a
well-structured MAP instance. The algorithm starts with a min-cost 2-edge
cover and then applies ear-augmentation steps. We analyze the cost of the
ear-augmentations using an approach similar to the one proposed by Vempala
 \& Vetta for the (unweighted) min-size 2-ECSS problem (``Factor 4/3
approximations for minimum 2-connected subgraphs,'' APPROX 2000, LNCS 1913,
pp.262--273).

\medskip

\noindent
{\bf Keywords}:
2-edge~connected graph,
2-edge~covers,
approximation algorithms,
bridges,
connectivity augmentation,
forest   augmentation problem,
matching augmentation problem,
network design.
\end{abstract}

\bigskip

\renewcommand{\thefootnote}{\fnsymbol{footnote}}
\setcounter{footnote}{6}

\section{ \label{s:intro}  Introduction}

A basic goal in the area of \emph{survivable network design} is to
design real-world networks of low cost that provide connectivity
between pre-specified pairs of nodes even after the failure of a
few edges/nodes.
Many of the problems in this area are NP-hard, and significant
efforts have been devoted in the last few decades to the design of
approximation algorithms, see \cite{WS:book}.

One of the fundamental problems in the area is the minimum-cost
2-edge connected spanning subgraph problem (abbreviated as min-cost 2-ECSS):
given a graph together with non-negative costs for the edges,
find a 2-edge connected spanning subgraph (abbreviated as 2-ECSS) of minimum cost.
This problem is closely related to the famous Traveling Salesman
Problem (TSP), and some of the earliest papers in the area of
approximation algorithms address the min-cost 2-ECSS problem \cite{FJ:sicomp81,FJ:tcs82}.
In the context of approximation algorithms,
this research led to the discovery of algorithmic paradigms
such as the \textit{primal-dual method} \cite{GW95,WS:book}
and the \textit{iterative rounding method} \cite{J01,LRS:book},
and led to dozens of publications.
Under appropriate assumptions, these methods achieve an approximation guarantee of~2
for several key problems in survivable network design, including min-cost 2-ECSS.
Unfortunately, these generic methods do not achieve approximation guarantees below~2.
Significant research efforts have been devoted to achieving approximation guarantees
below~2 for specific problems in the area of survivable network design.
For example, building on earlier work, an approximation guarantee of $\frac43$ has been
achieved for unweighted (\textit{min-size}) 2-ECSS \cite{SV:cca},
where each edge of the input graph has cost~one and the goal is to find
a 2-ECSS with the minimum number of edges.

There is an important obstacle beyond unweighted problems, namely,
the special case of min-cost 2-ECSS where the (input) edges have cost of zero or one,
and the aim is to design an algorithm that achieves an approximation guarantee below~2.
This problem is called the \textit{Forest Augmentation Problem} (FAP).
In more detail, we are given an undirected graph $G=(V,\;E_0\cup E_1)$,
where each edge in $E_0$ has cost~zero and each edge in $E_1$ has cost~one;
the goal is to compute a 2-ECSS $H=(V,F)$ of minimum cost.
We denote the cost of an edge $e$ of $G$ by $\cost(e)$,
and for a subgraph $G'$ of $G$, $\cost(G')$ denotes $\sum_{e\in E(G')}\cost(e)$.
Observe that $\cost(H)=|F\cap E_1|$,
so the goal is to augment $E_0$ to a 2-ECSS by adding the minimum
number of edges from $E_1$.
Intuitively, the zero-edges define some existing network that we
wish to augment (with edges of cost~one) such that the augmented
network is resilient to the failure of any one edge.
Without loss of generality (w.l.o.g.) we may contract each of the 2-edge~connected subgraphs formed by the 
zero-edges, and hence, we may assume that $E_0$ induces a forest:
this motivates the name of the problem.

A key special case of FAP is the \textit{Tree Augmentation Problem} (TAP),
where the edges of cost~zero form a spanning tree.
Nagamochi \cite{Na03} first obtained an approximation guarantee below~2 for TAP,
and since then there have been several advances
including recent work, see \cite{Adj17,EFKN09,FGKS17,GKZ18,KN:talg16}.

We focus on a different special case of FAP called the \textit{matching
augmentation problem} (MAP): given a multi-graph with edges of cost
either zero or one such that the edges of cost~zero form a matching,
find a 2-ECSS of minimum cost.  Note that loops are not allowed;
multi-edges (parallel edges) are allowed.
From the view-point of approximation algorithms,
MAP is ``orthogonal'' to TAP
in the sense that the forest of zero-cost edges
has many connected~components in MAP,
whereas this forest has only one connected~component in TAP.
In our opinion, MAP (like TAP) is an important special case of FAP
in the sense that none of the previous approaches
(including approaches developed for TAP over two decades)
give an approximation guarantee below~2 for MAP.

\subsection{Previous literature \& possible approaches for MAP}

There is extensive literature on
approximation algorithms for problems in survivable network design, and also on
the minimum-cost 2-ECSS problem including its key special cases
(including the unweighted problem, TAP, etc.).
To the best of our knowledge,
there is no previous publication on FAP or MAP,
although the former is well known to the researchers working in this area.

Let us explain briefly why previous approaches do not help for obtaining
an approximation guarantee below~2 for MAP.
Let $G$ denote the input~graph, and let $n$ denote $|V(G)|$.
Let $\opt$ denote the optimal value, i.e., the minimum cost of a 2-ECSS
of the given instance.
Recall that the standard cut-covering LP~relaxation of the min-cost 2-ECSS problem
has a non-negative variable $x_e$ for each edge $e$ of $G$,
and for each nonempty set of nodes $S$, $S\not=V$, there is a constraint
requiring the sum of the $x$-values in the cut $(S,V-S)$ to be $\ge2$;
the objective is to minimize $\sum_{e\in{E}} \cost(e) x_e$.

The primal-dual method and the iterative rounding method
are powerful and versatile methods for rounding LP~relaxations, but
in the context of FAP,
these methods seem to be limited to proving approximation guarantees
of at~least~2.

Several intricate combinatorial methods that may
also exploit lower-bounds from LP~relaxations
have been developed for approximation algorithms for
unweighted 2-ECSS,
e.g., the $\frac43$-approximation algorithm of \cite{SV:cca}.
For unweighted 2-ECSS, there is a key lower bound of $n$ on $\opt$
(since any solution must have $\ge n$ edges, each of cost~one).
This no longer holds for MAP; indeed, the analogous lower bound on $\opt$
is $\frac12 n$ for MAP.
It can be seen that an $\alpha$-approximation algorithm for unweighted
2-ECSS implies a $(3\alpha-2)$-approximation algorithm for MAP.
(We sketch the reduction:
let $M$ denote the set of zero-cost edges in an instance of MAP;
observe that $|M|\leq\opt$;
we subdivide (once) each edge in $M$, then we change all edge costs to
one, then we apply the algorithm for unweighted 2-ECSS, and finally we
undo the initial transformation; the optimal cost of the unweighted
2-ECSS instance is $\le\opt+2|M|$, hence, the solution of the MAP
instance has cost
$\leq\alpha(\opt+2|M|)-2|M|=\alpha\opt+(2\alpha-2)|M|\leq(3\alpha-2)\opt$.)
Thus the $\frac43$-approximation algorithm of \cite{SV:cca} for
unweighted 2-ECSS gives a $2$-approximation algorithm for MAP.
(Although preliminary results and claims have been published on
achieving approximation guarantees below~$\frac43$ for unweighted
2-ECSS, there are no refereed publications to date, see \cite{SV:cca}.)

Over the last two decades, starting with the work of \cite{Na03},
a few methods have been developed to obtain approximation guarantees
below~2 for TAP.
The recent methods of \cite{Adj17,FGKS17} rely on so-called bundle
constraints defined by paths of zero-cost edges.
Unfortunately, these methods (including methods that use the bundle
constraints) rely on the fact that the set of zero-cost edges forms a
connected graph that spans all the nodes, see \cite{FGKS17,KN:talg16,Na03}.
 Clearly, this property does not hold for MAP.

\subsection{Hardness of approximation of MAP and FAP}

MAP is a generalization of the unweighted 2-ECSS problem
(consider the special case of MAP with $M=\emptyset$).
The latter problem is known to be APX-hard;
thus, it has a ``hardness of approximation'' threshold of $1+\epsilon$
where $\epsilon>0$ is a constant, see \cite{GGTW09}.
Hence, MAP is APX-hard.

Given the lack of progress on approximation algorithms for FAP, one
may wonder whether there is a ``hardness of approximation'' threshold
that would explain the lack of progress.
Unfortunately, the results and techniques from the area of ``hardness
of approximation'' are far from the known approximation guarantees
for many problems in network design.
For example, even for the notorious Asymmetric TSP (ATSP),
the best ``hardness of approximation'' lower~bound known
is around $\frac{75}{74}\approx1.014$, see \cite{KLS15}.

\subsection{Our method for MAP}
We first present a reduction of any given instance of MAP to a
collection of well-structured MAP instances such that the approximation
guarantee is preserved, see
Sections~\ref{s:prelims},~\ref{s:algo},~\ref{s:pre-proc1}.
Then we present a $\frac74$ approximation algorithm for a
well-structured MAP instance,
see Sections~\ref{s:algo},~\ref{s:path-thicken},~\ref{s:algo-last}.
Our algorithm starts with a so-called \DTWO\
(this is a min-cost 2-edge cover) and
then applies ear-augmentation steps.
We analyze the cost of the
ear-augmentations using an approach similar to the one proposed by
Vempala \& Vetta for the unweighted 2-ECSS problem \cite{VV00}.
Our presentation is self-contained and formally independent of
Vempala \& Vetta's manuscript; also, we address a weighted version of
the 2-ECSS problem and our challenge is to improve on the approximation
guarantee of~2,
whereas Vempala \& Vetta's goal is to achieve an approximation
guarantee of~$\frac43$ for the unweighted 2-ECSS problem.

For the sake of completeness, we have included the proofs of several
basic results (e.g., so-called Facts); these should not be viewed as new
contributions.

An outline of the paper follows.
Section~\ref{s:prelims} has standard definitions and some preliminary results.
Section~\ref{s:algo} presents an outline of our algorithm for MAP,
and explains what is meant by a well-structured MAP instance.
Section~\ref{s:pre-proc1} presents the
pre-processing steps that give an approximation preserving reduction
from any instance of MAP to a collection of well-structured MAP instances;
some readers may prefer to skip this section
(and refer back to the results/details as needed).
Sections~\ref{s:path-thicken},~\ref{s:algo-last} present the
$\frac74$ approximation algorithm for 
well-structured MAP instances, and prove the approximation guarantee.
Section~\ref{s:lowerbounds} presents examples that give
lower~bounds on our results on MAP.
The first example gives a construction
such that $\opt \approx \frac74 \dtwocost$,
where $\dtwocost$ denotes the minimum cost of a 2-edge~cover.
The second example gives a construction such that the cost of the
solution computed by our algorithm is $\approx\frac74\opt$.

\section{ \label{s:prelims}  Preliminaries}

This section has definitions and preliminary results.
Our notation and terms are consistent with \cite{Diestel},
and readers are referred to that text for further information.

Let $G=(V,E)$ be a (loop-free) multi-graph with edges of cost either zero or one
such that the edges of cost~zero form a matching.
We take $G$ to be the input graph, and
we use $n$ to denote $|V(G)|$.
Let $M$ denote the set of edges of cost~zero.
Throughout, the reader should keep in mind that $M$ is a matching;
this fact is used in many of our proofs without explicit reminders.
We call an edge of $M$ a \textit{zero-edge}
and we call an edge of $E-M$ a \textit{unit-edge}.
We call a pair of parallel edges a $\{0,1\}$-edge-pair
if one of the two edges of the pair
has cost~zero and the other one has cost~one.

We use the standard notion of contraction of an edge, see \cite[p.25]{Schrijver}:
Given a multi-graph $H$ and an edge $e=vw$,
the contraction of $e$ results in the multi-graph $H/(vw)$ obtained from $H$
by deleting $e$ and its parallel copies and identifying the nodes $v$ and $w$.
(Thus every edge of $H$ except for $vw$ and its parallel copies
is present in $H/(vw)$; we disallow loops in $H/(vw)$.)

For a graph $H$ and a set of nodes $S\subseteq V(H)$,
$\delta_H(S)$ denotes the set of edges that have one end node in
$S$ and one end node in $V(H)-S$;
moreover, $H-S$ denotes $H[V(H)-S]$, the subgraph of $H$ induced by $V(H)-S$.
For a graph $H$ and a set of edges $F\subseteq E(H)$,
$H-F$ denotes the graph $(V(H),~E(H)-F)$.
We use relaxed notation for singleton sets, e.g.,
we use $\delta_H(v)$ instead of $\delta_H(\{v\})$,
we use $H-v$ instead of $H-\{v\}$, and
we use $H-e$ instead of $H-\{e\}$.

We denote the cost of an edge $e$ of $G$ by $\cost(e)$.
For a set of edges $F\subseteq E(G)$, $\cost(F):=\sum_{e\in F}\cost(e)$,
and for a subgraph $G'$ of $G$, $\cost(G'):=\sum_{e\in E(G')}\cost(e)$.

For ease of exposition, we often denote an instance $G,M$ by $G$;
then, we do not have explicit notation for the edge~costs of the instance,
but the edge~costs are given implicitly by $\cost:E(G)\rightarrow\{0,1\}$,
and $M$ is given implicitly by $\{e\in{E(G)}:\cost(e)=0\}$.
Also, we may not distinguish between a subgraph and its node~set;
for example, given a subgraph $H$ that contains nodes $v_1,v_2,v_3,v_4,\dots$
we may say that $\{v_1,v_2,v_3\}$ is contained in $H$.

\subsection{2EC, 2NC, bridges and \DTWO}

A multi-graph $H$ is called $k$-edge connected if $|V(H)|\ge2$ and for
every $F\subseteq E(H)$ of size $<k$, $H-F$ is connected.
Thus, $H$ is 2-edge connected if it has $\ge2$ nodes and the deletion
of any one edge results in a connected graph.
A multi-graph $H$ is called $k$-node connected if $|V(H)|>k$ and for
every $S\subseteq V(H)$ of size $<k$, $H-S$ is connected.
We use the abbreviations \textit{2EC} for ``2-edge connected," and
\textit{2NC} for ``2-node connected."

We assume w.l.o.g.\ that the input $G$ is 2-edge~connected.
Moreover, we assume w.l.o.g.\ that
there are $\leq 2$ copies of each edge (in any multi-graph that we consider);
this is justified since an edge-minimal 2-ECSS cannot have
three or more copies of any edge (see Proposition~\ref{propo:2ecdiscard} below).

For any instance $H$, let $\opt(H)$ denote the minimum cost of a
2-ECSS of $H$.  When there is no danger of ambiguity, we use $\opt$
rather than $\opt(H)$.

By a \textit{bridge} we mean a cut~edge,
i.e., an edge of a connected (sub)graph whose removal results in two
connected~components, and by a \textit{cut~node} we mean a node of a
connected (sub)graph whose deletion results in
two or more connected~components.
We call a bridge of cost~zero a \textit{zero-bridge} and
we call a bridge of cost~one a \textit{unit-bridge}.

By a \textit{2ec-block} we mean a
maximal connected subgraph with two or more nodes that has no bridges.
(Observe that each 2ec-block of a graph $H$ corresponds to a
connected~component of order $\ge2$ of the graph obtained from $H$
by deleting all bridges.)
We call a 2ec-block \textit{pendant} if it is incident to
exactly one bridge.

The next result characterizes edges that are not essential for 2-edge~connectivity.

\begin{proposition} \label{propo:2ecdiscard}
Let $H$ be a 2EC graph and let $e=vw$ be an edge of $H$.
If $H-e$ has two edge-disjoint $v,w$~paths, then $H-e$ is 2EC.
\end{proposition}

By a \textit{2-edge~cover} (of $G$) we mean
a set of edges $F$ of $G$ such that
each node $v$ is incident to at least two edges of $F$
(i.e., $F\subseteq E(G): |\delta_F(v)|\ge2, \forall v\in{V(G)}$).
By $\DTWO(G)$ we mean any minimum-cost 2-edge~cover of $G$
($G$ may have several minimum-cost 2-edge~covers, and $\DTWO(G)$
may refer to any one of them);
we use $\dtwocost(G)$ to denote the cost of $\DTWO(G)$;
when there is no danger of ambiguity, we use \DTWO\ rather than $\DTWO(G)$,
and we use $\dtwocost$ rather than $\dtwocost(G)$.
Note that \DTWO\ may have several connected~components, and
each may have one or more bridges;
moreover, if a connected~component of \DTWO\ has a bridge,
then it has two or more pendant 2ec-blocks.

The next result follows from Theorem~34.15 in \cite[Chapter~34]{Schrijver}.

\begin{proposition} \label{thm:computeD2}
There is a polynomial-time algorithm for computing \DTWO.
\end{proposition}

The next result states the key lower~bound used by our approximation algorithm.

\begin{fact} \label{fact:dtwolb}
Let $H$ be any 2EC graph. Then we have $opt(H) \geq \dtwocost(H)$.
\end{fact}

By a \textit{bridgeless 2-edge~cover} (of $G$) we mean a 2-edge~cover
(of $G$) that has no bridges; note that we have no requirements on
the cost of a bridgeless 2-edge~cover.  We mention that the problem
of computing a bridgeless 2-edge~cover of minimum cost is NP-hard
(there is a reduction from TAP), and there is no approximation
algorithm known for the case of nonnegative costs.

\subsection{Ear decompositions}

An \textit{ear decomposition} of a graph 
is a partition of the edge set into paths or cycles,
$P_0,P_1,\dots,P_k$, such that
$P_0$ is the trivial path with one node,
and each $P_i$ ($1\leq i\leq k$) is either
(1)~a path that has both end nodes in
$V_{i-1} = V(P_0) \cup V(P_1) \cup \ldots \cup V(P_{i-1})$
but has no internal nodes in $V_{i-1}$, or
(2)~a cycle that has exactly one node in $V_{i-1}$.
Each of $P_1,\ldots,P_k$ is called an \textit{ear};
note that $P_0$ is not regarded as an ear.
We call $P_i, i\in\{1,\dots,k\},$ an \textit{open ear} if it is a path,
and we call it a \textit{closed ear} if it is a cycle.
An {\it open} ear decomposition $P_0,P_1,\ldots,P_k$
is one such that all the ears $P_2,\ldots,P_k$ are open.
(The ear $P_1$ is always closed.)

\begin{proposition}[Whitney \cite{W32}]
\label{propo:whitney}
\begin{itemize}
\item[(i)]
A graph is 2EC iff it has an ear decomposition.
\item[(ii)]
A graph is 2NC iff it has an open ear decomposition.
\end{itemize}
\end{proposition}

\subsection{Redundant 4-cycles}

By a \textit{redundant 4-cycle} we mean a cycle $C$ consisting of
four nodes and four edges of $G$ such that
$V(C)\not=V(G)$,
two of the (non-adjacent) edges of $C$ have cost zero, and
two of the nonadjacent nodes of $C$
have degree two in $G$.
For example, a 4-cycle $C=u_1,u_2,u_3,u_4,u_1$ with zero-edges
${u_1u_2}$, ${u_3u_4}$ and unit-edges ${u_2u_3}$, ${u_4u_1}$ of a 2EC
graph $G$ is a redundant 4-cycle provided all edges between $V(C)$ and
$V(G)-V(C)$ are incident to either $u_1$ or $u_3$.

Observe that every 2-edge~cover of $G$ must contain the edges of
every redundant 4-cycle.
Also, observe that two different redundant 4-cycles are disjoint
(i.e., any node is contained in at most one redundant 4-cycle).

\subsection{Bad-pairs and bp-components}

For any MAP~instance (assumed to be 2EC), we define
a \textit{bad-pair} to be a pair of nodes $\{v,w\}$ of the
graph such that the edge ${vw}$ is present and has zero cost, and
moreover, the deletion of both nodes $v$~and~$w$ results in a
disconnected graph.
Throughout, unless mentioned otherwise,
the term bad-pair refers to a bad-pair of the graph $G$
\big(in Section~\ref{s:bad-pairs},
we discuss bad-pairs of some minors (i.e., subgraphs) of~$G$\big).

{
By a \textit{bp-component} we mean one of the connected~components
resulting from the deletion of a bad-pair from $G$;
see Figure~\ref{f:sec2:badpairs}.
The set of all bp-components (of all bad-pairs) is a cross-free family,
see \cite[Chap.~13.4]{Schrijver}.
Consider two bad-pairs $\{v,w\}$ and $\{y,z\}$.
One of the bp-components of $\{v,w\}$ contains $\{y,z\}$; call it
$C_1$.  Then it can be seen that all-but-one of the bp-components
of $\{y,z\}$ are contained in $C_1$; the one remaining bp-component
of $\{y,z\}$, call it $C_1'$, contains $\{v,w\}$ and all of the
bp-components of $\{v,w\}$ except $C_1$. Thus, the union of the two
bp-components $C_1$ and $C_1'$ contains $V(G)$.
}

The following fact is analogous to the fact that
every tree on $\ge2$ nodes contains a node $v$
such that all-but-one of the neighbors of $v$ are leaves.
(To see this, consider a longest path $P$ of a tree, and take
$v$ to be the second node of $P$.)

\begin{fact} \label{fact:pickbp}
Suppose that $G$ has at least one bad-pair.
Then there exists a bad-pair such that
all-but-one of its bp-components are free of bad-pairs.
\end{fact}

\authremark{
Let us sketch an algorithmic proof of Fact~\ref{fact:pickbp}.
For any subgraph $G'$ of $G$,
let $\nbp(G')$ denote the number of bad-pairs (of $G$) that are contained in $V(G')$.
For any bad-pair $\{v,w\}$ and its bp-components $C_1,\dots,C_k$,
let us assume that the indexing is in non-increasing order of
$\nbp(C_i)$, thus, we have $\nbp(C_1)\ge\nbp(C_2)\ge\dots\ge\nbp(C_k)$.

We start by computing all the bad-pairs of $G$.
Then we choose any bad-pair $\{v_{1},w_{1}\}$ and compute its bp-components
$C^{(1)}_1,\dots,C^{(1)}_{k_1}$.
If $\nbp(C^{(1)}_2)=0$, then we are done; our algorithm outputs
$\{v_{1},w_{1}\}$.  Otherwise, we pick $C^{(1)}_2$ and choose any bad
pair $\{v_{2},w_{2}\}$ contained in $V(C^{(1)}_2)$.  We compute the
bp-components $C^{(2)}_1,\dots,C^{(2)}_{k_2}$ of $\{v_{2},w_{2}\}$ (in $G$).
As stated above, $C^{(2)}_1$ contains $C^{(1)}_1$ and $V(C^{(2)}_1)$ contains
$\{v_{1},w_{1}\}$;
moreover, $\nbp(C^{(1)}_2) > \nbp(C^{(2)}_2)+\dots+\nbp(C^{(2)}_{k_2})$,
since the bad-pair $\{v_{2},w_{2}\}$ is contained in $V(C^{(1)}_2)$ but $\{v_{2},w_{2}\}$ is disjoint
from each of its own bp-components.
We iterate these steps
(i.e., if $\nbp(C^{(2)}_2)\not=0$, then we pick any bad-pair
$\{v_{3},w_{3}\}$ contained in $V(C^{(2)}_2)$, \dots), until we find a bad-pair
$\{v_{\ell},w_{\ell}\}$ such that $\nbp(C^{(\ell)}_2)=0$; then we are done.
Clearly, this algorithm is correct, and it terminates in $O(n)$ iterations.
}

{
\begin{figure}[hbt]
    \begin{subfigure}{0.3\textwidth}
        \centering
        \begin{tikzpicture}[scale=0.75]
            \begin{scope}[every node/.style={circle, fill=black, draw, inner sep=0pt,
            minimum size = 0.15cm
            }]
                
                \node[label={[label distance=2]90:$y$}] (a2) at (0.7071,0.7071) {};
                \node[] (a3) at (0,1) {};
                \node[] (a4) at (-0.7071,0.7071) {};
                \node[] (a5) at (-1,0) {};
                \node[label={[label distance=2]180:$v$}] (a6) at (-0.7071,-0.7071) {};
                \node[label={[label distance=2]330:$w$}] (a7) at (0,-1) {};
                \node[label={[label distance=2]160:$z$}] (a8) at (0.7071,-0.7071) {};
                
                \node[] (b1) at (-1.4,-1.25) {};
                \node[] (b2) at (-0.8,-2.2) {};
                \node[] (b3) at (0.2,-1.75) {};
                
                \node[] (c1) at (1.75,0.25) {};
                \node[] (c2) at (1.4571,0.9571) {};
                \node[] (c3) at (2,1.75) {};
                
                \node[] (d1) at (3,0.25) {};
                \node[] (d2) at (3,1.25) {};
                \node[] (d3) at (4,0.75) {};
                
                \node[] (e1) at (2,-1) {};
                \node[] (e2) at (2.75,-1.75) {};
                \node[] (e3) at (2,-2.5) {};
                \node[] (e4) at (1.25,-1.75) {};
                
            \end{scope}
    
            \begin{scope}[every edge/.style={draw=black}]
                \path[thick, dashed] (a8) edge node {} (a2);
                \path (a2) edge node {} (a3);
                \path (a3) edge node {} (a4);
                \path (a4) edge node {} (a5);
                \path (a5) edge node {} (a6);
                \path[thick, dashed] (a6) edge node {} (a7);
                \path (a7) edge node {} (a8);
                
                \path (a6) edge node {} (b1);
                \path (b1) edge node {} (b2);
                \path (b2) edge node {} (b3);
                \path (b3) edge node {} (a7);
                
                \path (c1) edge node {} (c2);
                \path (c2) edge node {} (c3);
                \path[thick, dashed] (c1) edge node {} (c3);
                \path (a8) edge node {} (c1);
                \path (a2) edge node {} (c2);
                \path (a2) edge node {} (c3);
                
                \path (c1) edge node {} (d1);
                \path (c3) edge node {} (d2);
                \path (d1) edge node {} (d2);
                \path (d2) edge node {} (d3);
                \path (d3) edge node {} (d1);
                
                \path (a2) edge node {} (e1);
                \path (a8) edge node {} (e4);
                \path (e1) edge node {} (e2);
                \path (e2) edge node {} (e3);
                \path (e3) edge node {} (e4);
                \path (e4) edge node {} (e1);
            \end{scope}
    
            \begin{scope}[every node/.style={draw=none,rectangle}]
            \end{scope}
        \end{tikzpicture}
        \caption{\centering
		A graph with bad-pairs (indicated by dashed lines).}
        \label{sec2sub-a}
    \end{subfigure}
    \hspace*{\fill}
    \begin{subfigure}{0.3\textwidth}
        \centering
        \begin{tikzpicture}[scale=0.75]
            \begin{scope}[every node/.style={circle, fill=black, draw, inner sep=0pt,
            minimum size = 0.15cm
            }]
                
                \node[label={[label distance=2]90:$y$}] (a2) at (0.7071,0.7071) {};
                \node[] (a3) at (0,1) {};
                \node[] (a4) at (-0.7071,0.7071) {};
                \node[] (a5) at (-1,0) {};
                \node[label={[label distance=2]160:$z$}] (a8) at (0.7071,-0.7071) {};
                
                \node[] (b1) at (-1.4,-1.25) {};
                \node[] (b2) at (-0.8,-2.2) {};
                \node[] (b3) at (0.2,-1.75) {};
                
                \node[] (c1) at (1.75,0.25) {};
                \node[] (c2) at (1.4571,0.9571) {};
                \node[] (c3) at (2,1.75) {};
                
                \node[] (d1) at (3,0.25) {};
                \node[] (d2) at (3,1.25) {};
                \node[] (d3) at (4,0.75) {};
                
                \node[] (e1) at (2,-1) {};
                \node[] (e2) at (2.75,-1.75) {};
                \node[] (e3) at (2,-2.5) {};
                \node[] (e4) at (1.25,-1.75) {};
                
            \end{scope}
    
            \begin{scope}[every edge/.style={draw=black}]
                \path[thick, dashed] (a8) edge node {} (a2);
                \path (a2) edge node {} (a3);
                \path (a3) edge node {} (a4);
                \path (a4) edge node {} (a5);
                
                \path (b1) edge node {} (b2);
                \path (b2) edge node {} (b3);
                
                \path (c1) edge node {} (c2);
                \path (c2) edge node {} (c3);
                \path[thick, dashed] (c1) edge node {} (c3);
                \path (a8) edge node {} (c1);
                \path (a2) edge node {} (c2);
                \path (a2) edge node {} (c3);
                
                \path (c1) edge node {} (d1);
                \path (c3) edge node {} (d2);
                \path (d1) edge node {} (d2);
                \path (d2) edge node {} (d3);
                \path (d3) edge node {} (d1);
                
                \path (a2) edge node {} (e1);
                \path (a8) edge node {} (e4);
                \path (e1) edge node {} (e2);
                \path (e2) edge node {} (e3);
                \path (e3) edge node {} (e4);
                \path (e4) edge node {} (e1);
            \end{scope}
    
            \begin{scope}[every node/.style={draw=none,rectangle}]
            \end{scope}
        \end{tikzpicture}
        \caption{\centering
		bp-components for the bad-pair $\{v,w\}$.}
        \label{sec2sub-b}
    \end{subfigure}
    \hspace*{\fill}
    \begin{subfigure}{0.3\textwidth}
        \centering
        \begin{tikzpicture}[scale=0.75]
            \begin{scope}[every node/.style={circle, fill=black, draw, inner sep=0pt,
            minimum size = 0.15cm
            }]
                
                \node[] (a3) at (0,1) {};
                \node[] (a4) at (-0.7071,0.7071) {};
                \node[] (a5) at (-1,0) {};
                \node[label={[label distance=2]180:$v$}] (a6) at (-0.7071,-0.7071) {};
                \node[label={[label distance=2]330:$w$}] (a7) at (0,-1) {};
                
                \node[] (b1) at (-1.4,-1.25) {};
                \node[] (b2) at (-0.8,-2.2) {};
                \node[] (b3) at (0.2,-1.75) {};
                
                \node[] (c1) at (1.75,0.25) {};
                \node[] (c2) at (1.4571,0.9571) {};
                \node[] (c3) at (2,1.75) {};
                
                \node[] (d1) at (3,0.25) {};
                \node[] (d2) at (3,1.25) {};
                \node[] (d3) at (4,0.75) {};
                
                \node[] (e1) at (2,-1) {};
                \node[] (e2) at (2.75,-1.75) {};
                \node[] (e3) at (2,-2.5) {};
                \node[] (e4) at (1.25,-1.75) {};
                
            \end{scope}
    
            \begin{scope}[every edge/.style={draw=black}]
                \path (a3) edge node {} (a4);
                \path (a4) edge node {} (a5);
                \path (a5) edge node {} (a6);
                \path[thick, dashed] (a6) edge node {} (a7);
                
                \path (a6) edge node {} (b1);
                \path (b1) edge node {} (b2);
                \path (b2) edge node {} (b3);
                \path (b3) edge node {} (a7);
                
                \path (c1) edge node {} (c2);
                \path (c2) edge node {} (c3);
                \path[thick, dashed] (c1) edge node {} (c3);
                
                \path (c1) edge node {} (d1);
                \path (c3) edge node {} (d2);
                \path (d1) edge node {} (d2);
                \path (d2) edge node {} (d3);
                \path (d3) edge node {} (d1);
                
                \path (e1) edge node {} (e2);
                \path (e2) edge node {} (e3);
                \path (e3) edge node {} (e4);
                \path (e4) edge node {} (e1);
            \end{scope}
    
            \begin{scope}[every node/.style={draw=none,rectangle}]
            \end{scope}
        \end{tikzpicture}
        \caption{\centering
		bp-components for the bad-pair $\{y,z\}$.}
    \end{subfigure}
    \caption{Illustration of bad-pairs and bp-components.}
    \label{f:sec2:badpairs}
\end{figure}
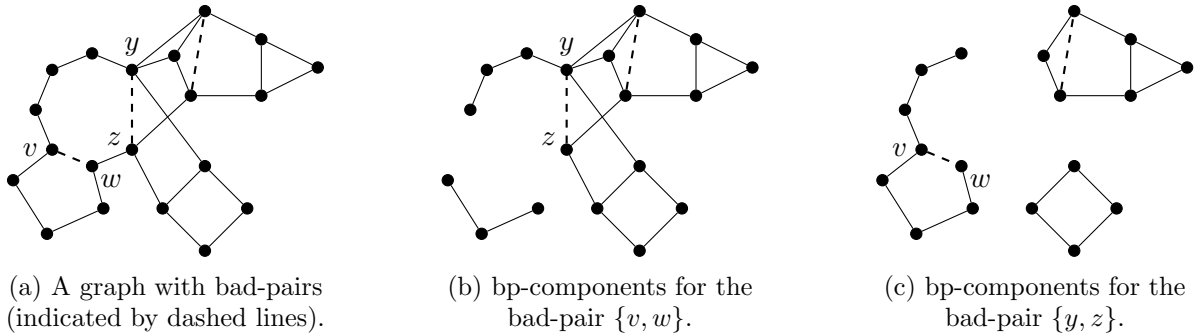
}

\subsection{Polynomial-time computations}

All of the computations in this paper can be easily implemented
in polynomial~time, see \cite{Schrijver}.
We state this explicitly in all relevant results (e.g.,
Proposition~\ref{thm:computeD2}, Theorem~\ref{thm:approxbydtwo}),
but we do not elaborate on this elsewhere.

\section{ \label{s:algo} Outline of the algorithm}

This section has an outline of our algorithm.
We start by defining a well structured MAP instance.

\begin{defn}
An instance of MAP is called \emph{well-structured} if it has
 \begin{itemize}\itemsep0pt
 \item[-] no $\{0,1\}$-edge-pairs,
 \item[-] no redundant 4-cycles,
 \item[-] no cut nodes, and
 \item[-] no bad-pairs.
 \end{itemize}
\end{defn}

Section~\ref{s:pre-proc1} explains
how to ``decompose'' any instance of MAP $G$ into
a collection of well-structured MAP instances $G_1,\dots,G_k$
such that a 2-ECSS $H$ of $G$ can be obtained
by computing 2-ECSSes $H_1,\dots,H_k$ of $G_1,\dots,G_k$,
and moreover,
the approximation guarantee is preserved, i.e.,
the approximation guarantee on $G$ 
is $\leq$ the maximum of the approximation guarantees on $G_1,\dots,G_k$
(in other words, $\ds {\frac{\cost(H)} {\opt(G)}} \leq
	\max_{i=1,\dots,k} \Big\{ {\frac{\cost(H_i)} {\opt(G_i)}} \Big\}$).

\medskip
\noi
\fbox{ \begin{minipage}{\textwidth}

\textbf{Algorithm (outline)}:
\begin{itemize}
\item[(0)] apply the pre-processing steps (reductions) from
	Section~\ref{s:pre-proc1} to obtain a collection of
	well-structured MAP instances {$G_1,\dots,G_k$};

\item[] for each $G_i$ ($i=1,\dots,k$), apply steps (1),(2),(3):

\item[(1)] compute \DTWO($G_i$) in polynomial time
	(w.l.o.g.\ assume \DTWO($G_i$) contains all zero-edges of~$G_i$);

\item[(2)] then apply ``bridge~covering" from Section~\ref{s:path-thicken}
	to \DTWO($G_i$) to obtain a bridgeless 2-edge~cover $\wt{H}_i$
	of $G_i$;

\item[(3)] then apply the ``gluing~step" from Section~\ref{s:algo-last}
	to $\wt{H}_i$ to obtain a 2-ECSS $H_i$ of $G_i$;

\item[(4)] finally, output a 2-ECSS $H$ of $G$ from the union of
	$H_1,\dots,H_k$ by undoing the transformations applied in step~(0).
\end{itemize}
\end{minipage}
}
\medskip

The pre-processing of step~(0) consists of four reductions:
\begin{itemize}
\item[(pp1)]~handle $\{0,1\}$-edge-pairs,
\item[(pp2)]~handle redundant~4-cycles,
\item[(pp3)]~handle cut~nodes, and
\item[(pp4)]~handle bad-pairs;
\end{itemize}
these reductions are discussed in Sections~\ref{s:pre-proc1-1}--\ref{s:bad-pairs}.
Step~(0) applies~(pp1) to obtain a collection of MAP~instances;
after that, there is no further need to apply~(pp1), see Fact~\ref{fact:allowed-contracts}.
Then, we iterate:
while the collection of MAP~instances
has one or more of the latter three ``obstructions''
(redundant~4-cycles, cut~nodes, bad-pairs),
we apply~(pp2),~(pp3),~(pp4) in sequence.
After $\le n$ iterations, we have a collection of
well-structured MAP~instances $G_i$.
Then, we compute an approximately~optimal 2-ECSS $H_i$ for each $G_i$ using the
algorithm of Theorem~\ref{thm:approxbydtwo} (below), and finally, we
use the $H_i$ subgraphs to construct an approximately~optimal 2-ECSS of $G$.

Our $\frac74$ approximation algorithm for MAP follows from 
the following key theorem, and the fact that the algorithm
runs in polynomial time.

\begin{theorem}
\label{thm:approxbydtwo}
Given a well-structured instance of MAP $G'$,
  there is a polynomial-time algorithm that
  obtains a 2-ECSS $H'$ from $\DTWO(G')$
  (by adding edges and deleting edges)
  such that $\cost(H') \leq \frac{7}{4} \dtwocost(G')$.
\end{theorem}

We use a credit~scheme to prove this theorem; the details are
presented in Sections~\ref{s:path-thicken} and~\ref{s:algo-last}.
The algorithm starts with $\DTWO=\DTWO(G')$ as the current subgraph;
we start by assigning $\frac74$ initial~credits to each unit-edge of
\DTWO; each such edge keeps one credit to pay for itself and the other
$\frac34$~credits are taken to be working~credits available to the
algorithm; the algorithm uses these working~credits to pay for the
augmenting edges ``bought" in steps~(2) or~(3) (see the outline); also,
the algorithm may ``sell" unit-edges of the current~subgraph
(i.e., such an edge is permanently discarded and is not contained in
the 2-ECSS output by the algorithm) and this supplies working~credits
to the algorithm (see Sections~\ref{s:path-thicken},~\ref{s:algo-last}).

In an ear-augmentation step, we may add either a single~ear
or a double~ear (i.e., a pair of ears);
see Section~\ref{s:path-thicken}
(double~ears may be added in Case~3, page~\pageref{page:bc-case3})
and Section~\ref{s:algo-last}
(double~ears may be added in Case~2, page~\pageref{page:gluing-case2}).
Although this is not directly relevant, we mention that double~ear
augmentations are essential in matching theory, see \cite[Ch.5.4]{LP09:book}.
As discussed above, in some of the ear-augmentation steps, we may
(permanently) delete some edges from the current subgraph; see
Section~\ref{s:path-thicken} (edges are deleted in double~ear
augmentations in Case~3, page~\pageref{page:bc-case3})
and Section~\ref{s:algo-last} (edges are deleted in both Cases~1,~2).

\authremark{
	The following examples show that when we relax the definition
	of a well-structured MAP~instance, then the inequality in
	Theorem~\ref{thm:approxbydtwo} could fail to hold.
	See Figure~\ref{f:sec3:thm6egs} for illustrations.
}

{
\begin{enumerate}
	\item[(a)]~$\{0,1\}$-edge-pairs (i.e., parallel edges of cost zero and
	one) are present.  Then ${\frac {\opt} {\dtwocost}} \approx 2$ is
	possible.  Our construction consists of a root 2ec-block $B_0$,
	say a 6-cycle of cost~6, and $\ell\gg1$ copies of the following
	gadget that are attached to $B_0$.  The gadget consists of a
	pair of nodes $v,w$ and two  incident edges:  a copy of edge
	${vw}$ of cost~zero, and a copy of edge ${vw}$ of cost~one.
	Moreover, we have an edge between $v$ and $B_0$ of cost~one,
	and an edge between $w$ and $B_0$ of cost~one.
	Observe that a (feasible) 2-edge~cover of this instance
	consists of $B_0$ and the two parallel edges (i.e., the two
	copies of the edge $vw$) of each copy of the gadget, and it has
	cost $6+\ell$.
	Observe that any 2-ECSS contains the two edges between
	$\{v,w\}$ and $B_0$.  Thus, $\opt\geq2\ell$, whereas
	$\dtwocost\leq6+\ell$.
	\item[(b)]~Redundant 4-cycles are present.  Then ${\frac {\opt}
	{\dtwocost}} \approx 2$ is possible.  Our construction consists
	of a root 2ec-block $B_0$, say a 6-cycle of cost~6, and
	$\ell\gg1$ copies of the following gadget that are attached to
	$B_0$.  The gadget consists of a 4-cycle $C=u_1,\dots,u_4,u_1$
	that has two zero-edges ${u_1u_2},{u_3u_4}$ and two unit-edges
	${u_2u_3},{u_4u_1}$; moreover, we have an edge between $u_1$
	and $B_0$ of cost~one, and an edge between $u_3$ and $B_0$ of
	cost~one.
	Observe that a (feasible) 2-edge~cover of this instance
	consists of $B_0$ and the 4-cycle $C$ of each copy of the
	gadget, and it has cost $6+2\ell$.
	Observe that for any 2-ECSS and for each copy of the gadget,
	the two edges between $C$ and $B_0$ as well as the four edges
	of $C$ are contained in the 2-ECSS.
	Thus, $\opt\geq4\ell$, whereas $\dtwocost\leq6+2\ell$.
	\item[(c)]~Cut~nodes are present.  Then ${\frac {\opt} {\dtwocost}}
	\approx 2$ is possible.  Our construction consists of $\ell$
	copies of a 3-cycle $C=u_1,u_2,u_3,u_1$ where ${u_1u_2}$
	is a zero~edge and the other two edges have cost~one.  We
	``string up'' the $\ell$ copies, i.e., node $u_3$ of the $i$th
	copy is identified with node $u_1$ of the $(i+1)$th copy.
	The optimal solution has all the edges, so $\opt=2\ell$,
	whereas a (feasible) 2-edge~cover consists of a Hamiltonian
	path together with two more edges incident to the two ends of
	this path, and it has cost $\ell+2$; thus
	$\dtwocost\leq2+\ell$.
	\item[(d)]~Bad-pairs are present.  Then ${\frac {\opt} {\dtwocost}}
	\approx 2$ is possible.  An example can be obtained by
	modifying example~(b) above.
	\big(Each copy of the gadget is connected to $B_0$ by two edges
	that are incident to the nodes $u_1$ and $u_2$ instead of the
	nodes $u_1$ and $u_3$. It can be seen that $\opt\ge4\ell$
	whereas $\dtwocost\leq6+2\ell$.\big)
\end{enumerate}
}

{
\begin{figure}[ht]
    \centering
    \begin{subfigure}{0.45\textwidth}
        \centering
        \begin{tikzpicture}[scale=0.75]
            \begin{scope}[every node/.style={circle, fill=black, draw, inner sep=0pt,
            minimum size = 0.15cm
            }]
                \draw[very thick,dotted] (0,2) circle (0.8cm);
                \node[label={[label distance=2]270:$v_1$}] (v1) at (-3,0) {};
                \node[label={[label distance=2]270:$w_1$}] (w1) at (-2,0) {};
                \node[label={[label distance=2]270:$v_2$}] (v2) at (-1,0) {};
                \node[label={[label distance=2]270:$w_2$}] (w2) at (0,0) {};
                \node[label={[label distance=2]270:$v_\ell$}] (vl) at (2,0) {};
                \node[label={[label distance=2]270:$w_\ell$}] (wl) at (3,0) {};
                \node[draw=none,fill=none] (lt) at (-0.3,2) {};
                \node[draw=none,fill=none] (rt) at (0.3,2) {};

                \node[draw=none,fill=none] (w22) at (0.25,0) {};
                \node[draw=none,fill=none] (vll) at (1.75,0) {};
                
            \end{scope}
    
            \begin{scope}[every edge/.style={draw=black}]
                
                \path[dashed] (v1) edge[bend left=30] node {} (w1);
                \path[] (v1) edge[bend right=30] node {} (w1);
                \path[] (v1) edge[bend left=15] node {} (lt);
                \path[] (w1) edge[bend right=15] node {} (rt);
                
                \path[dashed] (v2) edge[bend left=30] node {} (w2);
                \path[] (v2) edge[bend right=30] node {} (w2);
                \path[] (v2) edge[bend left=15] node {} (lt);
                \path[] (w2) edge[bend right=15] node {} (rt);
                
                \path[dashed] (vl) edge[bend left=30] node {} (wl);
                \path[] (vl) edge[bend right=30] node {} (wl);
                \path[] (vl) edge[bend left=15] node {} (lt);
                \path[] (wl) edge[bend right=15] node {} (rt);
                
                \path[thick,dotted] (w22) edge node {} (vll);
                
            \end{scope}
    
            \begin{scope}[every node/.style={draw=none,rectangle}]
                \node (B0label) at (0,2.25) {$B_0$};
            \end{scope}
        \end{tikzpicture}
        \caption{\centering Example~(a) of Remark.}
        \label{f:sec3:placeholder1-a}
    \end{subfigure}
    \hspace*{\fill}
    \begin{subfigure}{0.45\textwidth}
        \centering
        \begin{tikzpicture}[scale=0.75]
            \begin{scope}[every node/.style={circle, fill=black, draw, inner sep=0pt,
            minimum size = 0.15cm
            }]
                \draw[very thick,dotted] (0,2) circle (0.8cm);
                \node[label={[label distance=2]270:$u_1$}] (u1) at (-4,0) {};
                \node[label={[label distance=2]270:$u_2$}] (u2) at (-3.25,0.75) {};
                \node[label={[label distance=2]270:$u_3$}] (u3) at (-2.5,0) {};
                \node[label={[label distance=2]270:$u_4$}] (u4) at (-3.25,-0.75) {};
                \node (u21) at (-1.5,0) {};
                \node (u22) at (-0.75,0.75) {};
                \node (u23) at (0,0) {};
                \node (u24) at (-0.75,-0.75) {};
                \node[label={[label distance=1]225:$u_1^{(\ell)}$}] (ul1) at (2,0) {};
                \node[label={[label distance=2]45:$u_2^{(\ell)}$}] (ul2) at (2.75,0.75) {};
                \node[label={[label distance=1]0:$u_3^{(\ell)}$}] (ul3) at (3.5,0) {};
                \node[label={[label distance=1]270:$u_4^{(\ell)}$}] (ul4) at (2.75,-0.75) {};
                \node[draw=none,fill=none] (lt) at (-0.3,2) {};
                \node[draw=none,fill=none] (rt) at (0.3,2) {};

                \node[draw=none,fill=none] (w22) at (0.25,0) {};
                \node[draw=none,fill=none] (vll) at (1.75,0) {};
                
            \end{scope}
    
            \begin{scope}[every edge/.style={draw=black}]
                
                \path[dashed] (u1) edge node {} (u2);
                \path[] (u2) edge node {} (u3);
                \path[dashed] (u3) edge node {} (u4);
                \path[] (u4) edge node {} (u1);
                
                \path[dashed] (u21) edge node {} (u22);
                \path[] (u22) edge node {} (u23);
                \path[dashed] (u23) edge node {} (u24);
                \path[] (u24) edge node {} (u21);
                
                \path[dashed] (ul1) edge node {} (ul2);
                \path[] (ul2) edge node {} (ul3);
                \path[dashed] (ul3) edge node {} (ul4);
                \path[] (ul4) edge node {} (ul1);
                
                \path[] (u1) edge[bend left=30] node {} (lt);
                \path[] (u3) edge[bend right=15] node {} (rt);
                
                \path[] (u21) edge[bend left=15] node {} (lt);
                \path[] (u23) edge[bend right=15] node {} (rt);
                
                \path[] (ul1) edge[bend left=15] node {} (lt);
                \path[] (ul3) edge[bend right=30] node {} (rt);
                
                \path[thick,dotted] (w22) edge node {} (vll);
                
            \end{scope}
    
            \begin{scope}[every node/.style={draw=none,rectangle}]
                \node (B0label) at (0,2.25) {$B_0$};
            \end{scope}
        \end{tikzpicture}
        \caption{\centering Example~(b) of Remark.}
        \label{f:sec3:placeholder1-b}
    \end{subfigure}
    \hspace*{\fill}
    \begin{subfigure}{0.45\textwidth}
        \centering
        \begin{tikzpicture}[scale=0.75]
            \begin{scope}[every node/.style={circle, fill=black, draw, inner sep=0pt,
            minimum size = 0.15cm
            }]
                \node[label={[label distance=1]270:$u_1$}] (u1) at (0,0) {};
                \node[label={[label distance=1]90:$u_2$}] (u2) at (0.5,1) {};
                \node[label={[label distance=1]270:$u_3$}] (u3) at (1,0) {};
                \node[] (u22) at (1.5,1) {};
                \node[] (u23) at (2,0) {};
                \node[] (u32) at (2.5,1) {};
                \node[] (u33) at (3,0) {};
                
                \node[label={[label distance=0]225:$u_1^{(\ell)}$}] (ul1) at (5,0) {};
                \node[label={[label distance=0]90:$u_2^{(\ell)}$}] (ul2) at (5.5,1) {};
                \node[label={[label distance=0]315:$u_3^{(\ell)}$}] (ul3) at (6,0) {};
                
                \node[draw=none,fill=none] (ll) at (3.25,0) {};
                \node[draw=none,fill=none] (rr) at (4.75,0) {};
                
            \end{scope}
    
            \begin{scope}[every edge/.style={draw=black}]
                
                \path[dashed] (u1) edge node {} (u2);
                \path[] (u2) edge node {} (u3);
                \path[] (u1) edge node {} (u3);
                \path[dashed] (u3) edge node {} (u22);
                \path[] (u22) edge node {} (u23);
                \path[] (u3) edge node {} (u23);
                \path[dashed] (u23) edge node {} (u32);
                \path[] (u32) edge node {} (u33);
                \path[] (u23) edge node {} (u33);
                \path[dashed] (ul1) edge node {} (ul2);
                \path[] (ul2) edge node {} (ul3);
                \path[] (ul1) edge node {} (ul3);
                
                \path[thick,dotted] (ll) edge node {} (rr);
                
            \end{scope}
    
            \begin{scope}[every node/.style={draw=none,rectangle}]
            \end{scope}
        \end{tikzpicture}
        \caption{\centering Example~(c) of Remark.}
        \label{f:sec3:placeholder1-c}
    \end{subfigure}
    \hspace*{\fill}
    \begin{subfigure}{0.45\textwidth}
        \centering
        \begin{tikzpicture}[scale=0.75]
            \begin{scope}[every node/.style={circle, fill=black, draw, inner sep=0pt,
            minimum size = 0.15cm
            }]
                \draw[very thick,dotted] (0,1.5) circle (0.8cm);
                \node[label={[label distance=1]225:$u_1$}] (u1) at (-3,0) {};
                \node[label={[label distance=1]225:$u_2$}] (u2) at (-2,0) {};
                \node[label={[label distance=1]225:$u_3$}] (u3) at (-2,-1) {};
                \node[label={[label distance=1]225:$u_4$}] (u4) at (-3,-1) {};
                \node (u21) at (-1,0) {};
                \node (u22) at (0,0) {};
                \node (u23) at (0,-1) {};
                \node (u24) at (-1,-1) {};
                \node[label={[label distance=1]225:$u_1^{(\ell)}$}] (ul1) at (2,0) {};
                \node[label={[label distance=2]0:$u_2^{(\ell)}$}] (ul2) at (3,0) {};
                \node[label={[label distance=1]0:$u_3^{(\ell)}$}] (ul3) at (3,-1) {};
                \node[label={[label distance=1]200:$u_4^{(\ell)}$}] (ul4) at (2,-1) {};
                \node[draw=none,fill=none] (lt) at (-0.3,1.5) {};
                \node[draw=none,fill=none] (rt) at (0.3,1.5) {};

                \node[draw=none,fill=none] (w22) at (0.25,0) {};
                \node[draw=none,fill=none] (vll) at (1.75,0) {};
                
            \end{scope}
    
            \begin{scope}[every edge/.style={draw=black}]
                
                \path[dashed] (u1) edge node {} (u2);
                \path[] (u2) edge node {} (u3);
                \path[dashed] (u3) edge node {} (u4);
                \path[] (u4) edge node {} (u1);
                
                \path[dashed] (u21) edge node {} (u22);
                \path[] (u22) edge node {} (u23);
                \path[dashed] (u23) edge node {} (u24);
                \path[] (u24) edge node {} (u21);
                
                \path[dashed] (ul1) edge node {} (ul2);
                \path[] (ul2) edge node {} (ul3);
                \path[dashed] (ul3) edge node {} (ul4);
                \path[] (ul4) edge node {} (ul1);
                
                \path[] (u1) edge[bend left=30] node {} (lt);
                \path[] (u2) edge[bend right=15] node {} (rt);
                
                \path[] (u21) edge[bend left=15] node {} (lt);
                \path[] (u22) edge[bend right=15] node {} (rt);
                
                \path[] (ul1) edge[bend left=15] node {} (lt);
                \path[] (ul2) edge[bend right=30] node {} (rt);
                
                \path[thick,dotted] (w22) edge node {} (vll);
                
            \end{scope}
    
            \begin{scope}[every node/.style={draw=none,rectangle}]
                \node (B0label) at (0,1.8) {$B_0$};
            \end{scope}
        \end{tikzpicture}
        \caption{\centering Example~(d) of Remark.}
        \label{f:sec3:placeholder1-d}
    \end{subfigure}
    \caption{\protect\centering
MAP instances $G$ that are \textit{not} well-structured such that
${\frac {\opt(G)} {\dtwocost(G)}} \approx 2$.
Edges of cost~zero and~one are illustrated by dashed and solid lines, respectively.
	}
    \label{f:sec3:thm6egs}
\end{figure}
}

\section{ \label{s:pre-proc1} Pre-processing}
{
This section presents the four reductions
used in the pre-processing step of our algorithm, namely,
the handling of the $\{0,1\}$-edge-pairs, the redundant~4-cycles,
the cut~nodes, and the bad-pairs.

\subsection{ \label{s:pre-proc1-1} Handling $\{0,1\}$-edge-pairs}

We apply the following pre-processing step to
eliminate all $\{0,1\}$-edge-pairs.
We start with a simple result.

\begin{fact} \label{fact:del-edgepair}
Let $H$ be an (inclusion-wise) edge-minimal 2EC graph,
and let $e_1,e_2$ be a pair of parallel edges of $H$.
Then $H-\{e_1,e_2\}$ has precisely two connected~components,
and each of these connected~components is 2EC.
\end{fact}
\begin{proof}
Let $v$ and $w$ be the end~nodes of $e_1, e_2$.
Observe that $H-\{e_1,e_2\}$ has no $v,w$ path.
\big(Otherwise, $H$ would have three edge-disjoint $v,w$ paths; but then,
$H-e_1$ would have two edge-disjoint $v,w$ paths, and by
Proposition~\ref{propo:2ecdiscard}, $H-e_1$ would be 2EC; this would
contradict the edge-minimality of $H$.\big)
It follows that $H-\{e_1,e_2\}$ has precisely two connected~components
(deleting one edge from a connected graph
results in a graph with $\leq2$ connected~components).
Let $C_1,C_2$ be the two connected~components of $H-\{e_1,e_2\}$;
clearly, each of $C_1,C_2$ contains precisely one of the nodes $v,w$.
Suppose that $C_1$ is not 2EC;
then it has a bridge $f$.
Since the parallel edges $e_1,e_2$ have exactly one end~node in $C_1$,
$f$ stays a bridge of $C_1\cup{C_2}\cup\{e_1,e_2\}$.
This is a contradiction, since $H=C_1\cup{C_2}\cup\{e_1,e_2\}$ is 2EC.
\end{proof}

Let $H$ be any 2EC graph.
We call a $\{0,1\}$-edge-pair \textit{essential} if its deletion
results in a disconnected graph.  When we delete an essential
$\{0,1\}$-edge-pair then we get two connected~components and each
one is 2EC, by arguing as in the proof of Fact~\ref{fact:del-edgepair}.
Hence, when we delete all the essential $\{0,1\}$-edge-pairs,
then we get a number of connected~components $C_1,\dots,C_k$ such
that each one is 2EC.  Clearly, an approximately optimal 2-ECSS of
$H$ can be computed by returning the union of approximately optimal
2-ECSSes of $C_1,\dots,C_k$ together with all the essential
$\{0,1\}$-edge-pairs of $H$; moreover, it can be seen that the
approximation guarantee is preserved, that is, the approximation
guarantee on $H$ is $\leq$ the maximum of the approximation guarantees
on $C_1,\dots,C_k$.

The next observation allows us to handle the inessential $\{0,1\}$-edge-pairs.

\begin{fact} \label{fact:edgepair}
Suppose that $H$ is 2EC and it has no essential $\{0,1\}$-edge-pairs.
Then there exists a min-cost 2-ECSS of $H$ that does not contain
any unit-edge of any $\{0,1\}$-edge-pair.
\end{fact}
\begin{proof}
Consider a min-cost 2-ECSS $H'$ of $H$ that contains the
minimum number of $\{0,1\}$-edge-pairs,
i.e., among all the optimal subgraphs, we pick one
that has the fewest number of parallel edges $e,f$ such that
$\cost(e)+\cost(f)=1$.
If $H'$ has no $\{0,1\}$-edge-pair, then the fact holds.
Otherwise, we argue by contradiction.
We pick any $\{0,1\}$-edge-pair $e,f$.
Deleting both $e,f$ from $H'$ results in two connected~components $C_1,C_2$
such that each is 2EC, by Fact~\ref{fact:del-edgepair}.
Now, observe that $e,f$ is not essential for $H$,
hence, $H-\{e,f\}$ has an edge $e''$ between $C_1$ and $C_2$.
We obtain the graph $H''$ from $H'$ by
replacing the unit-edge of $e,f$ by the edge $e''$.
Clearly, $H''$ is a 2-ECSS of $H$ of cost $\le\cost(H')$, and
moreover, it has fewer $\{0,1\}$-edge-pairs than $H'$.
Thus, we have a contradiction.
\end{proof}

Now, focus on the input graph $G$.
By the discussion above,
we may assume that $G$ has no essential $\{0,1\}$-edge-pairs.
Then we delete the unit-edge of each $\{0,1\}$-edge-pair.
By Fact~\ref{fact:edgepair}, the resulting graph stays 2EC and
the optimal value is preserved.
Thus, we can eliminate all $\{0,1\}$-edge-pairs, while preserving
the approximation guarantee.

\begin{proposition}
Assume that $G$ has no essential $\{0,1\}$-edge-pairs.
Let $\hat{G}$ be the multi-graph obtained from $G$ by
eliminating all $\{0,1\}$-edge-pairs (as discussed above).
Then $\opt(G) = \opt(\hat{G})$.

There is a polynomial-time $\alpha$-approximation algorithm for MAP on $G$,
if there is such an algorithm for MAP on $\hat{G}$.
\end{proposition}

In what follows, we continue to use $G$ to denote the multi-graph
obtained by eliminating all $\{0,1\}$-edge-pairs (for the sake of
notational convenience).

The next fact states that the restriction on the zero-edges of $G$
is preserved when we contract a set of edges $E'$ such that none
of the end~nodes of the zero-edges in $E-E'$ is incident to an edge of $E'$
(i.e., the end~nodes of the zero-edges of $G/E'$ are
``original nodes'' rather than ``contracted nodes.'')

\begin{fact} \label{fact:allowed-contracts}
Let $G=(V,E)$ satisfy the restriction on the zero-edges
(i.e., $G$ has no $\{0,1\}$-edge-pairs, and
the zero-edges form a matching).
Suppose that we contract a set of edges $E' \subset E$ such that 
there exists no node that is incident to both an edge in $E'$
and a zero-edge in $E-E'$.
Then the restriction on the zero-edges
continues to hold for the 
contracted multi-graph $G/E'$.
\end{fact}

\subsection{ \label{s:pre-proc1-2} Handling redundant 4-cycles}

We contract all of the redundant 4-cycles in a pre-processing step.
Recall that two distinct redundant 4-cycles have
no nodes and no edges in common.
We first compute all the redundant~4-cycles 
and then we contract each of these cycles.

\begin{proposition}
Suppose that $G$ has $q$ redundant 4-cycles.
Let $\hat{G}$ be the multi-graph obtained from $G$ by
contracting all redundant 4-cycles.
Then $\opt(G) = \opt(\hat{G}) + 2q$.
There is a polynomial-time $\alpha$-approximation algorithm for MAP on~$G$,
if there is such an algorithm for MAP on~$\hat{G}$.
\end{proposition}

\authremark{
Note that the contraction of a redundant~4-cycle may result in ``new'' cut~nodes.
}

\subsection{ \label{s:pre-proc1-3} Handling cut~nodes}

Let $H$ be any 2EC graph. Then $H$ can be decomposed into blocks
$H_1,\dots,H_k$ such that each block is either 2NC or else it
consists of two nodes with two parallel edges between the two nodes.
(Thus, $E(H)$ is partitioned among $E(H_1),\dots,E(H_k)$ and
any two of the blocks $H_i, H_j$ are either disjoint or 
they have exactly one node in common.)
It is well known that an approximately optimal 2-ECSS of $H$
can be computed by taking the union of
approximately optimal 2-ECSSes of $H_1,\dots,H_k$;
moreover, the approximation guarantee is preserved,
see \cite[Proposition~1.4]{SV:cca}.

\begin{proposition}
Suppose that $G$ has cut~nodes; let $G_1,\dots,G_k$ be the blocks of $G$.
Then $\opt(G) = \sum_{i=1}^k \opt({G_i})$.
There is a polynomial-time $\alpha$-approximation algorithm for MAP on~$G$,
if there is such an algorithm for MAP on~$G_i$, $\forall{i}\in\{1,\dots,k\}$.
\end{proposition}

\authremark{
Note that the blocks $G_1,\dots,G_k$ of a 2EC graph $G$ may contain
redundant~4-cycles even if $G$ has no redundant~4-cycle.
}

\subsection{ \label{s:bad-pairs} Pre-processing for bad-pairs}
{
This sub-section presents the pre-processing step of our algorithm that
handles the bad-pairs.
This step partitions the edges of a 2NC~instance $G$ among a number of sub-instances
$G_i$ such that each sub-instance is 2NC and has no bad-pairs. 
We ensure the key property that the union of the 2-ECSSes of
the sub-instances $G_i$ forms a 2-ECSS of $G$ (see Fact~\ref{fact:2ECbyunion}).

Throughout this sub-section,
unless there is a statement to the contrary,
we assume that the instance $G$ has no $\{0,1\}$-edge-pairs.
This assumption is valid because the pre-processing step~(pp1) has
been applied already.

\subsubsection{Bad-pairs and bp-components}

Recall that a \textit{bad-pair} of a 2EC MAP~instance is a pair of nodes $\{v,w\}$
such that ${vw}$ is a zero-edge, and the deletion of both
nodes $v$~and~$w$ results in a disconnected graph.
Throughout, unless mentioned otherwise,
the term bad-pair refers to a bad-pair of the graph $G$.

For an instance $\wt{G}$ that has no cut~nodes and no bad-pairs,
we define $\ndtwocost(\wt{G})$ to be
$2q + \sum_{i=1}^{k}\dtwocost(\wt{G}_i)$,
where $q$ denotes the number of redundant 4-cycles of $\wt{G}$,
and $\wt{G}_1,\dots,\wt{G}_k$ denotes the 
collection of well-structured MAP~instances obtained from $\wt{G}$ by
repeatedly applying the pre-processing steps (pp2) and (pp3)
(observe that (pp2) and~(pp3) cannot introduce ``new'' bad-pairs nor ``new'' $\{0,1\}$-edge-pairs).

\begin{fact} \label{fact:sharp-lb}
Let $\wt{G}$ be an instance that has no cut~nodes and no bad-pairs
(by assumption $\wt{G}$ has no $\{0,1\}$-edge-pairs).
We have
$\ds \dtwocost(\wt{G}) \leq \ndtwocost(\wt{G}) \leq \opt(\wt{G})$.
Moreover, $\ndtwocost(\wt{G})$ can be computed in polynomial~time.
\end{fact}
\begin{proof}
Let $C_1,\dots,C_q$ denote the redundant~4-cycles of $\wt{G}$.
Let ${G'}$ denote the graph obtained from $\wt{G}$ by contracting $C_1,\dots,C_q$,
and let ${G'_1},\dots,{G'_k}$ denote the blocks of ${G'}$.
None of ${G'_1},\dots,{G'_k}$ contains a redundant~4-cycle
(since no cut~node of $G'$ is incident to a zero-edge), hence,
each of ${G'_1},\dots,{G'_k}$ is a well-structured MAP~instance.
Observe that a 2-edge~cover of $\wt{G}$ is given by the union of
$\bigcup_{i=1}^{k} E(\DTWO(G'_i))$ and $\bigcup_{j=1}^{q} E(C_j)$,
and we have $\cost(E(\DTWO(G'_i))) = \dtwocost(G'_i),~\forall\;i\in\{1,\dots,k\}$,
hence,
$\dtwocost(\wt{G}) \leq (\sum_{i=1}^{k} \dtwocost(G'_i)) + (2q) = \ndtwocost(\wt{G})$.
This proves the first inequality.

Let $\wt{G}^{opt}$ denote an arbitrary min-cost 2-ECSS of $\wt{G}$;
clearly, $\wt{G}^{opt}$ contains each of $C_1,\dots,C_q$.
Let $G^*$ be obtained from $\wt{G}^{opt}$ by contracting $C_1,\dots,C_q$.
Note that $E(G^*)$ can be partitioned into sets $E^*_1,\dots,E^*_k$
such that $E^*_i\subseteq E(G'_i)$ and $E^*_i$ induces a 2EC~subgraph
of ${G'_i}$, for~all $i\in\{1,\dots,k\}$.
Hence, $\cost(G^*)= \sum_{i=1}^{k}\cost(E^*_i) \ge \sum_{i=1}^{k}\dtwocost(G'_i)$.
Therefore, $\opt(\wt{G}) = \cost(\wt{G}^{opt}) = 2q+\cost(G^*) \ge
2q+\sum_{i=1}^{k}\dtwocost(G'_i) = \ndtwocost(\wt{G})$.
This proves the second inequality.
\end{proof}

The pre-processing algorithm and analysis of this subsection rely
on using the sharper lower~bound of $\ndtwocost$ on $\opt$.

For a bad-pair $\{v,w\}$ and one of its bp-components $C$
we use $\ubpcomp{C}{v}{w}$ to denote
the subgraph of $G$ induced by $V(C)\cup\{v,w\}$;
thus, we have $\ubpcomp{C}{v}{w} \;=\; G[V(C)\cup\{v,w\}]$;
moreover, if $C$ has $\ge2$ nodes,
then we use $\sbpcomp{C}{v}{w}$ to denote the multi-graph obtained
from $\ubpcomp{C}{v}{w}$ by contracting the zero-edge ${vw}$,
whereas if $C$ has only one node then
we take $\sbpcomp{C}{v}{w}$ to be the same as $\ubpcomp{C}{v}{w}$
(to ensure that $\sbpcomp{C}{v}{w}$ is 2NC, see
Fact~\ref{fact:2NCparts}, we have to ensure that it has $\ge3$ nodes;
if $C$ has only one node, then note that $\ubpcomp{C}{v}{w}/\{vw\}$ has
only 2~nodes).

We sketch our plan for handling the bad-pairs in this informal and
optional paragraph.
Assume that $G$ has one or more bad-pairs.
We traverse the ``tree of bp-components and bad-pairs'';
at each iteration, we pick a bad-pair 
$\{v_{\ell},w_{\ell}\}$ such that
all-but-one of its bp-components are free of bad-pairs,
see Fact~\ref{fact:pickbp}.
Let $C_1$ denote the (unique) bp-component that has one or more bad-pairs
(w.l.o.g.\ assume $C_1$ exists), and
let $C_2,\dots,C_k$ denote the bp-components free of bad-pairs;
we call these the \textit{leaf}~bp-components.
It is easily seen that for a bp-component $C_i$, the subgraph of any
optimal solution induced by $V(C_i)\cup\{v_{\ell},w_{\ell}\}$ has cost
$\ge \dtwocost(\sbpcomp{C_i}{v_{\ell}}{w_{\ell}})$
(see Fact~\ref{fact:piecebound}).
Now, focus on an optimal solution and let $F^*$ denote its set of edges,
i.e., $(V,F^*)$ is a 2-ECSS of $G$ of minimum cost.
Since $(V,F^*)$ is 2EC, it can be seen that there exists a
$j\in\{1,\dots,k\}$ such that the graph
$(\ubpcomp{C_j}{v_{\ell}}{w_{\ell}} - {v_{\ell}w_{\ell}})$ contains
a $v_{\ell},w_{\ell}$ path;
then, it follows that $F^*\cup\{v_{\ell} w_{\ell}\}$ induces a 2-ECSS
of $\ubpcomp{C_j}{v_{\ell}}{w_{\ell}}$. In other words, we may
assume w.l.o.g.\ that the zero-edge ${v_{\ell}w_{\ell}}$ (of the
bad-pair) is in $F^*$, and we may ``allocate'' it to one of the
bp-components.
Informally speaking, our plan is to return $k$ sub-instances such
that one of the sub-instances is of the form
$\ubpcomp{C_j}{v_{\ell}}{w_{\ell}}$ while the other sub-instances
are of the form $\sbpcomp{C_i}{v_{\ell}}{w_{\ell}}$;
this can be viewed as ``allocating'' the zero-edge ${v_{\ell}w_{\ell}}$
to one carefully chosen bp-component $C_j$ by
``mapping'' $C_j$ to the sub-instance $\ubpcomp{C_j}{v_{\ell}}{w_{\ell}}$,
while all other bp-components $C_i$ ($i\neq j$) are ``mapped'' to
sub-instances $\sbpcomp{C_i}{v_{\ell}}{w_{\ell}}$.
We ``allocate'' the zero-edge ${v_{\ell}w_{\ell}}$ as follows:
For each $i=2,\dots,k$,
we compute $\ndtwocost(\ubpcomp{C_i}{v_{\ell}}{w_{\ell}})$ and
$\ndtwocost(\sbpcomp{C_i}{v_{\ell}}{w_{\ell}})$.
Suppose that there is an index $j\in\{2,\dots,k\}$ such that
these two numbers are the same for $C_j$.
Then we ``allocate'' the zero-edge to $C_j$;
in case of ties, we pick any $j$ such that
$\ndtwocost(\ubpcomp{C_j}{v_{\ell}}{w_{\ell}}) =
\ndtwocost(\sbpcomp{C_j}{v_{\ell}}{w_{\ell}})$.
On the other hand, if the two numbers differ for each $j\in\{2,\dots,k\}$,
then we ``allocate'' the zero-edge to $C_1$.
(Although this allocation may disagree
with the allocation used by the optimal solution,
it turns out that we incur no ``loss''.)
This brings us to the end of the iteration for
the bad-pair $\{v_{\ell},w_{\ell}\}$;
the algorithm applies the same method to the ``remaining graph,'' namely,
either $\ubpcomp{C_1}{v_{\ell}}{w_{\ell}}$ or $\sbpcomp{C_1}{v_{\ell}}{w_{\ell}}$.
We mention that $C_1$ plays a special role in our pre-processing algorithm;
this will become clear when we prove its correctness
(see Lemma~\ref{lem:ppcorrect} below).
See the example in Figure~\ref{f:bp:badalloc}.

{
\begin{figure}[htb]
    \centering
    \begin{subfigure}{0.3\textwidth}
        \centering
        \begin{tikzpicture}[scale=0.6]
            \begin{scope}[every node/.style={circle, fill=black, draw, inner sep=0pt,
            minimum size = 0.15cm
            }]
                \node[label={[label distance=2]90:$a$}] (v1) at (0,0) {};
                \node[] (v2) at (-1,1) {};
                \node[] (v3) at (-2.25,0.75) {};
                \node[] (v4) at (-2.25,-0.75) {};
                \node[] (v5) at (-1,-1) {};
                \node[label={[label distance=2]180:$b$}] (w2) at (1,1) {};
                \node[] (w3) at (2,0) {};
                \node[] (w4) at (1,-1) {};
            \end{scope}
    
            \begin{scope}[every edge/.style={draw=black}]
                \path[] (v1) edge node {} (v2);
                \path[] (v2) edge node {} (v3);
                \path[] (v3) edge node {} (v4);
                \path[] (v4) edge node {} (v5);
                \path[] (v5) edge node {} (v1);
                \path[thick, dashed] (v1) edge node {} (w2);
                \path[] (w2) edge node {} (w3);
                \path[] (w3) edge node {} (w4);
                \path[] (w4) edge node {} (v1);
                \path[] (v3) edge[bend left=60] node {} (w2);
            \end{scope}
    
            \begin{scope}[every node/.style={draw=none,rectangle}]
                \node (C1label) at (-1.5,0) {$C_1$};
                \node (C2label) at (1,0) {$C_2$};
            \end{scope}
        \end{tikzpicture}
        \caption{\centering
	A bad-pair $\{a,b\}$ whose bp-components are $C_1$ and $C_2$.}
        \label{badalloc-a}
    \end{subfigure}
    \hspace*{\fill}
    \begin{subfigure}{0.3\textwidth}
        \centering
        \begin{tikzpicture}[scale=0.6]
            \begin{scope}[every node/.style={circle, fill=black, draw, inner sep=0pt,
            minimum size = 0.15cm
            }]
                \node[label={[label distance=2]90:$a$}] (v1) at (0,0) {};
                \node[] (v2) at (-1,1) {};
                \node[] (v3) at (-2.25,0.75) {};
                \node[] (v4) at (-2.25,-0.75) {};
                \node[] (v5) at (-1,-1) {};
                \node[label={[label distance=2]180:$b$}] (w2) at (1,1) {};
                \node[] (w1) at (2,0.5) {};
                \node[] (w3) at (3.5,0) {};
                \node[] (w4) at (2.5,-1) {}; 
            \end{scope}
    
            \begin{scope}[every edge/.style={draw=black}]
                \path[] (v1) edge node {} (v2);
                \path[] (v2) edge node {} (v3);
                \path[] (v3) edge node {} (v4);
                \path[] (v4) edge node {} (v5);
                \path[] (v5) edge node {} (v1);
                \path[] (w1) edge node {} (w3);
                \path[] (w3) edge node {} (w4);
                \path[] (w1) edge node {} (w4);
                \path[thick, dashed] (v1) edge node {} (w2);
                \path[] (v3) edge[bend left=60] node {} (w2);
            \end{scope}
    
            \begin{scope}[every node/.style={draw=none,rectangle}]
            \end{scope}
        \end{tikzpicture}
        \caption{\centering The graphs $C_1^{\{a,b\}}$ (left) and $C_2^{\oslash}$ (right).}
        \label{badalloc-b}
    \end{subfigure}
    \hspace*{\fill}
    \begin{subfigure}{0.3\textwidth}
        \centering
        \begin{tikzpicture}[scale=0.6]
            \begin{scope}[every node/.style={circle, fill=black, draw, inner sep=0pt,
            minimum size = 0.15cm
            }]
                \node[] (v1) at (0,0) {};
                \node[] (v2) at (-1,1) {};
                \node[] (v3) at (-2.25,0.75) {};
                \node[] (v4) at (-2.25,-0.75) {};
                \node[] (v5) at (-1,-1) {};
                \node[label={[label distance=2]90:$a$}] (w1) at (1,0) {};
                \node[label={[label distance=2]180:$b$}] (w2) at (2,1) {};
                \node[] (w3) at (3,0) {};
                \node[] (w4) at (2,-1) {};
            \end{scope}
    
            \begin{scope}[every edge/.style={draw=black}]
                \path[] (v1) edge node {} (v2);
                \path[] (v2) edge node {} (v3);
                \path[] (v3) edge node {} (v4);
                \path[] (v4) edge node {} (v5);
                \path[] (v5) edge node {} (v1);
                \path[] (v1) edge node {} (v3);
                \path[thick, dashed] (w1) edge node {} (w2);
                \path[] (w2) edge node {} (w3);
                \path[] (w3) edge node {} (w4);
                \path[] (w4) edge node {} (w1);
            \end{scope}
    
            \begin{scope}[every node/.style={draw=none,rectangle}]
            \end{scope}
        \end{tikzpicture}
        \caption{\centering The graphs $C_1^{\oslash}$ (left) and $C_2^{\{a,b\}}$ (right).}
        \label{badalloc-c}
    \end{subfigure}
    \caption{\protect\centering
	An example illustrating a ``bad allocation'' of the zero-edge of a bad-pair.
	The graph (left subfigure) has one bad-pair $\{a,b\}$
	and its bp-components are $C_1$ (with 4~nodes) and $C_2$ (with 2~nodes);
	all edges have cost~one, except $ab$; note that $\opt=8$.
	If $ab$ is ``allocated'' to $C_1$, the sum of the costs of
	the \DTWO\ subgraphs of the resulting sub-instances is~9. }
    \label{f:bp:badalloc}
\end{figure}
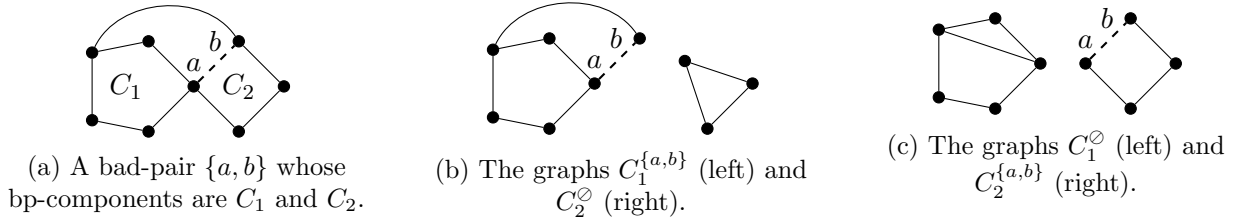
}

\begin{fact} \label{fact:2NCparts} \label{fact:dtwobpcomp}
Let $G$ be 2NC (by assumption $G$ has no $\{0,1\}$-edge-pairs).
Let $\{v,w\}$ be a bad-pair of $G$, and
let $C$ be one of the bp-components of $\{v,w\}$.
\begin{itemize}
\item[(i)]
Then both $\ubpcomp{C}{v}{w}$ and $\sbpcomp{C}{v}{w}$ are 2NC.
\item[(ii)]
Suppose that $G$ has no redundant~4-cycles and has $\ell\ge1$ bad-pairs.
Then $\sbpcomp{C}{v}{w}$
	has no redundant~4-cycles, and it has $\le\ell-1$ bad-pairs,
whereas $\ubpcomp{C}{v}{w}$
	has at~most~one redundant~4-cycle, and it has $\le\ell-1$ bad-pairs.
\item[(iii)]
Suppose that $G$ has no redundant~4-cycles and
suppose that $\sbpcomp{C}{v}{w}$ has no bad-pairs (of its own)
and $\ubpcomp{C}{v}{w}$ has no bad-pairs (of its own).
Then $\ds \ndtwocost(\sbpcomp{C}{v}{w}) \leq \ndtwocost(\ubpcomp{C}{v}{w})$.
\end{itemize}
\end{fact}
\begin{proof}
We prove each of the three parts.
\begin{description}
{
\item[(i)]
First, consider $\ubpcomp{C}{v}{w}$; observe that it has $\ge3$ nodes;
since $G$ is 2NC and $C$ is a connected~component of $G-\{v,w\}$,
for each node $z\in V(C)$,
there exist two openly disjoint paths between $z$ and $\{v,w\}$ in $\ubpcomp{C}{v}{w}$;
hence, it can be seen that $\ubpcomp{C}{v}{w}$ is 2NC.

Now, consider  $\sbpcomp{C}{v}{w}$, and w.l.o.g., assume that $C$ has $\ge2$ nodes.
Then, by definition, $\sbpcomp{C}{v}{w}$ has $\ge3$ nodes.
Let $v^*$ denote the node resulting from the contraction of $vw$.
Arguing as above, 
for each node $z\in V(C)$,
there exist two openly disjoint paths between $z$ and $v^*$ in $\sbpcomp{C}{v}{w}$;
moreover, $C = \sbpcomp{C}{v}{w} - v^*$ has a single connected~component;
hence, it can be seen that $\sbpcomp{C}{v}{w}$ is 2NC.

\item[(ii)]
Observe that any bad-pair of $\ubpcomp{C}{v}{w}$ (respectively,
$\sbpcomp{C}{v}{w}$) is a bad-pair of $G$.
Also, note that $\{v,w\}$ is a bad-pair of $G$ but it is not a
bad-pair of $\ubpcomp{C}{v}{w}$
(since $C = \ubpcomp{C}{v}{w} - \{v,w\}$ is connected).

Consider $\sbpcomp{C}{v}{w}$, and w.l.o.g., assume that $C$ has $\ge2$ nodes.
Let $v^*$ denote the node resulting from the contraction of $vw$.
Note that a redundant~4-cycle of $\sbpcomp{C}{v}{w}$ cannot be incident to $v^*$
(since every edge incident to $v^*$ in $\sbpcomp{C}{v}{w}$ is a unit-edge).
It follows that any redundant~4-cycle of $\sbpcomp{C}{v}{w}$ is a
redundant~4-cycle of $G$.
Clearly, $\sbpcomp{C}{v}{w}$ has $\leq \ell-1$ bad-pairs, and it has no
redundant~4-cycles; thus, part~(ii) holds for $\sbpcomp{C}{v}{w}$.

Next, we verify part~(ii) for $\ubpcomp{C}{v}{w}$.
Clearly, $\ubpcomp{C}{v}{w}$ has $\leq \ell-1$ bad-pairs.
Observe that $\ubpcomp{C}{v}{w}$ has at~most~one redundant~4-cycle
containing the zero-edge $vw$.  Also, observe that any redundant~4-cycle
of $\ubpcomp{C}{v}{w}$ that is disjoint from $\{v,w\}$ is also a
redundant~4-cycle of $G$.  It follows that $\ubpcomp{C}{v}{w}$ has
at~most~one redundant~4-cycle.

\item[(iii)]
Observe that $\sbpcomp{C}{v}{w}$ is a well-structured MAP~instance, hence,
$\ndtwocost(\sbpcomp{C}{v}{w}) = \dtwocost(\sbpcomp{C}{v}{w})$.

First, suppose that $\ubpcomp{C}{v}{w}$ has no redundant~4-cycles.
Then, $\ubpcomp{C}{v}{w}$ is a well-structured MAP~instance, hence,
$\ndtwocost(\ubpcomp{C}{v}{w}) = \dtwocost(\ubpcomp{C}{v}{w})$.
Moreover, we have
$\dtwocost(\sbpcomp{C}{v}{w}) \leq \dtwocost(\ubpcomp{C}{v}{w})$;
to see this, note that we can start with $\DTWO(\ubpcomp{C}{v}{w})$ and
contract the zero-edge $vw$ to get a 2-edge~cover of $\sbpcomp{C}{v}{w}$.
Therefore, part~(iii) holds in this case.

Next, suppose that $\ubpcomp{C}{v}{w}$ has a redundant~4-cycle.
Then, $\ubpcomp{C}{v}{w}$ has a unique redundant~4-cycle $Q$ that
contains the zero-edge $vw$.
Let $\wt{C}$ denote the graph obtained from $\ubpcomp{C}{v}{w}$ by contracting $Q$.
Then, we have
$\dtwocost(\sbpcomp{C}{v}{w}) \leq 2 + \dtwocost(\wt{C}) \leq \ndtwocost(\ubpcomp{C}{v}{w})$;
the second inequality follows from the definition of $\ndtwocost(\ubpcomp{C}{v}{w})$
\big(the cost of the union of the $\DTWO$ subgraphs of the
well-structured MAP~instances obtained from $\wt{C}$ by repeatedly
applying (pp2)~and~(pp3) is $\ndtwocost(\ubpcomp{C}{v}{w})-2$\big);
to verify the first inequality note that $E(Q) \cup E(\DTWO(\wt{C}))$
is a 2-edge~cover of $\ubpcomp{C}{v}{w}$, and if we contract the
zero-edge $vw$ then we get a 2-edge~cover of $\sbpcomp{C}{v}{w}$.
Thus, part~(iii) holds in this case.
}
\end{description}
\end{proof}

See Figure~\ref{f:show-tauhat} for an illustration.

{
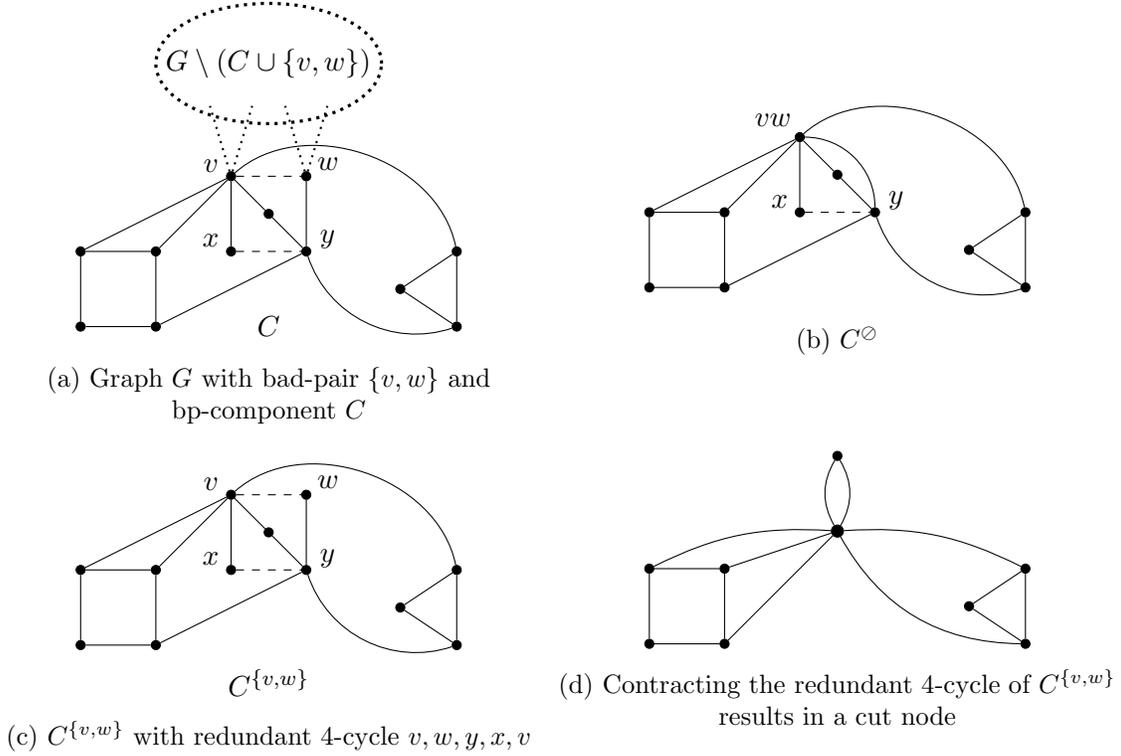
\begin{figure}[ht]
    \centering
    \begin{subfigure}{0.45\textwidth}
        \centering
        \begin{tikzpicture}[scale=1]
            \begin{scope}[every node/.style={circle, fill=black, draw, inner sep=0pt,
            minimum size = 0.12cm
            }]
                \draw[very thick, dotted] (0.5,1.5) ellipse (1.5cm and 0.8cm);
                \node[draw=none,fill=none] (v1) at (-0.3,1) {};
                \node[draw=none,fill=none] (v2) at (0.3,1) {};
                \node[draw=none,fill=none] (w1) at (0.7,1) {};
                \node[draw=none,fill=none] (w2) at (1.3,1) {};
                \node[label={[label distance=3]165:$v$}] (v) at (0,0) {};
                \node[label={[label distance=3]15:$w$}] (w) at (1,0) {};
                \node[label={[label distance=3]165:$x$}] (x) at (0,-1) {};
                \node[label={[label distance=3]15:$y$}] (y) at (1,-1) {};
                
                \node[] (vy) at (0.5,-0.5) {};
                
                \node[] (a1) at (-2,-1) {};
                \node[] (a2) at (-1,-1) {};
                \node[] (a3) at (-1,-2) {};
                \node[] (a4) at (-2,-2) {};
                
                \node[] (b1) at (3,-1) {};
                \node[] (b2) at (3,-2) {};
                \node[] (b3) at (2.25,-1.5) {};
            \end{scope}
    
            \begin{scope}[every edge/.style={draw=black}]
                
                \path[thick,dotted] (v) edge[] node {} (v1);
                \path[thick,dotted] (v) edge[] node {} (v2);
                \path[thick,dotted] (w) edge[] node {} (w1);
                \path[thick,dotted] (w) edge[] node {} (w2);
                
                \path[dashed] (v) edge[] node {} (w);
                \path[dashed] (x) edge[] node {} (y);
                \path[] (v) edge[] node {} (x);
                \path[] (w) edge[] node {} (y);
                
                \path[] (v) edge[] node {} (vy);
                \path[] (vy) edge[] node {} (y);
                
                \path[] (a1) edge[] node {} (a2);
                \path[] (a2) edge[] node {} (a3);
                \path[] (a3) edge[] node {} (a4);
                \path[] (a4) edge[] node {} (a1);
                \path[] (a1) edge[] node {} (v);
                \path[] (a2) edge[] node {} (v);
                \path[] (a3) edge[] node {} (y);
                
                \path[] (b1) edge[] node {} (b2);
                \path[] (b2) edge[] node {} (b3);
                \path[] (b3) edge[] node {} (b1);
                \path[] (b1) edge[bend right = 60] node {} (v);
                \path[] (b2) edge[bend left = 45] node {} (y);
                
            \end{scope}
    
            \begin{scope}[every node/.style={draw=none,rectangle}]
                \node (Clabel) at (0.5,-2) {$C$};
                \node (Glabel) at (0.5,1.5) {$G \setminus (C\cup\{v,w\})$};
            \end{scope}
        \end{tikzpicture}
        \caption{\centering Graph $G$ with bad-pair $\{v,w\}$ and bp-component $C$}
        \label{placeholder1-a}
    \end{subfigure}
    \begin{subfigure}{0.45\textwidth}
        \centering
        \begin{tikzpicture}[scale=1]
            \begin{scope}[every node/.style={circle, fill=black, draw, inner sep=0pt,
            minimum size = 0.12cm
            }]
                \node[label={[label distance=3]165:$vw$}] (vw) at (0,0) {};
                \node[label={[label distance=3]165:$x$}] (x) at (0,-1) {};
                \node[label={[label distance=3]15:$y$}] (y) at (1,-1) {};
                
                \node[] (vy) at (0.5,-0.5) {};
                
                \node[] (a1) at (-2,-1) {};
                \node[] (a2) at (-1,-1) {};
                \node[] (a3) at (-1,-2) {};
                \node[] (a4) at (-2,-2) {};
                
                \node[] (b1) at (3,-1) {};
                \node[] (b2) at (3,-2) {};
                \node[] (b3) at (2.25,-1.5) {};
            \end{scope}
    
            \begin{scope}[every edge/.style={draw=black}]
                
                \path[dashed] (x) edge[] node {} (y);
                \path[] (vw) edge[] node {} (x);
                \path[] (vw) edge[bend left=45] node {} (y);
                
                \path[] (vw) edge[] node {} (vy);
                \path[] (vy) edge[] node {} (y);
                
                \path[] (a1) edge[] node {} (a2);
                \path[] (a2) edge[] node {} (a3);
                \path[] (a3) edge[] node {} (a4);
                \path[] (a4) edge[] node {} (a1);
                \path[] (a1) edge[] node {} (vw);
                \path[] (a2) edge[] node {} (vw);
                \path[] (a3) edge[] node {} (y);
                
                \path[] (b1) edge[] node {} (b2);
                \path[] (b2) edge[] node {} (b3);
                \path[] (b3) edge[] node {} (b1);
                \path[] (b1) edge[bend right = 60] node {} (vw);
                \path[] (b2) edge[bend left = 45] node {} (y);
                
            \end{scope}
    
            \begin{scope}[every node/.style={draw=none,rectangle}]
            \end{scope}
        \end{tikzpicture}
        \caption{\centering $C^{\oslash}$}
        \label{placeholder1-d}
    \end{subfigure}
    \begin{subfigure}{0.45\textwidth}
        \centering
        \begin{tikzpicture}[scale=1]
            \begin{scope}[every node/.style={circle, fill=black, draw, inner sep=0pt,
            minimum size = 0.12cm
            }]
                \node[label={[label distance=3]165:$v$}] (v) at (0,0) {};
                \node[label={[label distance=3]15:$w$}] (w) at (1,0) {};
                \node[label={[label distance=3]165:$x$}] (x) at (0,-1) {};
                \node[label={[label distance=3]15:$y$}] (y) at (1,-1) {};
                
                \node[] (vy) at (0.5,-0.5) {};
                
                \node[] (a1) at (-2,-1) {};
                \node[] (a2) at (-1,-1) {};
                \node[] (a3) at (-1,-2) {};
                \node[] (a4) at (-2,-2) {};
                
                \node[] (b1) at (3,-1) {};
                \node[] (b2) at (3,-2) {};
                \node[] (b3) at (2.25,-1.5) {};
            \end{scope}
    
            \begin{scope}[every edge/.style={draw=black}]
                
                \path[dashed] (v) edge[] node {} (w);
                \path[dashed] (x) edge[] node {} (y);
                \path[] (v) edge[] node {} (x);
                \path[] (w) edge[] node {} (y);
                
                \path[] (v) edge[] node {} (vy);
                \path[] (vy) edge[] node {} (y);
                
                \path[] (a1) edge[] node {} (a2);
                \path[] (a2) edge[] node {} (a3);
                \path[] (a3) edge[] node {} (a4);
                \path[] (a4) edge[] node {} (a1);
                \path[] (a1) edge[] node {} (v);
                \path[] (a2) edge[] node {} (v);
                \path[] (a3) edge[] node {} (y);
                
                \path[] (b1) edge[] node {} (b2);
                \path[] (b2) edge[] node {} (b3);
                \path[] (b3) edge[] node {} (b1);
                \path[] (b1) edge[bend right = 60] node {} (v);
                \path[] (b2) edge[bend left = 45] node {} (y);
                
            \end{scope}
    
            \begin{scope}[every node/.style={draw=none,rectangle}]
                \node (Clabel) at (0.5,-2.5) {$C^{\{v,w\}}$};
            \end{scope}
        \end{tikzpicture}
        \caption{\centering $C^{\{v,w\}}$ with redundant~4-cycle $v,w,y,x,v$}
        \label{placeholder1-b}
    \end{subfigure}
    \begin{subfigure}{0.45\textwidth}
        \centering
        \begin{tikzpicture}[scale=1]
            \begin{scope}[every node/.style={circle, fill=black, draw, inner sep=0pt,
            minimum size = 0.12cm
            }]
                \node[minimum size = 0.16cm] (vwxy) at (0.5,-0.5) {};
                \node[] (vy) at (0.5,0.5) {};
                
                \node[] (a1) at (-2,-1) {};
                \node[] (a2) at (-1,-1) {};
                \node[] (a3) at (-1,-2) {};
                \node[] (a4) at (-2,-2) {};
                
                \node[] (b1) at (3,-1) {};
                \node[] (b2) at (3,-2) {};
                \node[] (b3) at (2.25,-1.5) {};
            \end{scope}
    
            \begin{scope}[every edge/.style={draw=black}]
                
                \path[] (vwxy) edge[bend left=30] node {} (vy);
                \path[] (vwxy) edge[bend right=30] node {} (vy);
                
                \path[] (a1) edge[] node {} (a2);
                \path[] (a2) edge[] node {} (a3);
                \path[] (a3) edge[] node {} (a4);
                \path[] (a4) edge[] node {} (a1);
                \path[] (a1) edge[bend left=15] node {} (vwxy);
                \path[] (a2) edge[] node {} (vwxy);
                \path[] (a3) edge[] node {} (vwxy);
                
                \path[] (b1) edge[] node {} (b2);
                \path[] (b2) edge[] node {} (b3);
                \path[] (b3) edge[] node {} (b1);
                \path[] (b1) edge[bend right=15] node {} (vwxy);
                \path[] (b2) edge[bend left=30] node {} (vwxy);
                
            \end{scope}
    
            \begin{scope}[every node/.style={draw=none,rectangle}]
            \end{scope}
        \end{tikzpicture}
        \caption{\centering Contracting the redundant~4-cycle of $C^{\{v,w\}}$ results in a cut node}
        \label{placeholder1-c}
    \end{subfigure}
    \caption{\protect\centering Illustration of a bad-pair $\{v,w\}$ and
$\sbpcomp{C}{v}{w}$ and $\ubpcomp{C}{v}{w}$ for a bp-component $C$.
Dashed lines indicate zero-edges.
We have $\ndtwocost(\sbpcomp{C}{v}{w}) = \dtwocost(\sbpcomp{C}{v}{w}) = 10$;
note that $\DTWO(\sbpcomp{C}{v}{w})$ consists of a 3-cycle and two 4-cycles.
Observe that $\ubpcomp{C}{v}{w}$ has a redundant~4-cycle $v,w,y,x,v$, and
the contraction of this redundant~4-cycle results in a cut~node;
the resulting graph ``decomposes'' into three blocks whose $\DTWO$ subgraphs
have costs $2,4,5$, respectively, hence, $\ndtwocost(\ubpcomp{C}{v}{w}) = 13$.
}
	\label{f:show-tauhat}
\end{figure}
}

\authremark{
Recall that $k$ denotes the number of bp-components of the bad-pair
$\{v,w\}$.  Readers may focus on the case of $k=2$ for the rest of
this section; our presentation is valid for any~$k\ge2$.
}

The next fact is essential for our analysis.
It states that for any 2-ECSS $H$ of $G$ that contains all the zero-edges, a bad-pair $\{v,w\}$ of $G$,
and any bp-component $C_i$ of $\{v,w\}$,
the subgraph~of~$H$ induced by
$V(C_i)\cup\{v,w\}$ is either 2EC or
it is connected and has $vw$ as its unique bridge.
\big(To see this, observe that for any node $u\in{V(C_i)}$,
$H$ has two edge-disjoint paths that start at $u$,
end at either $v$ or $w$, and have all internal nodes in $V(C_i)$.\big)
Therefore, the subgraph obtained from $H[V(C_i)\cup\{v,w\}]$
by contracting the edge $vw$ is 2EC.
This implies the key lower~bound,
$ \cost( H[V(C_i)\cup\{v,w\}] ) \ge \opt( \sbpcomp{C_i}{v}{w} )$.

\begin{fact} \label{fact:piecebound}
Let $\{v,w\}$ be  a bad-pair, and let $C_1,\dots,C_k$
be all of its bp-components.
Let $H$ be a 2-ECSS of $G$ that contains all the zero-edges.
Consider any $C_i$, where $i\in\{1,\dots,k\}$.
Then $H[V(C_i)\cup\{v,w\}]$ is connected, and moreover,
either it is 2EC or it has exactly one bridge, namely, ${vw}$.
Therefore, $H[V(\sbpcomp{C_i}{v}{w})]$ is 2EC.
Hence, $\cost(H[V(C_i)\cup\{v,w\}]) \geq
	\opt(\sbpcomp{C_i}{v}{w})$.
\end{fact}

\begin{fact} \label{fact:2ECbyunion}
Let $\{v,w\}$ be  a bad-pair, and let $C_1,\dots,C_k$
be all of its bp-components.
Suppose that we pick one of these bp-components $C_j$
and compute a 2-ECSS $H_j$
of $\ubpcomp{C_j}{v}{w}$.
For each of the other bp-components $C_i$, $i\in\{1,\dots,k\}-\{j\}$,
we compute a 2-ECSS $H_i$ 
of $\sbpcomp{C_i}{v}{w}$.
Then the union of the edge sets of
$H_1,\dots,H_{j-1}, H_j, H_{j+1},\dots,H_k$
gives a 2-ECSS of $G$.
\end{fact}

\pagebreak[2]

\newenvironment{absolutelynopagebreak}
  {\par\nobreak\vfil\penalty0\vfilneg
   \vtop\bgroup}
  {\par\xdef\tpd{\the\prevdepth}\egroup
   \prevdepth=\tpd}
\begin{absolutelynopagebreak}
{
\subsubsection{Pre-processing algorithm}

Suppose that $G$ is 2NC and it has one or more bad-pairs.

{
\medskip
\noi
\fbox{ \begin{minipage}{\textwidth}

\textbf{Pre-processing Algorithm (outline)}:
\begin{itemize}
\item[(0)]
	pick a bad-pair $\{v,w\}$ that satisfies the condition in Fact~\ref{fact:pickbp}, and
	let its bp-components be $C_1,C_2,\dots,C_k$, where
	$C_2,\dots,C_k$ are leaf~bp-components (free of bad-pairs);

\item[(1)]
for each $i=2,\dots,k$, we compare $\ndtwocost(\sbpcomp{C_i}{v}{w})$ and
  $\ndtwocost(\ubpcomp{C_i}{v}{w})$;
if the former is strictly smaller than the latter for all
  $i=2,\dots,k$, then we return the list of graphs
  $\sbpcomp{C_2}{v}{w},\dots,\sbpcomp{C_k}{v}{w}$ and
  $\ubpcomp{C_1}{v}{w}$;
informally speaking, we allocate the zero-edge
  ${vw}$ to $C_1$;

\item[(2)]
otherwise, we have at least one $j\in\{2,\dots,k\}$ such that
  $\ndtwocost(\sbpcomp{C_j}{v}{w}) ~=~ \ndtwocost(\ubpcomp{C_j}{v}{w})$;
then we return the list of graphs
  $\sbpcomp{C_1}{v}{w}, \dots, \sbpcomp{C_{j-1}}{v}{w},
  \sbpcomp{C_{j+1}}{v}{w}, \dots, \sbpcomp{C_k}{v}{w}$ and
  $\ubpcomp{C_j}{v}{w}$;
informally speaking, we allocate the zero-edge ${vw}$ to $C_j$;

\item[(3)]
let $G'$ denote the ``remaining graph,'' either $\ubpcomp{C_1}{v}{w}$ or
$\sbpcomp{C_1}{v}{w}$;
\textbf{stop} if $G'$ has no bad-pairs, otherwise,
apply the same pre-processing iteration to $G'$.

\end{itemize}
\end{minipage}
}
\medskip
}
}
\end{absolutelynopagebreak}

For any instance of MAP $G'$, we use $\LB(G')$ to denote
\[ \left\{ \begin{array}{ll}
	\ndtwocost(G'),\qquad & \textup{~if $G'$ has no cut~nodes and no bad-pairs (and no $\{0,1\}$-edge-pairs)}, \\
	\opt(G'),\qquad & \textup{~otherwise}.
	\end{array}
   \right. %\}
\]

\begin{fact} \label{fact:approxbytauhat}
Suppose that a 2-ECSS of cost $\leq\alpha\cdot \dtwocost(G_i)$ can be
computed in polynomial time for any well-structured MAP instance $G_i$.
Then for any MAP~instance $G'$ that has no cut~nodes and no bad-pairs (and no $\{0,1\}$-edge-pairs),
a 2-ECSS of cost $\le\alpha \cdot \ndtwocost(G')$ can be computed in polynomial~time.
\end{fact}
\begin{proof}
By the pre-processing,
we ``decompose'' $G'$ into a collection of well-structured MAP~instances $G'_1,\dots,G'_q$;
then, for each $G'_i$, we compute a 2-ECSS $H'_i$ of
cost $\leq\alpha\cdot \dtwocost(G'_i)$;
finally, we undo the transformations applied by the pre-processing
to obtain a 2-ECSS $H'$ of $G'$ from the collection $H'_1,\dots,H'_q$.
Our previous arguments (recall the proof of Fact~\ref{fact:sharp-lb})
imply that $\cost(H') \leq \alpha\cdot \ndtwocost(G')$.
\end{proof}

\begin{lemma} \label{lem:ppcorrect}
Let $G$ be a MAP~instance that has no cut-nodes and has
$\rho\ge1$ bad-pairs (by assumption $G$ has no $\{0,1\}$-edge-pairs).
Let $G_1,\dots,G_k$ denote the collection of graphs
obtained by applying one iteration of the above pre-processing algorithm
for a bad-pair $\{v,w\}$ of $G$
(thus, $\{v,w\}$ satisfies the condition in Fact~\ref{fact:pickbp}, and
each $G_i$ is of the form $\ubpcomp{C_i}{v}{w}$ or $\sbpcomp{C_i}{v}{w}$).
Then we have
\\
\centerline{
$\ds \LB(G) ~~\geq~~ \LB(G_1) ~+~ \sum_{i=2}^{k} \ndtwocost(G_i)$;
}
\\
moreover, $G_1$ has $\leq \rho-1$ bad-pairs, and $G_2,\dots,G_k$ have no bad-pairs.
\end{lemma}
\begin{proof}
Let $C_1,\dots,C_k$ denote all of the bp-components of $\{v,w\}$;
note that $C_2,\dots,C_k$ are leaf~bp-components (free of bad-pairs).
For each $i=1,\dots,k,$ let $\vbpcomp{V_i}{v}{w}$
denote the set of nodes $V(C_i)\cup\{v,w\}$;
thus, $\vbpcomp{V_i}{v}{w}$ denotes the node-set of $\ubpcomp{C_i}{v}{w}$.

Let ${G^*}$ denote an optimal solution,
i.e., a 2-ECSS of minimum cost, and
w.l.o.g.\ assume that ${G^*}$ contains all zero-edges of $G$.
We have
$\opt(G) = \cost({G^*}) = \sum_{i=1}^k \cost( G^*[\vbpcomp{V_i}{v}{w}] )$.

By Fact~\ref{fact:piecebound}, we have
\begin{multline*}
\textnormal{~~~~~~~~~~if $G_i$ is of the form $\sbpcomp{C_i}{v}{w}$, then~~} 
\cost( G^*[\vbpcomp{V_i}{v}{w}] ) \geq \opt(\sbpcomp{C_i}{v}{w}) \geq \ndtwocost(G_i),
	\quad \forall i\in\{2,\dots,k\}.
\qquad (*)
\end{multline*}

We complete the proof by examining a few cases.

\begin{description}{
\item[Case~1: the zero-edge ${vw}$ is allocated to $C_1$:\quad]
First, suppose that $G^*[\vbpcomp{V_1}{v}{w}]$ is 2EC. Then, we have
$\cost( G^*[\vbpcomp{V_1}{v}{w}] ) \ge \opt(G_1) \ge
	\LB(G_1)$;
combining this with $(*)$ we have
\[ \opt(G) = \sum_{i=1}^k \cost( G^*[\vbpcomp{V_i}{v}{w}] ) \ge
  \sum_{i=1}^k \opt(G_i) \ge \sum_{i=1}^k \LB(G_i).
\]

Now, suppose that $G^*[\vbpcomp{V_1}{v}{w}]$ is not 2EC;
then, by Fact~\ref{fact:piecebound},
it is connected and has only one bridge, namely, ${vw}$.
Moreover, $G$ has an edge $\hat{e}\not={vw}$
whose end~nodes are in two different connected~components
of ${G^*}[\vbpcomp{V_1}{v}{w}] - {vw}$.
(To see this, note that $G[\vbpcomp{V_1}{v}{w}]$
is 2NC, hence, $(G[\vbpcomp{V_1}{v}{w}] - {vw})$
has a $v,w$ path, and one of the edges of this path
satisfies the requirement on $\hat{e}$.)
Thus, adding $\hat{e}$ to ${G^*}[\vbpcomp{V_1}{v}{w}]$ results in a 2EC
graph of cost
$1+\cost({G^*}[\vbpcomp{V_1}{v}{w}])\ge\opt(G_1)\ge\LB(G_1)$.
Also, observe that ${G^*}$ has two edge-disjoint $v,w$ paths, hence,
one of the subgraphs $G^*[\vbpcomp{V_{\ell}}{v}{w}] \;-\; {vw}$ has a $v,w$ path,
where ${\ell}\in\{2,\dots,k\}$.  Hence,
\[
\cost({G^*}[\vbpcomp{V_{\ell}}{v}{w}]) \geq \opt(\ubpcomp{C_{\ell}}{v}{w})
	\geq \ndtwocost(\ubpcomp{C_{\ell}}{v}{w})
	\geq 1 + \ndtwocost(\sbpcomp{C_{\ell}}{v}{w}) = 1+\ndtwocost(G_{\ell})
\]
(we used the fact that $\ndtwocost(\ubpcomp{C_i}{v}{w}) ~>~
  \ndtwocost(\sbpcomp{C_i}{v}{w}),\; \forall{i}\in\{2,\dots,k\}$).
By the above inequalities and $(*)$, we have
\vspace{-6ex}

\begin{multline*}
\opt(G) = \sum_{i=1}^k \cost( G^*[\vbpcomp{V_i}{v}{w}] ) \ge \\[-4ex]
	(\LB(G_1)-1) + (\ndtwocost(G_{\ell})+1) +
	\sum_{i\in\{2,\dots,\ell-1,\ell+1,\dots,k\}} \ndtwocost(G_i) = \LB(G_1) ~+~ \sum_{i=2}^{k} \ndtwocost(G_i).
\end{multline*}

\item[Case~2: the zero-edge ${vw}$ is allocated to $C_{\ell}$, ${\ell}\in\{2,\dots,k\}$:\quad]
Thus, we have $\ndtwocost(\sbpcomp{C_{\ell}}{v}{w}) = \ndtwocost(\ubpcomp{C_{\ell}}{v}{w})$
(otherwise, ${vw}$ cannot be allocated to $C_{\ell}$).
Then, by Fact~\ref{fact:piecebound}, we have
\[ \cost( G^*[\vbpcomp{V_{\ell}}{v}{w}] ) \geq
    \opt(\sbpcomp{C_{\ell}}{v}{w}) \geq
    \ndtwocost(\sbpcomp{C_{\ell}}{v}{w}) = \ndtwocost(\ubpcomp{C_{\ell}}{v}{w}).
\]
This, together with $(*)$ and the inequalities
$\cost( G^*[\vbpcomp{V_1}{v}{w}] ) \ge \opt(G_1) \ge \LB(G_1)$,
implies that $\opt(G) \ge \LB(G_1) ~+~ \sum_{i=2}^{k} \ndtwocost(G_i)$.
}
\end{description}
This completes the proof of the lemma.
\end{proof}

\begin{theorem}
Suppose that a 2-ECSS of cost $\leq\alpha\cdot \dtwocost(G_i)$ can be
computed in polynomial time for any well-structured MAP instance $G_i$.
Then for any MAP~instance $G$,
a 2-ECSS of cost $\leq\alpha\cdot \LB(G)$
can be computed in polynomial~time
via the pre-processing
(consisting of (pp1) followed by iterations of (pp2), (pp3), (pp4)).
\end{theorem}
\begin{proof}
We use induction on the number of bad-pairs in the MAP~instance $G$.

Suppose that $G$ has no bad-pairs.
Then, by the pre-processing,
we ``decompose'' $G$ into a collection of well-structured MAP~instances $G_1,\dots,G_q$;
then, for each $G_i$, we compute a 2-ECSS $H_i$ of cost $\leq\alpha\cdot \dtwocost(G_i)$;
finally, we undo the transformations applied by the pre-processing
to obtain a 2-ECSS $H$ of $G$ from the collection $H_1,\dots,H_q$.
It follows from the results in Sections~\ref{s:pre-proc1-1}--\ref{s:pre-proc1-3}
that $\cost(H) \leq \alpha\cdot \LB(G)$.

Now, suppose that the theorem holds for any MAP~instance with $\leq \ell$ bad-pairs.

Suppose that $G$ has $\ell+1$ bad-pairs.
For notational convenience in the induction~step, let us assume that
$G$ has no $\{0,1\}$-edge-pairs and no cut~nodes.
\big(When we have an arbitrary instance of MAP, $G$, then we apply~(pp1)
and~(pp3) to ``decompose'' the instance into a collection of
MAP~instances that each have no $\{0,1\}$-edge-pairs and no cut~nodes;
then, we apply the following proof to each of the ``restricted''
MAP~instances; finally, we undo the transformations applied by~(pp1)
and~(pp3), and we apply the results in Sections~\ref{s:pre-proc1-1}
and~\ref{s:pre-proc1-3} to obtain a 2-ECSS of $G$ of cost
$\leq\alpha\cdot \LB(G)$.\big)

Then, by one iteration of~(pp4) and Lemma~\ref{lem:ppcorrect},
we ``decompose'' $G$ into a collection of MAP~instances $G_1,\dots,G_k$ such that
\begin{itemize}
\item
$G_1$ has no $\{0,1\}$-edge-pairs, no cut~nodes, and $\leq \ell$ bad-pairs,
 and $G_2,\dots,G_k$ have no $\{0,1\}$-edge-pairs, no cut~nodes, and no bad-pairs,
\item
$\LB(G) \ge \LB(G_1) ~+~ \sum_{i=2}^{k} \ndtwocost(G_i)$,
\item
given a 2-ECSS $H_i$ of $G_i$ for each $i\in \{1,\dots,k\}$,
the union of $H_1,\dots,H_k$ forms a 2-ECSS $H$ of $G$, by Fact~\ref{fact:2ECbyunion}.
\end{itemize}

By the induction hypothesis and by Fact~\ref{fact:approxbytauhat}
(which is essential for this analysis),
we may assume that 
$\cost(H_1) \leq \alpha\cdot \LB(G_1)$ and
$\cost(H_i) \leq \alpha\cdot \ndtwocost(G_i), \forall i \in \{2,\dots,k\}$.
Then, $H$ is a 2-ECSS of $G$, and
$\cost(H) \leq \alpha\cdot \LB(G)$
(since $\LB(G_1) + \sum_{i=2}^{k} \ndtwocost(G_i) \leq \LB(G)$).
\end{proof}
}
}

\section{ \label{s:path-thicken}  Bridge~covering}
{
In this section and in Section~\ref{s:algo-last}, we assume that
the input is a well-structured MAP instance.

We start by illustrating our method for bridge~covering on a small example.
After computing \DTWO, recall that we give $1.75$ initial~credits to each
unit-edge of \DTWO, thereby giving each of these edges 
$0.75$ working~credits
(we ``retain'' one credit to buy the unit-edge for our solution).
Consequently, each 2ec-block of \DTWO\ gets $\ge1.5$ working~credits (this is
explained below).  We want to buy more edges to add to \DTWO\ such that
all bridges are ``covered", and we pay for the newly
added edges via the working~credits.

Observe that each 2ec-block of \DTWO\ has $\ge1.5$ working~credits,
because it has $\ge2$ unit-edges;
to see this, suppose that a 2ec-block $B$ has $b$ nodes;
if $b=2$, then $B$ has two parallel unit-edges;
otherwise, $B$ has $\ge b$ edges (since $B$ is 2EC) and
	has $\le\floor{b/2}$ zero-edges (since the zero-edges form a matching),
so $B$ has $\ge\ceiling{b/2}$ unit-edges, and we have $b\ge3$.

In what follows, we use the term credits to mean the working~credits of
the algorithm; this excludes the unit credit retained by every
unit-edge of the current solution subgraph; for example, a 6-cycle of
\DTWO\ that contains four unit-edges has 3~credits (and it has
7~initial~credits).

{
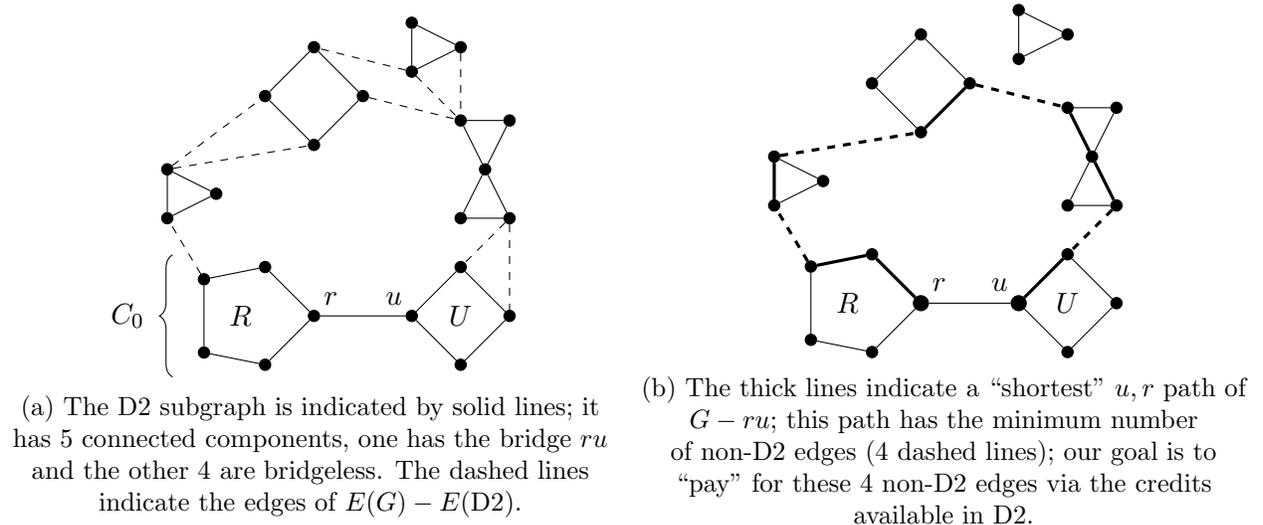
\begin{figure}[ht]
    \begin{subfigure}{0.49\textwidth}
        \centering
        \begin{tikzpicture}[scale=0.65]
            \begin{scope}[every node/.style={circle, fill=black, draw, inner sep=0pt,
            minimum size = 0.15cm
            }]
                
                \node[label={[label distance=3]45:$r$}] (r) at (0,0) {};
                \node[] (c01) at (-1,1) {};
                \node[] (c02) at (-2.25,0.75) {};
                \node[] (c03) at (-2.25,-0.75) {};
                \node[] (c04) at (-1,-1) {};
                \node[label={[label distance=3]135:$u$}] (u) at (2,0) {};
                \node[] (c05) at (3,-1) {};
                \node[] (c06) at (4,0) {};
                \node[] (c07) at (3,1) {};
                
                \node[] (a1) at (-3,2) {};
                \node[] (a2) at (-3,3) {};
                \node[] (a3) at (-2,2.5) {};
                
                \node[] (b1) at (0,3.5) {};
                \node[] (b2) at (-1,4.5) {};
                \node[] (b3) at (0,5.5) {};
                \node[] (b4) at (1,4.5) {};
                
                \node[] (c1) at (3,4) {};
                \node[] (c2) at (4,4) {};
                \node[] (c3) at (3.5,3) {};
                \node[] (c4) at (3,2) {};
                \node[] (c5) at (4,2) {};
                
                \node[] (d1) at (2,5) {};
                \node[] (d2) at (2,6) {};
                \node[] (d3) at (3,5.5) {};
                
            \end{scope}
    
            \begin{scope}[every edge/.style={draw=black}]
                \path (r) edge node {} (c01);
                \path (c01) edge node {} (c02);
                \path (c02) edge node {} (c03);
                \path (c03) edge node {} (c04);
                \path (c04) edge node {} (r);
                \path (r) edge node {} (u);
                \path (u) edge node {} (c05);
                \path (c05) edge node {} (c06);
                \path (c06) edge node {} (c07);
                \path (c07) edge node {} (u);
                
                \path (a1) edge node {} (a2);
                \path (a2) edge node {} (a3);
                \path (a3) edge node {} (a1);
                
                \path (b1) edge node {} (b2);
                \path (b2) edge node {} (b3);
                \path (b3) edge node {} (b4);
                \path (b4) edge node {} (b1);
                
                \path (c1) edge node {} (c2);
                \path (c2) edge node {} (c3);
                \path (c3) edge node {} (c1);
                \path (c3) edge node {} (c4);
                \path (c4) edge node {} (c5);
                \path (c5) edge node {} (c3);
                
                \path (d1) edge node {} (d2);
                \path (d2) edge node {} (d3);
                \path (d3) edge node {} (d1);
                
                \path[dashed] (c02) edge node {} (a1);
                \path[dashed] (a2) edge node {} (b1);
                \path[dashed] (a2) edge node {} (b2);
                \path[dashed] (b4) edge node {} (c1);
                \path[dashed] (c5) edge node {} (c06);
                \path[dashed] (c5) edge node {} (c07);
                \path[dashed] (d1) edge node {} (c1);
                \path[dashed] (d3) edge node {} (c1);
                \path[dashed] (d1) edge node {} (b3);
            \end{scope}
    
            \begin{scope}[every node/.style={draw=none,rectangle}]
                \node (Rlabel) at (-1.5,0) {$R$};
                \node (Ulabel) at (3,0) {$U$};
                \draw [decorate,decoration={brace,amplitude=5pt},xshift=-4pt,yshift=0pt] (-2.75,-1.25) -- (-2.75,1.25) node [black,midway,xshift=-0.6cm] {$C_0$};
            \end{scope}
        \end{tikzpicture}
	\caption{\centering The \DTWO\ subgraph is indicated
	by solid lines; it has 5~connected~components, one has the
	bridge $ru$ and the other~4 are bridgeless. The dashed lines
	indicate the edges of $E(G)-E(\DTWO)$.}
        \label{sec5sub-a}
    \end{subfigure}
    \hspace*{\fill}
    \begin{subfigure}{0.49\textwidth}
        \centering
        \begin{tikzpicture}[scale=0.65]
            \begin{scope}[every node/.style={circle, fill=black, draw, inner sep=0pt,
            minimum size = 0.15cm
            }]
                
                \node[minimum size = 0.2cm, label={[label distance=3]45:$r$}] (r) at (0,0) {};
                \node[] (c01) at (-1,1) {};
                \node[] (c02) at (-2.25,0.75) {};
                \node[] (c03) at (-2.25,-0.75) {};
                \node[] (c04) at (-1,-1) {};
                \node[minimum size = 0.2cm, label={[label distance=3]135:$u$}] (u) at (2,0) {};
                \node[] (c05) at (3,-1) {};
                \node[] (c06) at (4,0) {};
                \node[] (c07) at (3,1) {};
                
                \node[] (a1) at (-3,2) {};
                \node[] (a2) at (-3,3) {};
                \node[] (a3) at (-2,2.5) {};
                
                \node[] (b1) at (0,3.5) {};
                \node[] (b2) at (-1,4.5) {};
                \node[] (b3) at (0,5.5) {};
                \node[] (b4) at (1,4.5) {};
                
                \node[] (c1) at (3,4) {};
                \node[] (c2) at (4,4) {};
                \node[] (c3) at (3.5,3) {};
                \node[] (c4) at (3,2) {};
                \node[] (c5) at (4,2) {};
                
                \node[] (d1) at (2,5) {};
                \node[] (d2) at (2,6) {};
                \node[] (d3) at (3,5.5) {};
                
            \end{scope}
    
            \begin{scope}[every edge/.style={draw=black}]
                \path[very thick] (r) edge node {} (c01);
                \path[very thick] (c01) edge node {} (c02);
                \path (c02) edge node {} (c03);
                \path (c03) edge node {} (c04);
                \path (c04) edge node {} (r);
                \path (r) edge node {} (u);
                \path (u) edge node {} (c05);
                \path (c05) edge node {} (c06);
                \path (c06) edge node {} (c07);
                \path[very thick] (c07) edge node {} (u);
                
                \path[very thick] (a1) edge node {} (a2);
                \path (a2) edge node {} (a3);
                \path (a3) edge node {} (a1);
                
                \path (b1) edge node {} (b2);
                \path (b2) edge node {} (b3);
                \path (b3) edge node {} (b4);
                \path[very thick] (b4) edge node {} (b1);
                
                \path (c1) edge node {} (c2);
                \path (c2) edge node {} (c3);
                \path[very thick] (c3) edge node {} (c1);
                \path (c3) edge node {} (c4);
                \path (c4) edge node {} (c5);
                \path[very thick] (c5) edge node {} (c3);
                
                \path (d1) edge node {} (d2);
                \path (d2) edge node {} (d3);
                \path (d3) edge node {} (d1);
                
                \path[very thick, dashed] (c02) edge node {} (a1);
                \path[very thick, dashed] (a2) edge node {} (b1);
                \path[very thick, dashed] (b4) edge node {} (c1);
                \path[very thick, dashed] (c5) edge node {} (c07);
            \end{scope}
    
            \begin{scope}[every node/.style={draw=none,rectangle}]
                \node (Rlabel) at (-1.5,0) {$R$};
                \node (Ulabel) at (3,0) {$U$};
            \end{scope}
        \end{tikzpicture}
	\caption{\centering The thick lines indicate a ``shortest''
	$u,r$ path of $G-ru$;
	this path has the minimum number of~non-\DTWO\ edges (4
	dashed lines); our goal is to ``pay'' for these 4~non-\DTWO\
	edges via the credits available in \DTWO.}
        \label{f:PT:eg1}
    \end{subfigure}
    \caption{Illustration of bridge~covering on a simple example.}
    \label{f:PT:eg1-main}
\end{figure}
}

Consider an example such that \DTWO\ has a connected~component $C_0$ that
has one bridge and two 2-ec blocks $R$ and $U$; let $ru$ denote the
unique bridge where $r$ is in $R$ and $u$ is in $U$.
(It can be seen that $ru$ is a zero-bridge, but we will not use this fact.)
Since $G$ is assumed to be 2NC, $G-{ru}$ is connected, hence, it
contains a $u,r$ path; let us pick a $u,r$ path $Y$ of $G-{ru}$ that
has only its prefix and suffix in common with $C_0$ and that has the
minimum number of non-\DTWO\ edges.
Our plan is to augment \DTWO\ by adding the edge~set $E(Y)-E(\DTWO)$,
thus ``covering" the bridge $ru$.
We may view this as an ``ear-augmentation step'' that adds
the ear\footnote{Throughout, we abuse the term \textit{ear}; although
$Y$ is not an ear, one may view the minimal subpath of $Y$ from $U$
to $R$, call it $Y'$, as an ear of $G$ w.r.t.\ $C_0$, i.e., $Y'$
is a path of $G$ that has both end~nodes in $C_0$ and has no internal
node in $C_0$.}
$Y$.
We have to pay for the non-\DTWO\ edges of $Y$ by using the credits
available in \DTWO.
Let us traverse $Y$ from $u$ to $r$, and each time we see a
non-\DTWO\ edge of $Y$, then we will pay for this edge.
For the sake of illustration, consider the example in
Figure~\ref{f:PT:eg1-main};
note that the $u,r$ path $Y$ (in Figure~\ref{f:PT:eg1-main}(b), on the right)
has $\ell=4$ non-\DTWO\ edges (indicated by dashed lines).
When we traverse $Y$ starting from $u$, then observe that each of the first
$(\ell-1)$ of these edges has its last node in a distinct 2ec-block
of \DTWO\ (moreover, none of these $(\ell-1)$ 2ec-blocks is in $C_0$).
We pay for these $(\ell-1)$ edges by borrowing one credit from the
credit of each of these $(\ell-1)$ 2ec-blocks.  We need one more credit to
pay for the last non-\DTWO\ edge of $Y$.  We get this credit from
the prefix of $Y$ between its start and its first non-\DTWO\ edge;
in particular, we borrow one credit from the credit of $U$.  Thus,
we can pay for all the non-\DTWO\ edges of $Y$.
Observe that by adding the edge~set $E(Y)-E(\DTWO)$ to \DTWO, we
have merged several 2ec-blocks (including $R$ and $U$) into a new
2ec-block.  We give the new 2ec-block (that contains $R$ and $U$)
the credit of $R$ as well as any unused credit of the other 2ec-blocks
incident to $Y$.

The general case of bridge~covering is more complicated.

\authremark{
For the sake of exposition, we may impose a direction on a path, cycle,
or ear (e.g., we traverse $Y$ from $u$ to $r$ in the discussion above).
Nevertheless, the input $G$ is an undirected graph, so (formally speaking) there is no
direction associated with paths or cycles of $G$.
}

\subsection{Post-processing \DTWO}

Immediately after computing \DTWO, we apply a post-processing step that
replaces some unit-edges of \DTWO\ by other unit-edges to obtain
another \DTWO\ that we denote by $\wh{\DTWO}$ that satisfies the
following key property:
\begin{quote}
\textit{
	Every pendant 2ec-block $B$ of~ $\wh{\DTWO}$ that is
	incident to a zero-bridge has $\cost(B)\ge 3$, and
	hence, has $\ge 2.25$~credits
	(see part~(2) of the credit~invariant below).
}
\end{quote}
In other words, if a 2ec-block of $\wh{\DTWO}$
has $\leq 2$ unit-edges, then either the 2ec-block is not pendant
(i.e., it is incident to no bridges or $\ge2$ bridges) or it is
pendant and is incident to a unit-bridge.

For any subgraph $G'$ of $G$, let $F_0(G')$ denote the set of
zero-bridges (of $G'$) that are incident to
pendant 2ec-blocks (of $G'$) of cost $\leq2$
(the notation $F_0$ is used only in this subsection);
moreover, let $\ncomp(G')$ denote the number of connected~components of $G'$.
Thus, the goal of the post-processing step is to
compute $\wh{\DTWO}$ such that $F_0(\wh{\DTWO})$ is empty.

The post-processing step is straightforward.
W.l.o.g.\ assume that the initial \DTWO\ contains all the zero-edges;
we start with this \DTWO\ and iterate the following step.
Let $\DTWO^{old}$ denote the \DTWO\ at the start of the iteration.
If $F_0(\DTWO^{old})$ is empty then we are done, we take $\DTWO^{old}$
to be $\wh{\DTWO}$.
Otherwise, we pick any zero-bridge $v_0v_1$ in $F_0(\DTWO^{old})$, and
we take $B$ to be a pendant 2ec-block with $\cost(B)=2$ that is incident to
$v_0v_1$; note that $B$ exists by the definition of $F_0$.
Now, observe that all of the edges of $B$ incident to $v_0$ have
cost~one, hence, it can be seen that $B$ has either 2~or~3 nodes;
then, since $G$ is 2NC and $\DTWO^{old}$ contains all the zero-edges,
$G$ has a unit-edge $e$
between $V(B)-\{v_0\}$ and $V-V(B)$; moreover,
$B$ has a unit-edge $f$ that is incident to $e$ and to $v_0$.
We replace $f$ by $e$.
The resulting subgraph is a 2-edge~cover of the
same cost; we denote it by $\DTWO^{new}$.

We claim that $F_0(\DTWO^{new}) \subseteq F_0(\DTWO^{old})$,
and moreover, 
\[
|F_0(\DTWO^{new})| + \ncomp(\DTWO^{new}) ~<~
	|F_0(\DTWO^{old})| + \ncomp(\DTWO^{old}).
\]

To see this,
first suppose that the end~nodes of $e$ are in two different
connected~components of $\DTWO^{old}$; 
then it can be seen that $F_0(\DTWO^{new}) \subseteq F_0(\DTWO^{old})$,
and hence, the claimed inequality holds;
otherwise, (the end~nodes of $e$ are in the same connected~component
of $\DTWO^{old}$) $v_0v_1$ is not a bridge of
$\DTWO^{new}$, therefore, $F_0(\DTWO^{new})$ is a proper
subset of $F_0(\DTWO^{old})$, and hence, the claimed inequality holds.

Thus, after $O(n)$ iterations, we find $\wh{\DTWO}$ that satisfies the
key property.  See Figure~\ref{f:postproc} for illustrations.

\begin{proposition}
There is a polynomial-time algorithm that finds a 
2-edge~cover $\wh{\DTWO}$
(of cost $\dtwocost(G)$)
such that $F_0(\wh{\DTWO})$ is empty;
thus, $\wh{\DTWO}$ satisfies the key property.
\end{proposition}

{
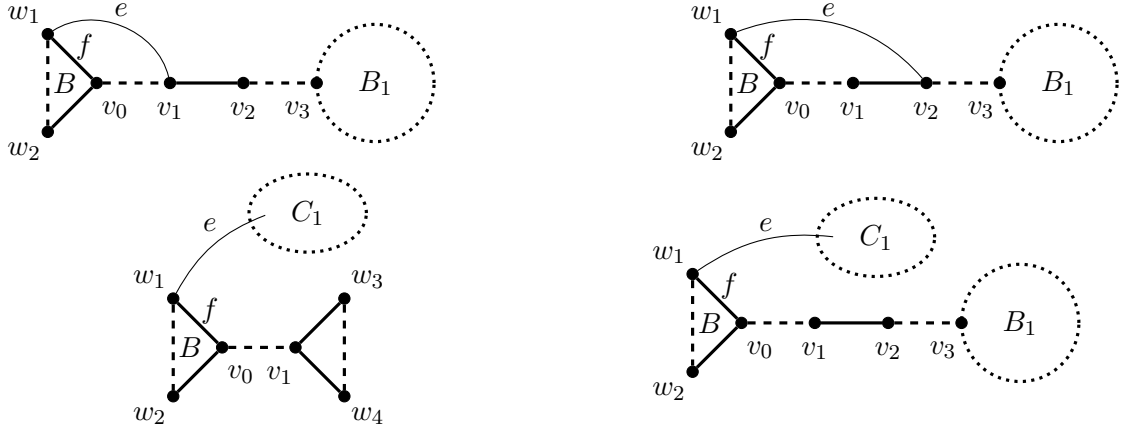
\begin{figure}[ht]
    \centering
    \begin{subfigure}{0.45\textwidth}
        \centering
        \begin{tikzpicture}[scale=0.65]
            \begin{scope}[every node/.style={circle, fill=black, draw, inner sep=0pt,
            minimum size = 0.15cm
            }]
                \node[label={[label distance=4]295:$v_0$}] (v0) at (0,0) {};
                \node[label={[label distance=2]150:$w_1$}] (w1) at (-1,1) {};
                \node[label={[label distance=2]210:$w_2$}] (w2) at (-1,-1) {};
                \node[label={[label distance=2]270:$v_1$}] (v1) at (1.5,0) {};
                \node[label={[label distance=2]270:$v_2$}] (v2) at (3,0) {};
                \node[label={[label distance=4]245:$v_3$}] (v3) at (4.5,0) {};
                \draw[very thick,dotted] (5.7,0) circle (1.2cm);
            \end{scope}
    
            \begin{scope}[every edge/.style={draw=black}]
                \path[very thick] (v0) edge node {} (w1);
                \path[very thick] (v0) edge node {} (w2);
                \path[very thick,dashed] (w1) edge node {} (w2);
                \path[very thick,dashed] (v0) edge node {} (v1);
                \path[very thick] (v1) edge node {} (v2);
                \path[very thick,dashed] (v2) edge node {} (v3);
                
                \path (w1) edge[bend left=55] node {} (v1);
            \end{scope}
    
            \begin{scope}[every node/.style={draw=none,rectangle}]
                \node (Blabel) at (-0.65,0) {$B$};
                \node (B1label) at (5.7,0) {$B_1$};
                \node (elabel) at (0.5,1.5) {$e$};
                \node (flabel) at (-0.25,0.75) {$f$};
            \end{scope}
        \end{tikzpicture}
        \label{f:postproc-a}
    \end{subfigure}
    \hspace*{\fill}
    \begin{subfigure}{0.45\textwidth}
        \centering
        \begin{tikzpicture}[scale=0.65]
            \begin{scope}[every node/.style={circle, fill=black, draw, inner sep=0pt,
            minimum size = 0.15cm
            }]
                \node[label={[label distance=4]295:$v_0$}] (v0) at (0,0) {};
                \node[label={[label distance=2]150:$w_1$}] (w1) at (-1,1) {};
                \node[label={[label distance=2]210:$w_2$}] (w2) at (-1,-1) {};
                \node[label={[label distance=2]270:$v_1$}] (v1) at (1.5,0) {};
                \node[label={[label distance=2]270:$v_2$}] (v2) at (3,0) {};
                \node[label={[label distance=4]245:$v_3$}] (v3) at (4.5,0) {};
                \draw[very thick,dotted] (5.7,0) circle (1.2cm);
            \end{scope}
    
            \begin{scope}[every edge/.style={draw=black}]
                \path[very thick] (v0) edge node {} (w1);
                \path[very thick] (v0) edge node {} (w2);
                \path[very thick,dashed] (w1) edge node {} (w2);
                \path[very thick,dashed] (v0) edge node {} (v1);
                \path[very thick] (v1) edge node {} (v2);
                \path[very thick,dashed] (v2) edge node {} (v3);
                
                \path (w1) edge[bend left=35] node {} (v2);
            \end{scope}
    
            \begin{scope}[every node/.style={draw=none,rectangle}]
                \node (Blabel) at (-0.65,0) {$B$};
                \node (B1label) at (5.7,0) {$B_1$};
                \node (elabel) at (1,1.5) {$e$};
                \node (flabel) at (-0.25,0.75) {$f$};
            \end{scope}
        \end{tikzpicture}
        \label{f:postproc-b}
    \end{subfigure}
    \hspace*{\fill}
    \begin{subfigure}{0.45\textwidth}
        \centering
        \begin{tikzpicture}[scale=0.65]
            \begin{scope}[every node/.style={circle, fill=black, draw, inner sep=0pt,
            minimum size = 0.15cm
            }]
                \node[label={[label distance=4]295:$v_0$}] (v0) at (0,0) {};
                \node[label={[label distance=2]150:$w_1$}] (w1) at (-1,1) {};
                \node[label={[label distance=2]210:$w_2$}] (w2) at (-1,-1) {};
                \node[label={[label distance=4]255:$v_1$}] (v1) at (1.5,0) {};
                \node[label={[label distance=2]30:$w_3$}] (w3) at (2.5,1) {};
                \node[label={[label distance=2]330:$w_4$}] (w4) at (2.5,-1) {};
                \draw[very thick,dotted] (1.75,2.75) ellipse (1.2cm and 0.8cm);
                \node[draw=none,fill=none] (e0) at (1,2.75) {};
            \end{scope}
    
            \begin{scope}[every edge/.style={draw=black}]
                \path[very thick] (v0) edge node {} (w1);
                \path[very thick] (v0) edge node {} (w2);
                \path[very thick,dashed] (w1) edge node {} (w2);
                \path[very thick,dashed] (v0) edge node {} (v1);
                \path[very thick] (v1) edge node {} (w3);
                \path[very thick] (v1) edge node {} (w4);
                \path[very thick,dashed] (w3) edge node {} (w4);
                \path (w1) edge[bend left=20] node {} (e0);
            \end{scope}
    
            \begin{scope}[every node/.style={draw=none,rectangle}]
                \node (Blabel) at (-0.65,0) {$B$};
                \node (elabel) at (-0.25,2.5) {$e$};
                \node (flabel) at (-0.25,0.75) {$f$};
                \node (c1label) at (1.75,2.75) {$C_1$};
            \end{scope}
        \end{tikzpicture}
        \label{f:postproc-c}
    \end{subfigure}
    \hspace*{\fill}
    \begin{subfigure}{0.45\textwidth}
        \centering
        \begin{tikzpicture}[scale=0.65]
            \begin{scope}[every node/.style={circle, fill=black, draw, inner sep=0pt,
            minimum size = 0.15cm
            }]
                \node[label={[label distance=4]295:$v_0$}] (v0) at (0,0) {};
                \node[label={[label distance=2]150:$w_1$}] (w1) at (-1,1) {};
                \node[label={[label distance=2]210:$w_2$}] (w2) at (-1,-1) {};
                \node[label={[label distance=2]270:$v_1$}] (v1) at (1.5,0) {};
                \node[label={[label distance=2]270:$v_2$}] (v2) at (3,0) {};
                \node[label={[label distance=4]245:$v_3$}] (v3) at (4.5,0) {};
                \draw[very thick,dotted] (5.7,0) circle (1.2cm);
                \draw[very thick,dotted] (2.75,1.75) ellipse (1.2cm and 0.8cm);
                \node[draw=none,fill=none] (e0) at (2,1.75) {};
            \end{scope}
    
            \begin{scope}[every edge/.style={draw=black}]
                \path[very thick] (v0) edge node {} (w1);
                \path[very thick] (v0) edge node {} (w2);
                \path[very thick,dashed] (w1) edge node {} (w2);
                \path[very thick,dashed] (v0) edge node {} (v1);
                \path[very thick] (v1) edge node {} (v2);
                \path[very thick,dashed] (v2) edge node {} (v3);
                \path (w1) edge[bend left=20] node {} (e0);
            \end{scope}
    
            \begin{scope}[every node/.style={draw=none,rectangle}]
                \node (Blabel) at (-0.65,0) {$B$};
                \node (B1label) at (5.7,0) {$B_1$};
                \node (elabel) at (0.5,2) {$e$};
                \node (flabel) at (-0.25,0.75) {$f$};
                \node (c1label) at (2.75,1.75) {$C_1$};
            \end{scope}
        \end{tikzpicture}
        \label{f:postproc-d}
    \end{subfigure}
    \hspace*{\fill}
    \caption{\protect\centering Illustrations of some cases that could
    arise during post-processing. Edges of \DTWO\ are indicated by
    thick lines. Edges of cost~zero and~one are indicated by dashed and
    solid lines, respectively.}
    \label{f:postproc}
\end{figure}
}
\subsection{ \label{s:credit-invariant} Credit invariant and charging lemma}

Let $H=(V,F)$ denote the current graph of picked edges;
thus, at the start, $H$ is the same as $\wh{\DTWO}$,
and we may assume that $H$ has one or more bridges
(otherwise, bridge~covering is trivial).
By an \textit{original} 2ec-block $B$ of $H$ we mean a 2ec-block (of
the current $H$) such that $B$ is also a 2ec-block of $\wh{\DTWO}$
(i.e., the set of edges incident to $V(B)$ is the same
in both $\wh{\DTWO}$ and the current $H$).
By a \textit{new} 2ec-block of $H$ we mean a 2ec-block (of the current
$H$) that is not an original 2ec-block.
Similarly, we define an original/new connected~component of $H$.

For any subgraph $H'$ of $H$, we use $\credits(H')$ to denote the
sum of the working~credits of the unit-edges of $H'$.
At the start of bridge~covering, $\credits(H)$ is equal to the sum
of the initial~credits of the unit-edges of $\wh{\DTWO}$ minus the number
of unit-edges of $\wh{\DTWO}$.

We call a node $v$ of $H$ a \textit{white} node if $v$ belongs to
a 2ec-block of $H$, otherwise, we call $v$ a \textit{black} node.

\begin{fact}
Suppose that $H$ has one or more black nodes.
Let $v$ be a black node of $H$.
Then all edges of $H$ incident to $v$ are bridges of $H$, and
$v$ is incident to $\ge2$ bridges of $H$.
Every maximal path of $H$ that starts with $v$ contains a white node.
\end{fact}

\authremark{
Let ${H'}$ denote the graph obtained from $H$ by contracting
each 2ec-block;
thus, each 2ec-block of $H$ maps to a ``contracted'' white node of ${H'}$,
each black node of $H$ maps to a black node of ${H'}$,
and each bridge of $H$ maps to a bridge of ${H'}$.
Clearly, ${H'}$ is a forest; it may have isolated nodes
(these correspond to 2EC~connected~components of $H$).
Clearly, ${H'}$ has $\ge2$ edges incident to each black node.
}

By a \textit{b-path} of $H$ we mean a path consisting
of bridges that starts and ends with arbitrary nodes (white or black)
such that all internal nodes are black nodes.
We say that two 2ec-blocks of a connected~component of $H$ are
\textit{b-adjacent} if there exists a b-path whose terminal nodes are
in these two 2ec-blocks.

Initially, our algorithm picks a connected~component $C_0$ of $H=\wh{\DTWO}$
that has one or more bridges.
If $C_0$ has any pendant 2ec-block that is incident to a zero-bridge,
then the algorithm picks such a 2ec-block and designates it as the \textit{root 2ec-block} $R$
(in this section, $R$ always denotes the root 2ec-block),
otherwise, the algorithm picks any pendant 2ec-block of $C_0$ and
designates it as the {root 2ec-block} $R$.
Also, the algorithm picks the unique bridge of $C_0$ incident to $R$;
we denote this bridge by $ru$, where $r$ is in $R$.
If $ru$ is a zero-bridge,
then since $C_0$ is an original connected~component,
$R$ has $\ge2.25$ credits (by the key~property of $\wh{\DTWO}$);
otherwise, $R$ has $\ge1.5$ credits.
Immediately after designating $R$, we ensure that $R$ has $\ge 2$ credits.
If $R$ is short of credits (i.e., it has only 1.5~credits),
then we borrow 0.5~credits from a 2ec-block of $C_0$ that is b-adjacent to $R$
(in this case, observe that every pendant 2ec-block of $C_0$ is incident to a unit-bridge);
see credit~invariant~(3) below.

Each iteration of bridge~covering maintains the following invariant,
i.e., if the invariant holds at the start of an iteration, then
it holds at the end of that iteration.
(At the end of this section, we argue that this invariant is preserved
at the end of every iteration, see Proposition~\ref{propo:CI}.
Note that the invariant could be temporarily violated within an iteration,
due to various updates.)

\medskip
\noi
\textbf{Credit invariant}:
\begin{enumerate}[(1)]
\item
\textit{
Each unit-edge of $H$ has a unit of retained~credit
(recall that retained~credits are distinct from working~credits, and
the term ``credit'' means working~credit).
Each original 2EC~connected~component has $\ge1.5$ credits, and
each new 2EC~connected~component has $\ge2$ credits.
Moreover, each unit-bridge of $H$ has $0.75$~credits.
}

\item
\textit{
Within each original connected~component of $H$,
each pendant 2ec-block that is incident to a zero-bridge
has $\ge2.25$ credits.
}

\item
\textit{
Suppose that the root 2ec-block $R$ is well defined.
Then either $R$ is incident to a zero-bridge and has $\ge2.25$ credits,
or $R$ has $\ge2$ credits.
Moreover, each 2ec-block that is b-adjacent to $R$ has $\ge1$ credits,
and every other 2ec-block of $H$ has $\ge1.5$ credits.
}
\end{enumerate}

In an arbitrary iteration of the algorithm,
either $R$ is a pendant 2ec-block of an original connected~component
$C_0$ of $H$, or $R$ has been designated as the root 2ec-block by
the previous iteration and there exists a bridge (of $H$) incident to $R$.
In the former case, (as discussed above) the algorithm chooses
$ru$ to be the unique bridge of $C_0$ incident to $R$;
in the latter case, 
if there exists a unit-bridge incident to $R$,
then the algorithm chooses $ru$ to be such a bridge,
otherwise, the algorithm chooses $ru$ to be any zero-bridge incident to $R$.

When we remove the edge ${ru}$ from $C_0$ then we get two connected~components;
let us denote them by $C_0^{r}$ (it contains $r$ but not $u$) and
$C_0^{u}$ (it contains $u$ but not $r$).
Recall that $G-ru$ has a path between $C_0^{u}$ and $C_0^{r}$.
Let $P$ be such a path of $G-ru$ that
starts at some node $\hat{a}$ of $C_0^{u}$,
ends at some node $\hat{z}$ of $C_0^{r}$,
has no internal nodes in $C_0$,
and (subject to the above) minimizes $|E(P)-E(H)|$.
Note that there could be many choices for the nodes $\hat{a}$ and $\hat{z}$;
although the choice of these nodes is important for our analysis
(see Lemma~\ref{lem:findP} and its proof),
the next lemma (Lemma~\ref{lem:thick}) and its proof apply
for all valid choices of these two nodes.
See Figure~\ref{f:PT:intro}.

{
\begin{figure}[ht]
    \centering
    \begin{tikzpicture}[scale=0.5]
    \begin{scope}[every node/.style={circle, fill=black, draw, inner sep=0pt,
    minimum size = 0.2cm
    }]
        \node[fill=none, minimum size = 1.5cm, thick] (R) at (0,0) {};
        \node (r1) at (-2,0) {};
        \node[minimum size = 0.2cm, label={[label distance=2]45:$\hat{z}$}] (r2) at (-3,0) {};
        \node[draw=none, fill=none] (r3) at (-4,0) {};
        
        \node[fill=none, minimum size = 0.5cm, thick] (x1) at (-1,2) {};
        \node (x2) at (-0.5,3) {};
        \node (x3) at (0,4) {};
        \node[fill=none, minimum size = 0.5cm, thick] (x4) at (0.5,5) {};
        
        \node[fill=none, minimum size = 0.5cm, thick] (y1) at (3,2) {};
        \node (y2) at (3.5,3) {};
        \node (y3) at (4,4) {};
        \node[fill=none, minimum size = 0.5cm, thick] (y4) at (3.5,5) {};
        \node[fill=none, minimum size = 0.5cm, thick] (y5) at (5,4.5) {};

       \node[minimum size = 0.2cm, label={[label distance=3]45:$r$}] (rlabel) at (1.25,0) {};
       \node[minimum size = 0.2cm, label={[label distance=3]90:$u$}] (v0) at (3.25,0) {};
       \node[minimum size = 0.2cm, label={[label distance=3]135:$\hat{a}$}] (v1) at (5.5,0) {};
        \node (v2) at (7.75,0) {};
        \node (v3) at (10,0) {};
        \node[draw=none, fill=none] (v4) at (11,0) {};

    \end{scope}
    
    \begin{scope}[every edge/.style={draw=black,thick}]
    \path (R) edge node {} (r1);
    \path (r1) edge node {} (r2);
    \path[dotted] (r2) edge node {} (r3);
    
    \path (R) edge node {} (v0);
    \path (v0) edge node {} (v1);
    \path (v1) edge node {} (v2);
    \path (v2) edge node {} (v3);
    \path[dotted] (v3) edge node {} (v4);
    
    \path (x1) edge node {} (x2);
    \path (x3) edge node {} (x4);
    
    \path (y1) edge node {} (y2);
    \path (y3) edge node {} (y4);
    \path (y3) edge node {} (y5);

    \path[very thick, dashed, bend left=60] (r2) edge node {} (x2);
    \path[very thick, dashed] (x2) edge node {} (x3);
    \path[very thick, dashed] (x3) edge node {} (y2);
    \path[very thick, dashed] (y2) edge node {} (y3);
    \path[very thick, dashed, bend left=30] (y3) edge node {} (v1);

    \end{scope}
    
    \begin{scope}[every node/.style={draw=none,rectangle}]
        \node (Rlabel) at (0,0) {$R$};
        \node (Plabel) at (-3,2.5) {$P$};
        \node (Qlabel) at (3.25,-1) {$Q$};
    \end{scope}
    \end{tikzpicture}
    \caption{Illustration of plan for ``bridge~covering.''
    Large circles indicate 2ec-blocks.
    Dashed lines indicate $P$.
    $Q$ is the path of $C_0$ between $r$ and $\hat{a}$, where $\hat{a}$
    is the unique node of $P$ in $C_0^{u}$.}
    \label{f:PT:intro}
\end{figure}
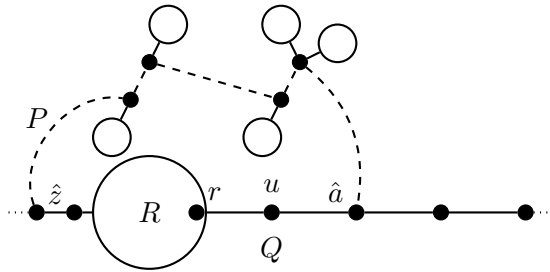
}

Let $H^{new}$ denote $H\cup(E(P)-E(H))$ and let $B^{new}$ denote the
2ec-block of $H^{new}$ that contains $\hat{a}$, $\hat{z}$, and $R$.
We designate $B^{new}$ as the root 2ec-block of $H^{new}$, provided
$B^{new}$ is incident to a bridge of $H^{new}$.

Each iteration may be viewed as an ear-augmentation step
that adds either one open ear to $C_0$ or two open ears to $C_0$, e.g.,
the path $P$ may be viewed as an open~ear of $G$ w.r.t.\ $C_0$.
On the other hand,
our charging scheme (for paying for the edges of $E(P)-E(H)$)
views each iteration as adding an ear w.r.t.\ $R$,
that is, we take the ear to be the union of three paths, namely,
an $r,\hat{a}$ path of $C_0$ that contains $ru$,
the path $P$,
and a path of $C_0$ between $\hat{z}$ and $R$
(we mention that our charging scheme also uses
the credits available in the first of these three paths).
\label{page:r-ear}
To avoid confusion, we call these the $R$-ears, and
unless mentioned otherwise, an ``ear'' means an ear w.r.t.\ $C_0$.

Let us outline our credit scheme (for bridge~covering) and its interaction with
an arbitrary iteration of the algorithm.
Let $\hat{a}\in{V(C_0^{u})}$, $\hat{z}\in{V(C_0^{r})}$, and $P$ be as described above.
Let $C_1,\dots,C_k$ denote the connected components of $H$ that
contain (at least) one internal node of $P$
(thus, $C_0\not=C_i,\;\forall i=1,\dots,k$).
Lemma~\ref{lem:thick} and its proof, see below, explain how the credits of $C_1,\dots,C_k$
can be redistributed such that each of $C_1,\dots,C_k$ releases one unit of credit
such that the credit invariant holds again at the end of the iteration.
Thus, we can buy all-but-one of the unit-edges of $E(P)-E(H)$ using the
credits released by $C_1,\dots,C_k$.
In order to establish the credit invariant at the end of the iteration,
we have to find another unit of credit (to pay for one unit-edge of $E(P)-E(H)$),
and, in addition, we have to ensure that credit~invariant~(3) is maintained
(e.g., we may have to find another 0.25 units of credit).
We defer this issue to Section~\ref{s:bridge-cover-last};
in fact, this is the most intricate part of bridge~covering.

{
\begin{lemma} \label{lem:thick}
Let $H,C_0,C_0^{u},C_0^{r}$ be as stated above, and suppose that $H$ satisfies the
credit invariant.
Let $P$ be an open~ear w.r.t.\ $C_0$ with end~nodes $\hat{a},\hat{z},$
where $\hat{a}$ is in $C_0^{u}$ and $\hat{z}$ is in $C_0^{r}$,
such that each connected component of $H$ that contains an internal
node of $P$ is incident to exactly two edges of $E(P)-E(H)$
(i.e., $E(P)-E(H)$ forms a path in the auxiliary graph obtained
from $G$ by contracting all connected components of $H$ other than $C_0$).
Let $C_1,\dots,C_k$ denote the connected components of $H$ that
contain (at least) one internal node of $P$
(thus, $C_0\not=C_i,\;\forall i=1,\dots,k$).
Then the credits of $C_1,\dots,C_k$ can be redistributed such that
each of $C_1,\dots,C_k$ releases one unit of credit such that the
credit invariant holds again at the end of the iteration.

More precisely, if $C_i$ ($i\in\{1,\dots,k\}$) has a 2ec-block that
contains a node of $P$, then we take one unit from the credit of
one such 2ec-block $B$.
(At the end of the iteration, $B$ is ``merged into'' $B^{new}$,
hence, the credit invariant holds again.)
Otherwise, we take 0.5 units from the credits of two distinct
2ec-blocks $B_1,B_2$ of $C_i$ that are b-adjacent to $P$.
(At the end of the iteration, $B_1$ and $B_2$ are b-adjacent to $B^{new}$,
hence, the credit invariant holds again.)
\end{lemma}
\medskip

\begin{proof}
Our goal is to show that we can pay via the credits of $C_1,\dots,C_k$ for
all-but-one of the edges of $E(P)-E(H)$, while ensuring that the
credit invariant holds again for $H^{new}$.
When we traverse $P$ from $\hat{a}$ to $\hat{z}$, observe that each
of the edges of $E(P)-E(H)$, except the last such edge, has its last
end~node in a distinct connected~component $C_i\not=C_0$ of $H$; we use the
credits available in that connected~component to pay the cost of the edge.
The rest of the proof shows how we can obtain one~unit from the
credit of the relevant connected~component while ensuring that
the credit invariant continues to hold for $H^{new}$.
(Proposition~\ref{propo:CI} below shows that the credit invariant is
preserved at the end of each iteration.)

{
\begin{figure}[htb]
    \centering
    \begin{tikzpicture}[scale=0.6]
    \begin{scope}[every node/.style={circle, fill=black, draw, inner sep=0pt,
    minimum size = 0.2cm
    }]
        \node[fill=none, minimum size = 1.5cm, thick] (bs) at (0,0) {};
        \node[fill=none, minimum size = 1.5cm, thick] (bt) at (10.5,0) {};
        \node[fill=none, minimum size = 1.5cm, thick] (2ec) at (11.5,5) {};
        \node[fill=none, thick] (r') at (4.5,6) {};
        \node[fill=none, thick] (w') at (6.5,0) {};
        \node (v1) at (2,2) {};
        \node (v2) at (3,3) {};
        \node (v3) at (2,4) {};
        \node[draw=none, fill=none] (v4) at (1,5) {};
        \node[label={[label distance=3]270:$s_0$}] (s) at (4.5,3) {};
        \node (v6) at (6,3) {};
        \node (v7) at (6,4.5) {};
        \node[draw=none, fill=none] (v8) at (6,6) {};
        \node (v9) at (7.5,3) {};
        \node[label={[label distance=3]180:$t_0$}] (t) at (8.5,2) {};
        \node (v11) at (8.5,4) {};
        \node[draw=none, fill=none] (v12) at (9.5,5) {};
        \node (v13) at (9.5,3) {};
    \end{scope}
    
    \begin{scope}[every edge/.style={draw=black,thick}]
    \path[<-] (bs) edge node {} (v1);
    \path[<-] (v1) edge node {} (v2);
    \path (v2) edge node {} (v3);
    \path[dotted] (v3) edge node {} (v4);
    \path[<-] (v2) edge node {} (s);
    \path[very thick, dashed] (r') edge node {} (s);
    \path[very thick, dashed] (s) edge node {} (v6);
    \path (v6) edge node {} (v7);
    \path[dotted] (v7) edge node {} (v8);
    \path[very thick, dashed] (v6) edge node {} (v9);
    \path[very thick, dashed] (v9) edge node {} (t);
    \path (v9) edge node {} (v11);
    \path[dotted] (v11) edge node {} (v12);
    \path (t) edge node {} (v13);
    \path (v13) edge node {} (2ec);
    \path[->] (t) edge node {} (bt);
    \path[very thick, dashed] (t) edge node {} (w');
    \end{scope}
    
    \begin{scope}[every node/.style={draw=none,rectangle}]
        \node (bs) at (0,0) {$B({s_1})$};
        \node (bt) at (10.5,0) {$B({t_1})$};
	\node[label={[label distance=5]180:$f$}] (fpc) at (5.0,4.5) {};
	\node[label={[label distance=5]180:$e$}] (epc) at (7.8,1) {};
    \end{scope}
    \end{tikzpicture}
    \caption{Illustration of the notation in Lemma~\ref{lem:thick}.
    A connected~component $C\not=C_0$ of $H$ is shown ($C$ does not contain $R$).
    The large circles indicate 2ec-blocks of $C$.
    The path $P$ is indicated by dashed lines;
    $P$ contains 3~edges of $C$.
    The arrows indicate maximal b-paths ending at 2ec-blocks.}
    \label{f:PT:lemma}
\end{figure}
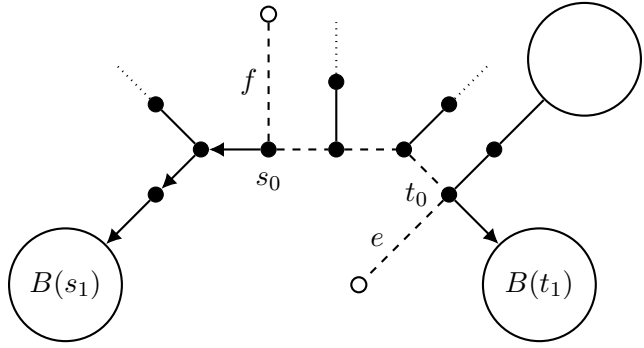
}

Consider any original connected~component $C\not=C_0$ of $H$ that
contains one of the internal nodes of $P$
(thus $C$ is one of $C_1,\dots,C_k$).
See Figure~\ref{f:PT:lemma}.
By our choice of $P$, there is a unique edge of $P$ that ``enters" $C$
and there is a unique edge of $P$ that ``exits" $C$,
i.e., there is a unique edge $f$ with one end~node in $C$ and the
other end~node in the subpath of $P$ between $\hat{a}$ and $C$, and similarly,
      there is a unique edge $e$ with one end~node in $C$ and the
other end~node in the subpath of $P$ between $C$ and $\hat{z}$.

Let $s_0$ denote the end~node of $f$ in $C$, and
let $t_0$ denote the end~node of $e$ in $C$.
Possibly, $s_0=t_0$.
Let $P(s_0,t_0)$ denote the $s_0,t_0$~sub-path of $P$.
Clearly, $P(s_0,t_0)$ is contained in $H$.

First, suppose that $P(s_0,t_0)$ contains a white node;
there is a 2ec-block of $C$, call it $B$, that contains this white node;
then we take one~unit from the credit of $B$
and use that to pay for the edge $f$.
(Recall that $B^{new}$ is designated as the root 2ec-block of
$H^{new}$, and note that $B^{new}$ contains both $s_0$ and $t_0$;
moreover, the 2ec-block $B$ (of $H$) is also contained in $B^{new}$;
hence, at the end of the iteration,
it can be seen that the credit~invariant is maintained in $H^{new}$
although we borrowed one credit from $B$;
see the proof of Proposition~\ref{propo:CI}.)

Otherwise, both $s_0$ and $t_0$ are black nodes. Then we resort to a more complex
scheme.  Let $C(s_0,s_1)$ be any maximal b-path of $C - E(P(s_0,t_0))$
that starts with $s_0$ and ends with a white node $s_1$, and let
$B(s_1)$ denote the 2ec-block that contains the white node $s_1$;
note that $B(s_1)$ has no nodes in common with $B^{new}$.
Similarly, let $B(t_1)$ denote a 2ec-block that contains the terminal
white node $t_1$ (where, $t_1\not=s_1$) of a maximal b-path of
$C - E(P(s_0,t_0))$ that starts with $t_0$;
it can be seen that $B(t_1)$ has no nodes in common with $B^{new}$.\footnote{
If $s_0=t_0$, then note that $s_0$ is incident to $\ge2$ bridges
of $C = C - E(P(s_0,t_0))$, hence, we can ensure that $t_1\not=s_1$.
}
We take 0.5 credits from each of $B(s_1)$ and $B(t_1)$
and use that to pay for the edge $f$.
(In $H^{new}$, observe that both $s_0$ and $t_0$ are contained in the
root 2ec-block $B^{new}$;
moreover, $B(s_1)$ and $B(t_1)$ are 2ec-blocks,
and moreover, both these 2ec-blocks are b-adjacent to $B^{new}$;
hence, at the end of the iteration, the credit~invariant is maintained in $H^{new}$
although we borrowed 0.5 credits from each of $B(s_1)$ and $B(t_1)$;
see the proof of Proposition~\ref{propo:CI}.)

Finally, consider any new connected~component $C\not=C_0$ of $H$ that
contains one of the internal nodes of $P$.  Our algorithm ensures that
every new connected~component of $H$, except for $C_0$, is 2EC, hence,
$C$ is 2EC. We take one unit from the credit of $C$ and use that to pay
for the unique edge of $P$ that ``enters'' $C$. (At the end of the iteration, the credit invariant
is maintained in $H^{new}$ because $C$ is contained in $B^{new}$, the
designated root 2ec-block of $H^{new}$;
see the proof of Proposition~\ref{propo:CI}.)
\end{proof}
}

\subsection{
	\label{s:bridge-cover-last} Algorithm and analysis for bridge~covering}

We present the algorithm and analysis of bridge~covering
based on Lemma~\ref{lem:thick}.

Recall that
$H$ satisfies the credit invariant initially, and that
Lemma~\ref{lem:thick} allows us to pay for all-but-one of the
unit-edges added by a single ear-augmentation.
Our goal is to charge the remaining cost (of one) to a prefix of the $R$-ear that we
denote by $Q$; $Q$ is the maximal path of $C_0$ contained in the $R$-ear
and starting with the edge $ru$.  Let $\pstart$ denote the other end
node of $Q$ (thus, when we traverse the edges (and nodes) of the $R$-ear
starting with the edge $ru$, then $\pstart$ is the first node incident
to an edge of the $R$-ear in $E(G)-E(H)$).
Since our goal is to collect as much credit as possible from $Q$
we choose the $R$-ear such that
either (i)~$Q$ has a white node $w$ ($w\not=r$)
or (ii)~$Q$ has no white nodes (other than $r$),
$Q$ has the maximum cost possible, and subject to this,
$Q$ has the maximum number of bridges possible.
In other words, we choose (the 3-tuple) $P$, $\pstart$, $\pend$
such that $P$ is a path of $G-ru$ with one end~node $\pstart$ in $C_0^{u}$
and the other end~node $\pend$ in $C_0^{r}$,
none of the internal nodes of $P$ is in $C_0$,
the associated prefix $Q$ satisfies condition~(i) or~(ii) (stated above), and
(subject to all the above requirements)
$P$ has the minimum number of edges from $E(G)-E(H)$.
We can easily compute (the 3-tuple) $P$, $\pstart$, $\pend$ in polynomial time
via standard methods from graph algorithms; this is discussed in
the next lemma.

\begin{lemma} \label{lem:findP}
$P$, $\pstart$, $\pend$ satisfying the
requirements stated above can be computed in polynomial time.
\end{lemma}

\begin{proof}
{
For each node $v\in C_0^{u}$, we define $\gamma(v)$ as follows:
$\gamma(v)=\infty$
if every path of $C_0$ between $v$ and $r$ contains a white node $w$, $w\not=r$;
otherwise, $\gamma(v)=|E(C_0(v,r))| + n\cdot \cost(C_0(v,r))$,
where $C_0(v,r)$ denotes the unique b-path of $C_0$ between $v$ and $r$.
(Informally speaking, $\gamma(v)$ assigns a ``rank'' to each node 
$v$ of $C_0^{u}$;
if every $v,r$ path of $C_0$ contains two or more white nodes,
then the rank is $\infty$, otherwise,
the rank is determined by the unique b-path of $C_0$ between $v$ and $r$,
and we rank according to the 2-tuple consisting of
the cost and the number of bridges of the relevant path.)

Then, we construct the following weighted directed graph:
the directed graph has two oppositely oriented edges for
each edge of $G-V(C_0)$,
it has an edge oriented out of $V(C_0^{r})$ for each edge of $G-ru$
in the cut $\delta_G(V(C_0^{r}))$,
and it has an edge oriented into $V(C_0^{u})$ for each edge of $G-ru$
in the cut $\delta_G(V(C_0^{u}))$
(there are no oriented edges corresponding to other edges of $G$).
We assign weights of zero to the oriented edges associated with the
edges of $H$, and weights of one to the other oriented edges
(associated with the edges of $E(G)-E(H)$).
Then, we apply a reachability computation,
taking all the nodes in $C_0^{r}$ to be the sources. 
We claim that a node $v\in C_0^{u}$ is reachable from $C_0^{r}$ (in
the directed graph) iff $G-ru$ has a path between $C_0^{r}$ and $v$
such that no internal node of the path is in $C_0$.

Thus, we can find $\pstart$ and $\pend$ that satisfy the requirements:
we choose $\pstart$ to be a node $v\in C_0^{u}$
that is reachable from $C_0^{r}$ (in the directed graph) and
that has the maximum $\gamma()$ value, and then
we choose $\pend$ to be a node in $C_0^{r}$ such that
the directed graph has a path from this node to $\pstart$.
Then, we take $P$ to be a shortest $\pend,\pstart$ path in the
(weighted) directed graph.
}
\end{proof}

An outline of the bridge~covering~step follows.

\medskip
\noi
\fbox{ \begin{minipage}{\textwidth}

\textbf{Bridge covering (outline)}:
\begin{itemize}
\item[(0)] compute \DTWO, then post-process \DTWO\ to obtain $\wh{\DTWO}$, and let $H=\wh{\DTWO}$;

\item[(1)] pick any (original) connected~component $C_0$ of $H$ that has a bridge;
	if possible, pick a pendant 2ec-block of $C_0$ that is incident to
a zero-bridge and designate it as the root 2ec-block $R$;
	otherwise, (every pendant 2ec-block of $C_0$ is incident to a
unit-bridge) pick any pendant 2ec-block of $C_0$ and designate it
as the root 2ec-block $R$;

\item[(2)] \textbf{repeat}
	\begin{quote}
	if possible, pick a unit-bridge of $C_0$ incident to $R$,
	otherwise, pick any zero-bridge of $C_0$ incident to $R$,
	and denote the picked bridge by $ru$;
	\\
	apply one ear-augmentation step (add one or two ears, see Cases~1--3 below) to cover a sub-path of
	bridges starting with $ru$, and let $B^{new}$ denote the resulting 2ec-block
	that contains $R$ and $ru$;
	\\
	let $R~~:=~~B^{new}$;
	\end{quote}
	\textbf{until} $R = B^{new}$ has no incident bridges of $H$;

\item[(3)] \textbf{stop} if $H$ has no bridges, otherwise, \textbf{go to} (1).
\end{itemize}
\end{minipage}
}
\medskip

Recall from Section~\ref{s:credit-invariant} that we have to
establish the credit invariant at the end of the iteration, hence,
we have to find another unit of credit (to pay for one unit-edge of $E(P)-E(H)$),
and, in addition, we have to ensure that credit~invariant~(3) is maintained
(e.g., we may have to find another 0.25 units of credit).

\begin{lemma} \label{lem:creditQ}
Let $H, C_0, R, P, Q, B^{new}$ be as stated above, and suppose that $H$ satisfies the
credit invariant.
Except for one case,
$\credits(Q) + \credits(R)$ suffices to
give one unit of credit for $E(P)-E(H)$ and
to give sufficient credit to $B^{new}$ such that
the credit~invariant holds at the end of the iteration.

In the exceptional case (described in Case~3, see below),
the algorithm adds another ear $P_{\star}$
with associated prefix $Q_{\star}$ such that $\credits(Q_{\star}-Q)\ge0.75$ and,
moreover, no connected component of $H-C_0$ is incident to both $P$ and $P_{\star}$,
and ``sells'' a unit-bridge of $Q$
(i.e., a unit-bridge of $\wh{\DTWO}$ is permanently~discarded
and is not contained in the 2-ECSS output by the algorithm,
thereby releasing 1.75~credits).
Thus, the algorithm obtains $\ge2.5$~credits, and this suffices to
give one unit of credit for each of $E(P)-E(H)$ and $E(P_{\star})-E(H)$, and
to give sufficient credit to $B^{new}$ such that
the credit~invariant holds at the end of the iteration.
\end{lemma}

The rest of this subsection presents a proof of this lemma.
Before presenting the details, we give an informal overview of the credit~scheme
in the following two tables.
The goal is to find one unit of credit for each ear added by the current iteration,
and ensure that $B^{new}$ has sufficient credits at the end of the iteration,
assuming that the credit~invariant holds at the start of the iteration.
If our iteration adds one ear, then clearly it suffices to find 1.25~credits
(but there are cases where the goal can be achieved with fewer credits).
We have two cases: either $Q$ has a white node $w$ other than $r$,
or else $Q$ has no white node other than $r$.

\newlength\replength
\newcommand\repfrac{.33}
\newcommand\dashfrac[1]{\renewcommand\repfrac{#1}}
\setlength\replength{1.5pt}
\newcommand\rulewidth{.6pt}
\newcommand\tdashfill[1][\repfrac]{\cleaders\hbox to \replength{%
  \smash{\rule[\arraystretch\ht\strutbox]{\repfrac\replength}{\rulewidth}}}\hfill}
\newcommand\tabdashline{%
  \makebox[0pt][r]{\makebox[\tabcolsep]{\tdashfill\hfil}}\tdashfill\hfil%
  \makebox[0pt][l]{\makebox[\tabcolsep]{\tdashfill\hfil}}%
  \\[-\arraystretch\dimexpr\ht\strutbox+\dp\strutbox\relax]%
}
\newcommand\tdotfill[1][\repfrac]{\cleaders\hbox to \replength{%
  \smash{\raisebox{\arraystretch\dimexpr\ht\strutbox-.1ex\relax}{.}}}\hfill}
\newcommand\tabdotline{%
  \makebox[0pt][r]{\makebox[\tabcolsep]{\tdotfill\hfil}}\tdotfill\hfil%
  \makebox[0pt][l]{\makebox[\tabcolsep]{\tdotfill\hfil}}%
  \\[-\arraystretch\dimexpr\ht\strutbox+\dp\strutbox\relax]%
}

\begin{center}
\begin{tabulary}{0.95\textwidth}{ |L|L| }
\hline
Properties of $Q$ & Credits available to algorithm excluding $\credits(R)$
\\
\hline
\hline
$Q$ has a white node $w$ other than $r$ & $\ge1.0$ (from 2ec-block containing $w$)
\\
\hline
subcase:
$Q$ contains a unit-bridge & $\ge1.75$ (1.0 from above \& 0.75 from unit-bridge)
\\
subcase:
$Q$ contains no unit-bridges & $\ge1.0$ (1.0 from above; \textbf{note}: $R$ has $\ge2.25$ credits)
\\
\hline
\end{tabulary}
\end{center}

{
\begin{center}
\begin{tabulary}{0.95\textwidth}{ |L| p{1.35in} | }
\hline
Properties of $Q$ & Credits available to algorithm excluding $\credits(R)$
\\
\hline
\hline
$Q$ contains no white nodes other than $r$,
$Q$ has $\ge2$ bridges,
$\cost(Q)\ge1.0$,
node sequence of $Q$: $v_0=r, v_1=u, v_2, \dots, v_k=\pstart$
&
\\
\hline
case:
$\cost(Q) \ge2$ & $\ge1.5$ (this suffices)
\\
\tabdashline & \tabdashline
& \\
case:
$\cost(Q) =1.0$ &  (see 4 subcases below)
\\
\tabdashline & \tabdashline
& \\
subcase:
$Q$ has 2 bridges, thus $\pstart=v_2$, \textbf{and}
$\cost(v_0 v_1)=1, \cost(v_1 v_2)=0$ & not possible
\\
\tabdashline & \tabdashline
& \\
subcase:
$Q$ has 2 bridges, thus $\pstart=v_2$, \textbf{and}
$\cost(v_0 v_1)=0, \cost(v_1 v_2)=1$,  \textbf{and}
$H$ has a   unit-bridge $v_2v_3$ &
	$=0.75$ (this suffices)
\\
\tabdashline & \tabdashline
& \\
subcase:
$Q$ has 2 bridges, thus $\pstart=v_2$, \textbf{and}
$\cost(v_0 v_1)=0, \cost(v_1 v_2)=1$,  \textbf{and}
the only other bridge of $H$ incident to $v_2$ is a zero-bridge $v_2 v_3$ &
\\
& \\
	\mbox{~~~~~~} \textbf{update $H$}:~~add another ear $P_{\star}$, permanently~discard unit-bridge $v_1 v_2$ &
	$\ge2.5$ (this suffices)
\\
\tabdashline & \tabdashline
& \\
subcase:
$Q$ has 3 bridges, thus $\pstart=v_3$, \textbf{and}
$\cost(v_0 v_1)=0, \cost(v_1 v_2)=1, \cost(v_2 v_3)=0$ &
	$=0.75$ (this suffices)
\\
\hline
\end{tabulary}
\end{center}
}

Suppose that the prefix $Q$ of the $R$-ear contains a white node $w$ other than $r$
(thus, $w\not=r$).
Let $B\not=R$ denote the 2ec-block that contains $w$;
clearly, $B$ has at least one credit (see credit~invariant~(3)).
Observe that $B^{new}$ contains $B$ and
the credits of both $B$ and $Q$ are available.
If $Q$ has at least one unit-bridge, then the sum of the credits
of $B$ and $Q$ is $\ge1.75$, and this suffices.
Now, suppose that $Q$ has no unit-bridges.
Thus, $Q$ has only one bridge, namely, the zero-bridge $ru$.
Then, by our choice of $ru$ and the credit invariant, $R$ already
has $\ge2.25$ credits, hence, we need only one credit for the
ear-augmentation step (since $B^{new}$ gets all the credits of $R$).
Thus, $\credits(B)=1$ suffices when $Q$ has no bridge other than $ru$.

Now, assume that the prefix $Q$ has no white node other than $r$;
in particular, the nodes $u$ and $\pstart$ are black.
It is easily seen that $Q$ has $\ge2$ bridges.
(Since $G$ is 2NC, $G-u$ has a path between the
connected~component of $C_0-u$ that contains $R$ and
each of the other connected~components of $C_0-u$;
this implies that there exists an $R$-ear such that
its associated prefix ${Q'}$ has $\ge2$ bridges and has $\cost(Q')\ge1$;
then, by our choice of $Q$, we have $\cost(Q)\ge\cost(Q')\ge1$ and $Q$ has $\ge2$~bridges.)

Suppose that $\cost(Q)\ge2$, thus, $Q$ has $\ge2$ unit-bridges.
Then, $\credits(Q)\ge 1.5$.
This suffices to pay for the current ear-augmentation step and to
re-establish the credit invariant.

Now, we may assume that $\cost(Q)=1$.
Let us denote the node sequence of $Q$ by $v_0,v_1,v_2,\dots,v_k$, where
$v_0=r$, $v_1=u$, and $v_k=\pstart$.
Clearly, we have three possibilities:

\begin{description}{
\item[Case~1: $Q$ consists of 2~bridges and $\cost(v_0v_1)=1$, $\cost(v_1v_2)=0$:~~]
We argue that this possibility cannot occur.
Observe that $v_2=\pstart$.
Observe that $G - \{v_1,v_2\}$ is connected;
otherwise, $\{v_1,v_2\}$ would be a bad-pair.
Consider the connected~components of $C_0 - \{v_1,v_2\}$;
let $S$ denote the set of nodes of the connected~component that contains $R$,
and let $T$ denote $V(C_0) - \{v_1,v_2\} - S$.
Since $G - \{v_1,v_2\}$ is connected,
it has path $P_{\star}$ between a node $\auxpend\in S$ and a node
$\auxpstart\in T$ such that $P_{\star}$ has no internal nodes in $C_0$.
Observe that every path of $C_0$ between $r$ and $\auxpstart$ has cost $\ge2$
(because such a path contains the unit-edge $ru$ as well as another edge of
the cut $\delta_{C_0}(\{v_1,v_2\})$, and all edges of this cut have cost~one);
thus, the prefix $Q_{\star}$ associated with $P_{\star}$ has cost $\ge2$.
This contradicts our choice of $P, \pstart, \pend$
(because the prefix $Q$ has cost~one).
See Figure~\ref{f:bc-case1}.

{
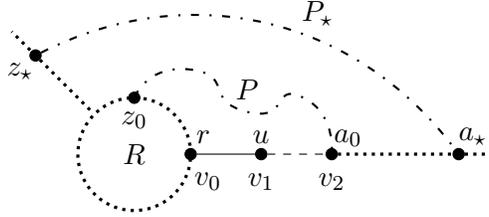
\begin{figure}[ht]
    \centering
    \begin{tikzpicture}[scale=0.75]
        \begin{scope}[every node/.style={circle, fill=black, draw, inner sep=0pt,
        minimum size = 0.15cm
        }]
            \draw[very thick,dotted] (0,0) circle (1cm);
            \node[label={[label distance=2]60:$r$}, label={[label distance=4]290:$v_0$}] (r) at (1,0) {};
            \node[label={[label distance=0]90:$u$}, label={[label distance=2]270:$v_1$}] (u) at (2.25,0) {};
            \node[label={[label distance=0]270:$z_0$}] (z0) at (0,1) {};
            \node[label={[label distance=0]45:$a_0$}, label={[label distance=2]270:$v_2$}] (a0) at (3.5,0) {};
            
            \node[fill=none,draw=none] (x0) at (-0.7071,0.7071) {};
            \node[fill=none,draw=none] (x1) at (-2.25,2.25) {};
            \node[label={[label distance=0]225:$\auxpend$}] (zstar) at (-1.75,1.75) {};
            \node[label={[label distance=0]70:$\auxpstart$}] (astar) at (5.75,0) {};
            \node[fill=none,draw=none] (x2) at (6.5,0) {};
        \end{scope}

        \begin{scope}[every edge/.style={draw=black}]
            \path[very thick,dotted] (x0) edge node {} (x1);
            \path[] (r) edge node {} (u);
            \path[dashed] (u) edge node {} (a0);
            \path[very thick,dotted] (a0) edge node {} (x2);
            
            \draw[thick, loosely dash dot] plot [smooth,tension=1.4] coordinates {(0,1)(1,1.5)(2,0.75)(3,1)(3.5,0)};
            
            \path[thick,loosely dash dot] (zstar) edge[bend left=40] node {} (astar);
        \end{scope}

        \begin{scope}[every node/.style={draw=none,rectangle}]
            \node (Rlabel) at (0,0) {$R$};
            \node (Plabel) at (2,1.15) {$P$};
            \node (Pstarlabel) at (3.25,2.5) {$P_{\star}$};
        \end{scope}
    \end{tikzpicture}
    \caption{\protect\centering Illustration of Case 1 (bridge~covering).}
    \label{f:bc-case1}
\end{figure}
}

\item[Case~2: $Q$ consists of 3~bridges and $\cost(v_0v_1)=0$, $\cost(v_1v_2)=1$, $\cost(v_2v_3)=0$:~~]
Then, we have $0.75$ credits available in $Q$.
We argue that this suffices to pay for the current ear-augmentation
step and to re-establish the credit invariant.
Observe that $v_3=\pstart$.
By our choice of $ru$ and the credit invariant,
$R$ has $\ge2.25$ credits.
Thus, we can pay one~unit for one unit-edge of $E(P)-E(H)$, and we
have $2$~credits (from $R$) left for $B^{new}$.
Moreover, $v_3$ is a black node, and it can be seen that
one of the unit-bridges of $C_0$ incident to $v_3$ becomes
a bridge incident to $B^{new}$;
in other words, at the next iteration, when we designate $B^{new}$
as the root 2ec-block, then only $2$~credits are required for $B^{new}$.
See Figure~\ref{f:bc-case2}.

{
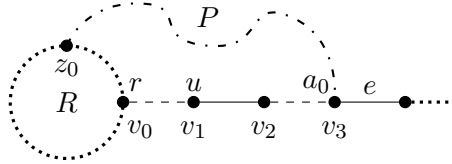
\begin{figure}[ht]
    \centering
    \begin{tikzpicture}[scale=0.75]
        \begin{scope}[every node/.style={circle, fill=black, draw, inner sep=0pt,
        minimum size = 0.15cm
        }]
            \draw[very thick,dotted] (0,0) circle (1cm);
            \node[label={[label distance=2]60:$r$}, label={[label distance=4]290:$v_0$}] (r) at (1,0) {};
            \node[label={[label distance=0]90:$u$}, label={[label distance=2]270:$v_1$}] (u) at (2.25,0) {};
            \node[label={[label distance=2]270:$v_2$}] (v2) at (3.5,0) {};
            \node[label={[label distance=0]270:$z_0$}] (z0) at (0,1) {};
            \node[label={[label distance=1]135:$a_0$}, label={[label distance=2]270:$v_3$}] (a0) at (4.75,0) {};
            \node[] (v4) at (6,0) {};
            
            \node[fill=none,draw=none] (x0) at (7,0) {};
        \end{scope}

        \begin{scope}[every edge/.style={draw=black}]
            \path[dashed] (r) edge node {} (u);
            \path[] (u) edge node {} (v2);
            \path[dashed] (v2) edge node {} (a0);
            \path[] (a0) edge node {} (v4);
            \path[very thick,dotted] (v4) edge node {} (x0);
            
            \draw[thick, loosely dash dot] plot [smooth,tension=1.4] coordinates {(0,1)(1.25,1.75)(2.5,1)(4,1.5)(4.75,0)};
        \end{scope}

        \begin{scope}[every node/.style={draw=none,rectangle}]
            \node (Rlabel) at (0,0) {$R$};
            \node (Plabel) at (2.5,1.5) {$P$};
            \node (elabel) at (5.375,0.25) {$e$};
        \end{scope}
    \end{tikzpicture}
    \caption{\protect\centering Illustration of Case 2 (bridge~covering).}
    \label{f:bc-case2}
\end{figure}
}

\item[Case~3: $Q$ consists of 2~bridges and $\cost(v_0v_1)=0$, $\cost(v_1v_2)=1$:~~]
\label{page:bc-case3}
Observe that $v_2=\pstart$.
Note that $v_2$ is a black node, and it is incident to
a bridge other than $v_2v_1$.
There are two subcases, namely,
either $v_2$ is incident to two unit-bridges or not,
and we choose $v_3$ appropriately.
In the former subcase, we take $v_2v_3$ to be a unit-bridge, and in the
latter subcase we have to take $v_2v_3$ to be the unique zero-bridge incident to $v_2$.

We handle the first subcase (with $\cost(v_2v_3)=1$) similarly to Case~2 above.
By our choice of $ru$ and the credit invariant,
$R$ has $\ge2.25$ credits.
Thus, we can pay one~unit for one unit-edge of $E(P)-E(H)$, and we
have $2$~credits (from $R$) left for $B^{new}$.
Moreover, $v_2$ is a black node, and it can be seen that
the unit-bridge $v_2v_3$ becomes a bridge incident to $B^{new}$;
in other words, at the next iteration, when we designate $B^{new}$
as the root 2ec-block, then only $2$~credits are required for $B^{new}$.

In the second subcase (with $\cost(v_2v_3)=0$), 
observe that $H$ has precisely two bridges incident to $v_2$,
because $v_2$ is incident to only one unit-bridge.
Our plan is to add a second ear and then observe that the edge
$v_1v_2$ becomes redundant after the addition of the two ears,
hence, we can permanently~discard this edge from $H$ thereby gaining $1.75$~credits.
(Note that when we permanently~discard a unit-edge of $\wh{\DTWO}$
from our solution subgraph $H$, then all of the retained~credits and
the working~credits of that edge become available.)
Moreover, we get another $0.75$ (or more) credits from the addition
of the two ears.  Thus we get $\ge 2.5$~credits, and this suffices
to pay for the addition of two ears and to re-establish the credit
invariant.  See Figure~\ref{f:bc-case3}.

{
\begin{figure}[ht]
    \centering
    \begin{tikzpicture}[scale=0.75]
        \begin{scope}[every node/.style={circle, fill=black, draw, inner sep=0pt,
        minimum size = 0.15cm
        }]
            \draw[very thick,dotted] (0,0) circle (1cm);
            \node[label={[label distance=2]60:$r$}, label={[label distance=4]290:$v_0$}] (r) at (1,0) {};
            \node[label={[label distance=0]90:$u$}, label={[label distance=2]270:$v_1$}] (u) at (2.25,0) {};
            \node[label={[label distance=0]45:$a_0$}, label={[label distance=2]270:$v_2$}] (a0) at (3.5,0) {};
            \node[label={[label distance=2]270:$v_3$}] (v3) at (4.75,0) {};
            
            \node[label={[label distance=0]270:$z_0$}] (z0) at (0,1) {};
            \node[label={[label distance=0]180:$w_1$}] (w1) at (2.25,-1.5) {};
            \node[label={[label distance=0]180:$\auxpend$}] (zstar) at (2.25,-2.75) {};
            
            \node[fill=none,draw=none] (x1) at (2.25,-3.5) {};
            \node[label={[label distance=0]70:$\auxpstart$}] (astar) at (7,0) {};
            \node[fill=none,draw=none] (x2) at (7.75,0) {};
        \end{scope}

        \begin{scope}[every edge/.style={draw=black}]
            \path[very thick,dotted] (w1) edge node {} (x1);
            \path[dashed] (r) edge node {} (u);
            \path[] (u) edge node {} (a0);
            \path[dashed] (a0) edge node {} (v3);
            \path[very thick,dotted] (v3) edge node {} (x2);
            
            \path[bend left = 45] (u) edge node {} (w1);
            
            \draw[thick, loosely dash dot] plot [smooth,tension=1.4] coordinates {(0,1)(1,1.5)(2,0.75)(3,1)(3.5,0)};
            
            \draw[thick, loosely dash dot] plot [smooth,tension=1.4] coordinates {(2.25,-2.75)(3.5,-2.75)(4.5,-1.5)(6,-1)(7,0)};
        \end{scope}

        \begin{scope}[every node/.style={draw=none,rectangle}]
            \node (Rlabel) at (0,0) {$R$};
            \node (Plabel) at (2,1.15) {$P$};
            \node (Pstarlabel) at (4.5,-2) {$P_{\star}$};
        \end{scope}
    \end{tikzpicture}
    \caption{\protect\centering Illustration of Case 3 (bridge~covering).}
    \label{f:bc-case3}
\end{figure}
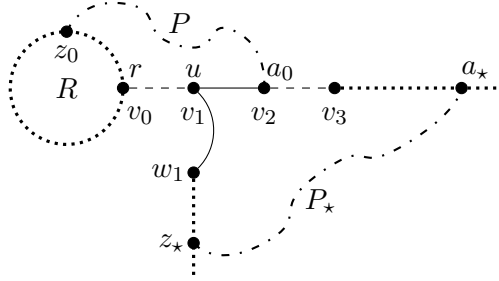
}

Observe that $G - \{v_2,v_3\}$ is connected;
otherwise, $\{v_2,v_3\}$ would be a bad-pair.
Consider the connected~components of $C_0 - \{v_2,v_3\}$;
let $S$ denote the set of nodes of the connected~component that contains $R$,
and let $T$ denote $V(C_0) - \{v_2,v_3\} - S$.
Since $G - \{v_2,v_3\}$ is connected,
it has a path $P_{\star}$ between a node $\auxpend\in S$ and a node
$\auxpstart\in T$ such that $P_{\star}$ has no internal nodes in $C_0$.
Moreover, w.l.o.g.\ we assume that $P_{\star}$ has the minimum number
of edges from $E(G)-E(H)$.
It can be seen that $P_{\star}, \auxpstart, \auxpend$ satisfy the following:

\begin{enumerate}[(i)]
{
\item
Either $\auxpend = v_1$ or there is a bridge $v_1w_1$ (of $C_0$) such
that $v_1w_1\not=v_1v_0$, $v_1w_1\not=v_1v_2$, and $\auxpend$ is in the
connected~component of $C_0 - v_1w_1$ that contains $w_1$.
This follows from a contradiction argument;
the only other possibility is that $\auxpend$ is in $C_0^{r}$
(the connected~component of $C_0-ur$ that contains $r$);
but then we would define the prefix $\wt{Q}_{\star}$ associated with
$P_{\star}$ to be a path of $C_0$ between $r$ and $\auxpstart$;
note that $\wt{Q}_{\star}$ would contain $v_3$ (since $C_0$ has only two
bridges incident to $v_2$),
hence, $\wt{Q}_{\star}$ would have $\ge3$~bridges and we would have
$\cost(\wt{Q}_{\star})\ge1$;
this would contradict our choice of $P, \pstart, \pend$.

\item
Now, we define the prefix $Q_{\star}$ associated with $P_{\star}$ to be
a path of $C_0$ between $v_1$ and $\auxpstart$.
Note that $Q_{\star}$ contains $v_3$
(since $C_0$ has only two bridges incident to $v_2$).
We have two cases:
either $v_3$ is a black node or it is a white node;
in the first case, $Q_{\star}-Q$ contains a unit-bridge incident to $v_3$
so we have $\credits(Q_{\star}-Q)\ge0.75$, whereas in the second case,
$Q_{\star}-Q$ contains the white node $v_3$
so we can obtain $\ge1$ credit from the 2ec-block (of $C_0$) that contains $v_3$.

\item
There is no connected~component of $H - C_0$ that is incident to both
$P$ and $P_{\star}$.  Otherwise, suppose that some connected~component
$\hat{C}$ of $H$ is incident to both P and  $P_{\star}$. Then there is
a path in ($P \cup  P_{\star} \cup \hat{C}$) that starts at $z_0$ and
ends at $a_{\star}$, and whose prefix $Q_{\star} \cup \{v_0v_1\}$ in
$C_0$ either has a white node or has cost $\ge2$, contradicting our
choice of $P,a_0,z_0$.

Hence, we can apply Lemma~\ref{lem:thick} separately to
each of $P_{\star}$ and $P$ and thus pay for all-but-one of the unit-edges
added by each of the two ears.

\item
After adding the two ears, the edge $v_1v_2$ can be permanently~discarded from
$H$ while preserving 2-edge~connectivity, by Proposition~\ref{propo:2ecdiscard}.
To see this, observe that $(C_0 - v_1v_2) \cup P$ contains a $v_1,v_2$
path, and also $(C_0 - v_1v_2) \cup P_{\star}$ contains a $v_1,v_2$
path, and moreover, these two paths have no internal nodes in common
(i.e., $(C_0 - v_1v_2) \cup P \cup P_{\star}$ contains a cycle incident
to $v_1$ and $v_2$).

\item
Consider the credits available from $Q$ and $Q_{\star}-Q$ for
the double ear-augmentation.
``Selling'' the unit-edge $v_1v_2$ gives $1.75$~credits
(since $v_1v_2$ is a unit-bridge of $\wh{\DTWO}$, the sum
of its retained~credit and working~credit is $1.75$, and
since $v_1v_2$ is permanently discarded from $H$,
all of this credit is released).
Moreover, we have $\ge0.75$ credits available in $Q_{\star}-Q$.
Thus, $\ge2.5$ credits are available.

\item
By modifying the arguments in the proof of Lemma~\ref{lem:findP},
we can compute (a 3-tuple) $P_{\star}, \auxpstart, \auxpend$ that
has the required properties in polynomial time.
Hence, a double ear-augmentation can be implemented in polynomial~time.
}
\end{enumerate}

Summarizing, when $v_2v_3$ is a zero-edge, then we add the two ears
$P, P_{\star}$ and permanently~discard the edge $v_1v_2$ from $H$; we have
sufficient credits to pay for the addition of the two ears and to
re-establish the credit invariant.
}
\end{description}

\begin{proposition} \label{propo:CI}
 (i)~The credit~invariant holds for $\wh{\DTWO}$.
(ii)~The credit~invariant holds at the end of every iteration of bridge~covering
(i.e., every iteration of bridge~covering preserves the credit~invariant).
\end{proposition}
\label{page:CIproof}
\begin{proof}
It is easily seen that (i)~holds.

Now, focus on any iteration.
We use the notation of this section; in particular, we use $P$ to
denote the first ear added in an iteration;
moreover, we use $\earP$ to denote the corresponding $R$-ear (see
page~\pageref{page:r-ear}); also, let $R^{new}$ denote the root
2ec-block of $H^{new}$ (assume it exists).

It is easy to verify that parts~(1) and~(2) of the credit~invariant are
preserved by every iteration.
Whenever we add a unit-edge $e$ to $H$ (in an ear-augmentation)
then we ensure that $e$ has a unit of retained~credit,
by Lemmas~\ref{lem:thick},~\ref{lem:creditQ}.
Whenever we take~away credits (either working~credit or retained~credit) from a
unit-bridge $e$ of $C_0$ (e.g., when $e$ is in the prefix $Q$ of
$\earP$), then we retain sufficient credits for $e$ (if $e$ stays in
$H^{new}$ then it keeps $\ge1$ retained~credit and $\ge0$ working~credit, otherwise, it keeps $\ge0$
retained~credit and $\ge0$ working~credit);
moreover, if we take~away credit from a bridge $e$ of $C_0$, then
either $e$ is contained in the 2ec-block $R^{new}$ of $H^{new}$ or else
$e$ is permanently~discard from $H$ (thus, at most one iteration
can take~away credit from a unit-bridge of $C_0$).

Consider part~(3) of the credit~invariant.
First, let us consider the credits of $R$ and $R^{new}$.
Suppose that $R$ has $\alpha$ credits at the start of an iteration.
Then, aside from three exceptions, at the start of the next iteration,
$R^{new}$ has $\ge0.5+\alpha$ credits
(see the tables placed after the statement of Lemma~\ref{lem:creditQ}).
The first exception occurs when $Q$ has a white node $w$ other than $r$
and $Q$ consists of one zero-bridge;
in this case, $\credits(R^{new}) = \credits(R) \ge 2.25$.
The other two exceptions occur in cases~2 and~3(first~subcase);
in these two cases, $R$ has $\alpha\ge2.25$ credits (since $ru$ is a zero-bridge)
while $R^{new}$ has $\ge\alpha-0.25\ge2$ credits and
$R^{new}$ is guaranteed to be incident to a unit-bridge of $H^{new}$.
Hence, the credit~invariant pertaining to $R$ is preserved.

Consider any single ear-augmentation.
Consider the credits of the other (non-root) 2ec-blocks of $H$ and
$H^{new}$.  Any 2ec-block of $H$ that contains a node of $\earP$ is
``merged'' into $R^{new}$, hence, the credit~invariant is not relevant
for such 2ec-blocks of $H$ (we argued above that $R^{new}$ satisfies
the credit~invariant).
Consider a 2ec-block $B$ of $H$ such that there is a b-path (of $H$)
between $B$ and a black~node of $\earP$.  Then, in $H^{new}$, $R^{new}$
is b-adjacent to $B$, hence, $B$ is required to have $\ge1$ credit (by
the credit~invariant).  Although we may remove credits from such
2ec-blocks (e.g., see the proof of Lemma~\ref{lem:thick}, and note that
we take away 0.5~credits from each of $B(s_1)$ and $B(t_1)$), we ensure
that each such 2ec-block has at least one credit at the next iteration.
Lastly, consider any 2ec-block $B$ of $H$ that is (node) disjoint from
$R\cup\earP$ and is not b-adjacent to any black node of $\earP$.  Clearly, $B$
has the same credit in both $H$ and $H^{new}$.
It follows that credit~invariant~(3) is preserved in
every single ear-augmentation.

Now, consider a double ear-augmentation; thus, we have the second
subcase of case~3 that adds the two ears $P, P_{\star}$.
Let $\dblearP$ denote the union of the $R$-ear corresponding to $P$
and the $R$-ear corresponding to $P_{\star}$.
The arguments in the previous paragraph can be re-applied with one change:
we replace $\earP$ by $\dblearP$.
It follows that credit~invariant~(3) is preserved in
every double ear-augmentation.
\end{proof}

This concludes the discussion of bridge~covering.

\begin{proposition}
At the termination of the bridge~covering~step,
 $H$ is a bridgeless 2-edge~cover and the credit invariant holds
(thus, every original 2ec-block of $H$ has $\ge1.5$ credits and
every new 2ec-block of $H$ has $\ge2$ credits).
The bridge~covering~step can be implemented in polynomial time.
\end{proposition}
}

\section{ \label{s:algo-last}  The gluing step}
{
In this section and in Section~\ref{s:path-thicken}, we assume that
the input is a well-structured MAP instance.

In this section, we focus on the last step of the algorithm, namely,
the gluing step. Our goal here is to show that
the credits in $H$ suffice to update $H$ to a 2-ECSS of $G$ by adding
some edges and deleting some edges (i.e., the difference between the
number of edges added and the number of edges deleted in the
gluing~step is $\leq$ the credit of $H$ at the start of the
gluing~step).
The following result summarizes this section:

\begin{proposition}
At the termination of the bridge-covering step,
let $H$ denote the bridgeless 2-edge~cover computed by the algorithm
and suppose that the credit invariant holds;
let $\gamma$ denote $\credits(H)$.
Then the gluing step augments $H$ to a 2-ECSS $H'$ of $G$
(by adding edges and deleting edges)
such that $\cost(H') \leq \cost(H) + \gamma$.
The gluing~step can be implemented in polynomial time.
\end{proposition}

It is convenient to define the following multi-graph:
let $\hat{G}$ be the multi-graph obtained from $G$ by contracting each 2ec-block
$B_i$ of $H$ into a single node that we will denote by $B_i$.
Observe that $\hat{G}$ is 2EC.
Note that the algorithm ``operates'' on $G$ and never refers to $\hat{G}$;
but, for our discussions and analysis, it is convenient to refer to $\hat{G}$.

At the start of the gluing step,
recall that each original 2ec-block of $H$ has $\ge1.5$ credits and
each new 2ec-block of $H$ has $\ge2$ credits.
We pick any 2ec-block $R_0$ of $H$ and designate it as the root $R$;
then we apply iterations;
each iteration adds to $H$ the edges of an ear whose start~node
and end~node are in $R$; some iterations may add a second ear.
After each iteration, we update the notation so that $R$
denotes the 2ec-block of the current subgraph $H$ that contains $R_0$.
Our plan is to keep adding ears to $H$ until all of the nodes are in
$R$. Thus, we terminate when $H$ is 2EC, and on termination $H$ is a 2-ECSS of $G$.
In the following discussion, we assume that $R$ has zero credits, i.e.,
we ignore the credits available in $R$ while paying for the
edges added in ear-augmentation steps.

Consider a cycle $\hat{C}$ of $\hat{G}$ incident to $R$;
such a cycle exists since $\hat{G}$ is 2EC.
Let $|\hat{C}|$ denote the number of edges of $\hat{C}$; clearly, $|\hat{C}|=|V(\hat{C})|$.
Observe that each of the non-root nodes of $\hat{C}$ has $\ge1.5$ credits, hence,
we have $\ge1.5(|\hat{C}|-1)$ credits available from (the nodes of) $\hat{C}$
and this is $\ge|\hat{C}|$ whenever $|\hat{C}|\ge3$.
Thus, we have enough credit to buy all the edges of $\hat{C}$ whenever
$|\hat{C}|\ge3$. Moreover, if $|\hat{C}|=2$ and the non-root node of
$\hat{C}$ has $\ge2$ credits, then we have enough credit to buy all the
edges of $\hat{C}$.  Thus, we have insufficient credit only when
$|\hat{C}|=2$ and the non-root node of $\hat{C}$ has exactly 1.5 credits.

In what follows, we assume that all the cycles of $\hat{G}$ incident
to $R$ have insufficient credit.
Let $\hat{C}$ be a cycle of $\hat{G}$ that has insufficient credit,
where $\hat{C}=R,B,R$ and $B$ denotes the non-root node of $\hat{C}$.
Clearly, $B$ is an original 2ec-block of $H$ (otherwise, it
would have 2~credits rather than 1.5~credits).
Let $b$ denote the number of nodes\footnote{From this sentence
till the end of Section~\ref{s:algo-last},
``node'' means a node of $G$ (and not a node of $\hat{G}$).}
of $B$.
We have $b\in\{2,3,4\}$ because $B$ has $\ge\ceiling{b/2}$ unit-edges
($B$ has $\ge b$ edges by 2EC and has $\le\floor{b/2}$ zero-edges),
hence, $\credits(B)<2$ implies $b<5$.
Moreover, for $b\in\{3,4\}$, note that $B$ cannot have a cut~node
(otherwise, $B$ would have $\ge1+\ceiling{b/2}$ unit-edges), and hence,
(since $B$ is 2NC) $B$ must contain a spanning cycle.
For $b\in\{3,4\}$, let $Q(B)$ denote any spanning cycle of $B$.

The two edges of $\hat{C}$ correspond to two edges of $G$ between
$B$ and $R$; let $e$ denote one of these edges, and let $v_e$ denote
the end~node of $e$ in $B$.
Since $G$ is 2NC, $G-{v_e}$ has a path between
$(B-{v_e})$ and $R$.  Each such path has all its internal
nodes in $V(R) \cup (V(B)-\{v_e\})$
(by our assumption on cycles of $\hat{G}$ incident to $R$), hence,
there exists an edge $f$ of $G$ between $(B-{v_e})$ and $R$;
let $u_f$ denote the end~node of $f$ in $B-{v_e}$.
See Figure~\ref{f:gluing:1}.

{\begin{figure}[ht]
\centering
	\begin{tikzpicture}[thick,scale=0.70]
\draw (0,0) circle [radius=1.5];
\draw (0,0) node(r) {$R$};
\draw (5,0) circle [radius=1.5];
\draw (5,0) node(b) {$B$};
\draw (80:1) -- (3.5,0);
\filldraw
(30:1) -- (3.5,0) circle(2pt) node[pos=0.5, below]{$e$};
\draw (3.5,0) node[above left](ve)  {$v_e$};
\filldraw
(300:1) -- (3.93934,-1.06066) circle(2pt) node[pos=0.5, below]{$f$};
\draw (3.93934,-1.06066) node[below left](uf)  {$u_f$};
\end{tikzpicture}
\caption{\protect\centering Illustration of $R$, $B$, and $v_e, u_f$.}
	\label{f:gluing:1}
\end{figure}
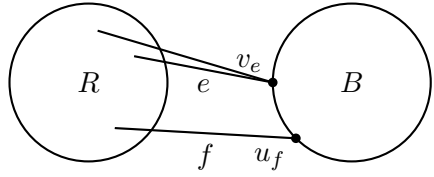
}

Now, we have two cases, depending on whether $v_e$ and $u_f$ are
adjacent in $B$ or not.

\begin{description}{
\item[Case~1: $e, f$ can be chosen such that
	$B$ has an edge between $v_e$ and $u_f$:\quad]
In this case,
we claim that $e,f$ can be chosen such that $B$ has a unit-edge
between $v_e$ and $u_f$.
By way of contradiction, suppose that $v_eu_f$ is a zero-edge.
Then we have $b\in\{3,4\}$
(otherwise, if $b=2$, then $B$ would consist of two parallel unit-edges),
and moreover, $G-\{v_e,u_f\}$ is connected
(since $\{v_e,u_f\}$ is not a bad-pair).
Hence, $G$ has an edge $rv_0$ such that $r$ is in $R$ and
$v_0$ is in $B-\{v_e,u_f\}$
(we also used our assumption on cycles of $\hat{G}$ incident to $R$).
Moreover, the spanning cycle~$Q(B)$ has an edge between
$v_0$ and $\{v_e,u_f\}$.
W.l.o.g.\ suppose that $B$ has the edge $v_0 u_f$;
this is a unit-edge (since $v_eu_f$ is a zero~edge).
Then by replacing the pair of edges $e,f$ by $rv_0,f$, we have two
edges between $R$ and $B$ such that their end~nodes in $B$ are distinct
and there exists a unit-edge of $B$ between these two end~nodes.
Our claim follows.

Now, observe that the graph $H \cup \{e,f\} - \{v_e u_f\}$ has two
edge-disjoint $v_e,u_f$~paths.
We buy the edges $e,f$ and permanently~discard the unit-edge ${v_e u_f}$;
that is, we add the two edges $e, f$ to $H$ and remove the edge
 ${v_e u_f}$ from $H$.
(In the gluing~step, when we permanently~discard a unit-edge from our solution
subgraph $H$, then one unit of retained~credit of that edge become available.)
In the resulting graph $H^{new}$,
the connected component containing $R_0$ (as well as $R$ and $B$) is 2EC,
by Proposition~\ref{propo:2ecdiscard}.
This step results in a surplus of 0.5~credits
(we get 1.5~credits from $B$, one credit from selling ${v_e u_f}$,
and we pay two credits for the edges $e,f$).

\item[Case~2: for any choice of $e, f$
	there is no edge between $v_e$ and $u_f$ in $B$:\quad]
\label{page:gluing-case2}
Then, clearly $b=|V(B)|=4$, and
$B$ has a spanning cycle $Q(B)$.
Let $Q=v_1,v_2,v_3,v_4,v_1$ denote $Q(B)$, where w.l.o.g.\
$v_e=v_1$ and $u_f=v_3$.  Since $B$ has 1.5~credits, two of the
(non-adjacent) edges of $Q$ must be zero-edges.
There must be one or more edges of $G$ incident to $v_2$ or  $v_4$, otherwise,
$Q$ would be a redundant 4-cycle of $G$.

Suppose that $G$ has the edge ${v_2v_4}$.  
See Figure~\ref{f:gluing:main}(a).
We buy the three edges $e,f,{v_2v_4}$ and we permanently~discard the two
unit-edges of $Q$.
In the resulting graph $H^{new}$,
the connected component containing $R_0$ (as well as $R$ and $B$) is 2EC,
by Proposition~\ref{propo:2ecdiscard}.
This step results in a surplus of 0.5~credits
(we get 1.5~credits from $B$, two credits from selling the two unit~edges of $Q$,
and we pay three credits for the edges $e,f,\;{v_2v_4}$).

Lastly, suppose that $v_2$ and $v_4$ are nonadjacent in $G$. Then
$G$ has an edge between another 2ec-block
${B'}$ of $H$ (where ${B'}\not=B$ and ${B'}\not=R$) and
one of $v_2$ or $v_4$, say $v_2$; let us
denote this edge by $\widetilde{e}$.  Since $G$ is 2NC, there is a path in
$G-\{v_2\}$ between ${B'}$ and $B-\{v_2\}$.  Let $\widetilde{P}$
denote such a path that has the fewest edges of $E(G)-E(H)$.  In
$\hat{G}$, observe that $\widetilde{P} \cup \{\widetilde{e}\}$ corresponds to a
cycle $\hat{C}_{{B'}}$ that is incident to $B$ and ${B'}$.
Moreover, in $\hat{G}$, note that $\hat{C}_{{B'}}$ cannot be
incident to $R$
(by our assumption on cycles of $\hat{G}$ incident to $R$).
See Figure~\ref{f:gluing:main}(b).
Let $e_Q$ denote the unit-edge of $Q$ that has its end~nodes among
$v_1, v_2, v_3$.
It can be seen that
$H \cup \{e,f,\widetilde{e}\} \cup E(\widetilde{P}) - \{e_Q\}$
has two edge-disjoint paths between the end~nodes of $e_Q$.
We buy $e,f$, we permanently~discard $e_Q$, and moreover, we buy the edges of
$\big(E(\widetilde{P})-E(H)\big) \cup \{\widetilde{e}\}$.
In the resulting graph $H^{new}$,
the connected component containing $R_0$ (as well as $R$ and $B$) is 2EC,
by Proposition~\ref{propo:2ecdiscard}.
(For example, suppose that $e_Q=v_1v_2$ and the end~node of $\widetilde{P}$ in $B$ is $v_3$;
then, we add the ear formed by $e,v_1v_4,v_4v_3,f$, and after that,
we add the edges in $E(G)-E(H)$ of the closed ear formed by $v_3v_2,\widetilde{e},\widetilde{P}$.)
This step results in a surplus of credits;
note that the sum of the credits of the 2ec-blocks (excluding $B$)
incident to $\widetilde{P}$ minus the size of
$\big(E(\widetilde{P})-E(H)\big) \cup \{\widetilde{e}\}$ is at least~$-0.5$.
}
\end{description}
{
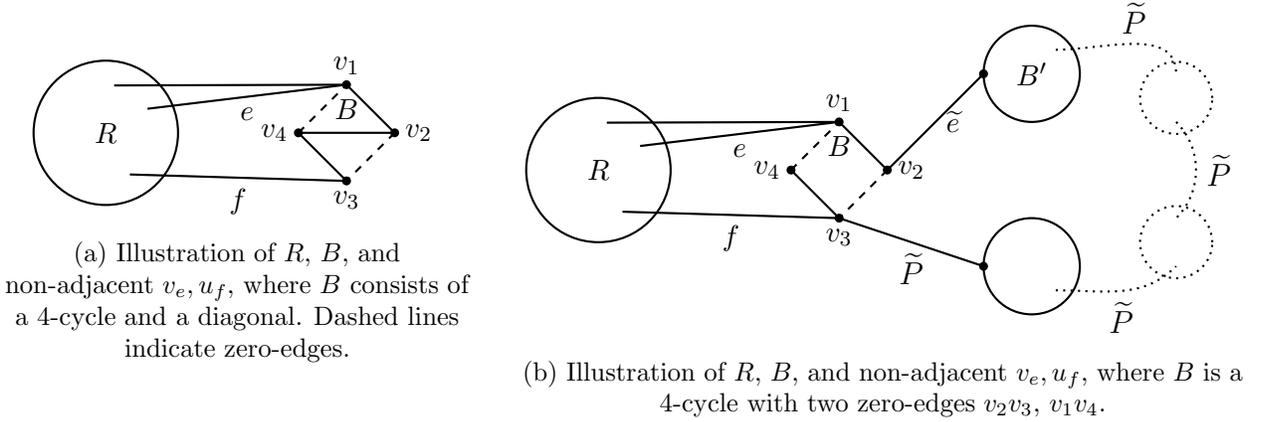
\begin{figure}[htb]
\centering
\begin{subfigure}{0.38\textwidth}
    \centering
    \begin{tikzpicture}[thick,scale=0.64]
    \draw (0,0) circle [radius=1.5];
    \draw (0,0) node(r) {$R$};
    \draw[thick] (5,1) node[above] (v1){$v_1$} -- (6,0);
    \draw[dashed]	(6,0) node[right] (v2){$v_2$} -- (5,-1);
    \draw[thick]	(5,-1) node[below] (v3){$v_3$} -- (4,0);
    \draw[dashed]	(4,0) node[left] (v4){$v_4$} -- (5,1);
    \draw (5,0.5) node(b) {$B$};
    \draw[thick]	(v2) -- (v4);
    \draw (80:1) -- (5,1);
    \filldraw
    (30:1) -- (5,1) circle(2pt) node[pos=0.5, below]{$e$};
    \filldraw
    (300:1) -- (5,-1) circle(2pt) node[pos=0.5, below]{$f$};
    \filldraw (6,0) circle(2pt);
    \filldraw (4,0) circle(2pt);
    \end{tikzpicture}
    \caption{\protect\centering Illustration of $R$, $B$, and
    non-adjacent $v_e, u_f$, where $B$ consists of a 4-cycle and a
    diagonal. Dashed lines indicate zero-edges.}
	\label{f:gluing:2}
\end{subfigure}
\hspace*{\fill}
\begin{subfigure}{0.58\textwidth}
    \centering
    \begin{tikzpicture}[thick,scale=0.64]
    \draw (0,0) circle [radius=1.5];
    \draw (0,0) node(r) {$R$};
    \draw[thick] (5,1) node[above] (v1){$v_1$} -- (6,0);
    \draw[dashed] (6,0) node[right] (v2){$v_2$} -- (5,-1);
    \draw[thick] (5,-1) node[below] (v3){$v_3$} -- (4,0);
    \draw[dashed] (4,0) node[left] (v4){$v_4$} -- (5,1);
    \draw (5,0.5) node(b) {$B$};
    \draw (80:1) -- (5,1);
    \filldraw
    (30:1) -- (5,1) circle(2pt) node[pos=0.5, below]{$e$};
    \filldraw
    (300:1) -- (5,-1) circle(2pt) node[pos=0.5, below]{$f$};
    \filldraw (6,0) circle(2pt);
    \filldraw (4,0) circle(2pt);
    \draw (9,2) node(hb0) {${B'}$} circle [radius=1];
    \filldraw
    (6,0) -- (8,2) circle(2pt) node[pos=0.5, right]{$\widetilde{e}$};
    \draw[dotted] (12,1.5)  node(hb1) {} circle [radius=.75];
    \draw[dotted] (12,-1.5) node(hb2) {} circle [radius=.75];
    \draw (9,-2) node(hb3) {} circle [radius=1];
    \draw[dotted] (9.5,2.5)  to[out=360,in=90]
    	node[pos=0.5, above]{\large $\widetilde{P}$}
    	(12,2.0);
    \draw[dotted] (12, 1.0) to[out=315,in=45]
    	node[pos=0.5, right]{\large $\widetilde{P}$}
    	(12,-1.0);
    \draw[dotted] (9.5,-2.5) to[out=360,in=225]
    	node[pos=0.5, below]{\large $\widetilde{P}$}
    	(12,-2.0);
    \filldraw
    (5,-1) -- (8,-2) circle(2pt) node[pos=0.5, below]{\large $\widetilde{P}$};
    \end{tikzpicture}
    \caption{\protect\centering Illustration of $R$, $B$, and
    non-adjacent $v_e, u_f$, where $B$ is a 4-cycle with two zero-edges $v_2v_3$, $v_1v_4$.}
	\label{f:gluing:3}
\end{subfigure}
\caption{\protect\centering Illustrations for Case~2.}
\label{f:gluing:main}
\end{figure}
}
}

{
\section{ \label{s:lowerbounds} Examples showing lower bounds}

This section presents two examples that give
lower~bounds on our results on MAP;
each example is a well-structured instance of MAP.
The first example gives a construction
such that $\opt \approx \frac74 \dtwocost$.
This shows that Theorem~\ref{thm:approxbydtwo} is essentially tight.
The second example gives a construction such that the cost of the
solution computed by our algorithm is $\approx\frac74\opt$.

\subsection{Optimal solution versus min-cost 2-edge~cover}

{
\begin{figure}[ht]
    \centering
    \begin{tikzpicture}[scale=0.8]
        \begin{scope}[every node/.style={circle, fill=black, draw, inner sep=0pt,
        minimum size = 0.15cm
        }]
            
            \node[label={[label distance=2]45:$v_1$}] (v1) at (0,0) {};
            \node[label={[label distance=2]45:$v_2$}] (v2) at (2,0) {};
            \node[label={[label distance=2]45:$v_3$}] (v3) at (2,2) {};
            \node[label={[label distance=2]45:$v_4$}] (v4) at (0,2) {};
            \node[label={[label distance=2]45:$v_5$}] (v5) at (4,0) {};
            \node[label={[label distance=2]45:$v_6$}] (v6) at (6,0) {};
            \node[label={[label distance=2]45:$v_7$}] (v7) at (6,2) {};
            \node[label={[label distance=2]45:$v_8$}] (v8) at (4,2) {};
            
            \node[fill=none, thick, dotted, minimum size=1.4cm] (b0) at (-2,3.25) {};
        \end{scope}

        \begin{scope}[every edge/.style={draw=black}]
            \path (v1) edge[bend left=30] node {} (b0);
            \path (v3) edge[bend right=30] node {} (b0);
            \path (v1) edge[min distance=4.4cm,in=315,out=345] node {} (v7);
            
            \path (v1) edge node {} (v2);
            \path[dashed] (v2) edge node {} (v3);
            \path (v3) edge node {} (v4);
            \path[dashed] (v4) edge node {} (v1);
            \path (v5) edge node {} (v6);
            \path[dashed] (v6) edge node {} (v7);
            \path (v7) edge node {} (v8);
            \path[dashed] (v8) edge node {} (v5);
            \path (v3) edge node {} (v8);
            \path (v2) edge node {} (v5);
        \end{scope}

        \begin{scope}[every node/.style={draw=none,rectangle}]
            \node (b0label) at (-2,3.25) {$B_0$};
        \end{scope}
    \end{tikzpicture}
    \caption{\protect\centering
	Example graph (with $k=1$ copy of the 8-node gadget) where
  $\opt\geq7k+\epsilon$ and $\dtwocost\leq4k+\epsilon$, where $\epsilon$
  is a constant.
  Edges of cost~zero and~one are illustrated by dashed and solid lines, respectively.
  $B_0$ is the root 2ec-block; $\epsilon$ is the cost of an optimal
  solution on $B_0$.}
	\label{f:tightD2algo}
\end{figure}
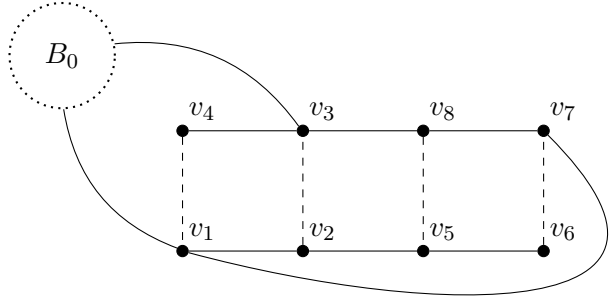
}

\begin{proposition} \label{propos:example-opt}
For any $k \in \mathbb{N}$, there exists a well-structured $MAP$ instance $G_k$ such
that $\dtwocost(G_k) \leq 4k+3$ and $\opt(G_k) \geq 7k+3$.
\end{proposition}

\begin{proof}
The graph $G_k$ consists of a root 2-ec block $B_0$ and $k$ copies
$J_1,\dots,J_k$ of a gadget subgraph $J$.
The gadget subgraph $J$ consists of 8 nodes $v_1,\dots,v_8$ and 11 edges;
there are four zero-edges $v_1v_4,v_2v_3,v_5v_8,v_6v_7$, and
seven unit-edges $v_1v_2,v_1v_7,v_2v_5,v_3v_4,v_3v_8,v_5v_6,v_7v_8$;
see the subgraph induced by the nodes $v_1,\dots,v_8$ in
Figure~\ref{f:tightD2algo}; observe that 8~of the 11~edges form two
4-cycles (namely, $v_1,v_2,v_3,v_4,v_1$ and $v_5,v_6,v_7,v_8,v_5$)
and the other three edges are $v_2v_5$, $v_3v_8$, and $v_1v_7$.

Let $B_0$ be a 6-cycle $w_1,\dots,w_6,w_1$ that has 3~unit-edges and 3~zero-edges.

$G_k$ has two unit-edges between each copy of the gadget subgraph $J_i$
($i=1,\dots,k$) and $B_0$; these two edges are incident to the nodes
$v_1$ and $v_3$ of $J_i$ (see the illustration in
Figure~\ref{f:tightD2algo}) and to the nodes $w_1$ and $w_4$ of $B_0$.
Observe that the subgraph of $G_k$ consisting of $B_0$ and the two
4-cycles (namely, $v_1,v_2,v_3,v_4,v_1$ and $v_5,v_6,v_7,v_8,v_5$) of
each copy of the gadget subgraph is a (feasible) 2-edge~cover of $G_k$
of cost $4k+3$.
Hence, $\dtwocost(G_k)\leq 4k+3$.

Finally, we claim that $\opt(G_k) \geq 7k+3$.
In what follows, we use $\optsol$ to denote an optimal solution of
$G_k$, i.e., $\optsol$ denotes an arbitrary but fixed min-cost
2-ECSS of $G_k$.
Clearly, $\optsol$ has to contain all the edges of $B_0$ as well as
the two edges between $B_0$ and each copy of the gadget subgraph.
Now, we focus on one copy $J_i$ of the gadget subgraph, and let
$\opt(G_k,J_i)$ denote the cost of the edges of $\optsol$ incident to $J_i$.
We will show that $\opt(G_k,J_i)\geq7$, hence, it follows that
$\opt(G_k)=3+\sum_{i=1}^{k}\opt(G_k,J_i)\geq 7k+3$.
Since $\deg(v_4)=\deg(v_6)=2$, $\optsol$ must pick the edges ${v_1 v_4}$
and ${v_3 v_4}$ as well as the edges ${v_5 v_6}$ and ${v_6 v_7}$.
Consider the cut $\delta_{G_k}(\{v_5, v_6, v_7, v_8\})$.
This cut has three unit-edges: ${v_1 v_7},\; {v_2 v_5},\; {v_3 v_8}$.
We have the following cases:\\
Case 1: $\optsol$ picks all three edges of the cut. Then $\opt(G_k,J_i) \geq 7$.\\
Case 2: $\optsol$ does not pick all three edges of the cut. Then we
will show that it must pick two edges from the cut and one more
unit-edge, thus giving us $\opt(G_k,J_i)\geq7$. We have three subcases:\\
	Case 2.1: ${v_1 v_7} \;\not\in\; \optsol$: 
	  then $\optsol$ must pick ${v_7 v_8}$ and the other 2 edges of the cut.\\
	Case 2.2: ${v_2 v_5} \;\not\in\; \optsol$: 
	  then $\optsol$ must pick ${v_1 v_2}$ and the other 2 edges of the cut.\\
	Case 2.3: ${v_3 v_8} \;\not\in\; \optsol$: 
	  then $\optsol$ must pick ${v_7 v_8}$ and the other 2 edges of the cut.\\
Hence, $\opt(G_k,J_i) \geq 7$ and we have $\opt(G_k) \geq 7k+3$. This completes the proof.
\end{proof}

\subsection{Optimal solution versus algorithm's solution}

We present a family of graphs $G_k$, $k = 1,2,3\ldots$
for which the ratio of the cost of a solution obtained by our algorithm
and the cost of an optimal solution approaches $\frac{7}{4}$.
Let the solution subgraph found by applying our algorithm to $G_k$ be
denoted by $\ALGO(G_k)$.

\begin{proposition} \label{propos:alg-opt}
For any $k \in \mathbb{N}$, there exists a well-structured $MAP$ instance $G_k$
such that $\cost(\ALGO(G_k)) \geq 7k + 6$ and
$\opt(G_k) \leq 4k + 7$.
\end{proposition}

\begin{proof}
Let $J$ be a gadget on nodes $v_1,v_2,v_3,v_4, w_1,w_2,w_3,w_4$, with
edges $v_1v_4$, $v_2v_3$, $w_1w_2$, and $w_3w_4$ of cost~zero, and
edges $v_1v_2$, $v_3v_4$, $w_1w_4$, $w_2w_3$, $g=v_4w_1$ and $h=v_2v_4$
of cost~one, as shown in Figure~\ref{f:algo:gadget} (where dashed and solid
lines represent edges of cost~zero and~one, respectively).

{
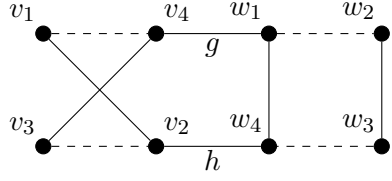
\begin{figure}[ht]
	\centering
	\begin{tikzpicture}[scale=0.75]
	\begin{scope}[every node/.style={circle, fill=black, draw, inner sep=0pt,
	minimum size = 0.2cm
	}]
	\node[label={[label distance=2]135:$v_1$}] (v1) at (0,0) {};
	\node[label={[label distance=2]45:$v_4$}] (v4) at (2,0) {};
	\node[label={[label distance=2]135:$v_3$}] (v3) at (0,-2) {};
	\node[label={[label distance=2]45:$v_2$}] (v2) at (2,-2) {};
	\node[label={[label distance=2]135:$w_1$}] (w1) at (4,0) {};
	\node[label={[label distance=2]135:$w_2$}] (w2) at (6,0) {};
	\node[label={[label distance=2]135:$w_4$}] (w4) at (4,-2) {};
	\node[label={[label distance=2]135:$w_3$}] (w3) at (6,-2) {};
	\end{scope}
	\begin{scope}[every edge/.style={draw=black}]
		\path[dashed] (v1) edge node {} (v4);
		\path[dashed] (v3) edge node {} (v2);
		\path[dashed] (w1) edge node {} (w2);
		\path[dashed] (w3) edge node {} (w4);
		\path (v1) edge node {} (v2);
		\path (v3) edge node {} (v4);
		\path (w1) edge node {} (w4);
		\path (w2) edge node {} (w3);
		\path (v4) edge node {} (w1);
		\path (v2) edge node {} (w4);
	\end{scope}
	\begin{scope}[every node/.style={draw=none,rectangle}]
	    \node (g) at (3,-0.25) {$g$};
	    \node (h) at (3,-2.25) {$h$};
	\end{scope}
	\end{tikzpicture}
	\caption{\protect\centering
	The gadget $J$.
  Edges of cost~zero and~one are illustrated by dashed and solid lines, respectively.}
	\label{f:algo:gadget}
\end{figure}
}

The graph $G_k = (V_k, E_k)$ is constructed as follows. We start with a
6-cycle $B_0 = b_1,b_2,b_3,b_4,b_5,b_6,b_1$ of unit-edges, of cost $6$. We
place $k$ copies $J_1,\ldots,J_k$ of the gadget $J$ in the following
manner. Let $v_j^i$ and $w_{\ell}^i$ denote the nodes $v_j$ and
$w_{\ell}$ of $J_i$, and let $g_i$ and $h_i$ denote the edges $g$ and
$h$ of $J_i$. First, we attach $J_1$ to $B_0$ by adding the two
unit-edges $e_1 = b_1v_1^1$ and $f_1 = b_4v_3^1$. Then, for each
$i\in\{2,\ldots,k\}$, we attach $J_i$ to $J_{i-1}$ by adding the two
unit-edges $e_i = w_2^{i-1}v_1^i$ and $f_i = w_3^{i-1}v_3^i$. The graph
$G_k$ is illustrated in Figure~\ref{f:algo:Gk}. Observe that $G_k$ is
a well-structured MAP instance.

{
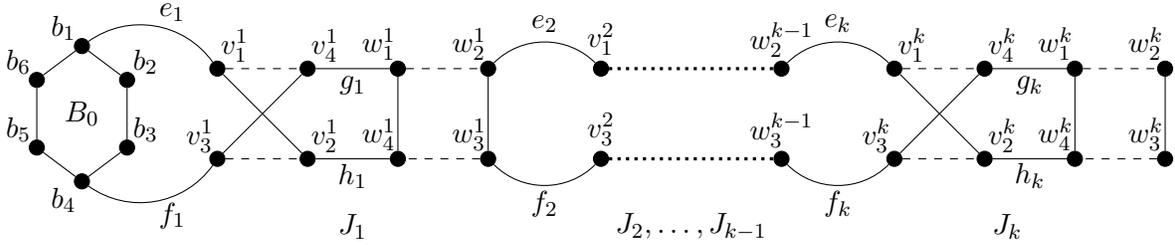
\begin{figure}[ht]
	\centering
	\begin{tikzpicture}[scale=0.6]
	\begin{scope}[every node/.style={circle, fill=black, draw, inner sep=0pt,
	minimum size = 0.2cm
	}]
    	\node[label={[label distance=0]135:$b_1$}] (b1) at (-3,0.5) {};
    	\node[label={[label distance=0]45:$b_2$}] (b2) at (-2,-0.25) {};
    	\node[label={[label distance=0]45:$b_3$}] (b3) at (-2,-1.75) {};
    	\node[label={[label distance=0]225:$b_4$}] (b4) at (-3,-2.5) {};
    	\node[label={[label distance=0]135:$b_5$}] (b5) at (-4,-1.75) {};
    	\node[label={[label distance=0]135:$b_6$}] (b6) at (-4,-0.25) {};
    	
    	\node[label={[label distance=0]75:$v_1^1$}] (v1) at (0,0) {};
    	\node[label={[label distance=0]75:$v_4^1$}] (v4) at (2,0) {};
    	\node[label={[label distance=0]105:$v_3^1$}] (v3) at (0,-2) {};
    	\node[label={[label distance=0]75:$v_2^1$}] (v2) at (2,-2) {};
    	\node[label={[label distance=0]105:$w_1^1$}] (w1) at (4,0) {};
    	\node[label={[label distance=0]105:$w_2^1$}] (w2) at (6,0) {};
    	\node[label={[label distance=0]105:$w_4^1$}] (w4) at (4,-2) {};
    	\node[label={[label distance=0]105:$w_3^1$}] (w3) at (6,-2) {};
    	
    	\node[label={[label distance=0]90:$v_1^2$}] (v21) at (8.5,0) {};
    	\node[label={[label distance=0]90:$v_3^2$}] (v23) at (8.5,-2) {};
    	
    	\node[label={[label distance=-6]90:$w_2^{k-1}$}] (w2k1) at (12.5,0) {};
    	\node[label={[label distance=-6]90:$w_3^{k-1}$}] (w3k1) at (12.5,-2) {};
    	
    	\node[label={[label distance=0]75:$v_1^k$}] (vv1) at (15,0) {};
    	\node[label={[label distance=0]75:$v_4^k$}] (vv4) at (17,0) {};
    	\node[label={[label distance=0]105:$v_3^k$}] (vv3) at (15,-2) {};
    	\node[label={[label distance=0]75:$v_2^k$}] (vv2) at (17,-2) {};
    	\node[label={[label distance=0]105:$w_1^k$}] (ww1) at (19,0) {};
    	\node[label={[label distance=0]105:$w_2^k$}] (ww2) at (21,0) {};
    	\node[label={[label distance=0]105:$w_4^k$}] (ww4) at (19,-2) {};
    	\node[label={[label distance=0]105:$w_3^k$}] (ww3) at (21,-2) {};
	
	\end{scope}
	\begin{scope}[every edge/.style={draw=black}]
		\path (b1) edge node {} (b2);
		\path (b2) edge node {} (b3);
		\path (b3) edge node {} (b4);
		\path (b4) edge node {} (b5);
		\path (b5) edge node {} (b6);
		\path (b6) edge node {} (b1);
		
		\path (b1) edge[bend left=45] node {} (v1);
		\path (b4) edge[bend right=45] node {} (v3);
		
		\path[dashed] (v1) edge node {} (v4);
		\path[dashed] (v3) edge node {} (v2);
		\path[dashed] (w1) edge node {} (w2);
		\path[dashed] (w3) edge node {} (w4);
		\path (v1) edge node {} (v2);
		\path (v3) edge node {} (v4);
		\path (w1) edge node {} (w4);
		\path (w2) edge node {} (w3);
		\path (v4) edge node {} (w1);
		\path (v2) edge node {} (w4);
		
		\path (w2) edge[bend left=45] node {} (v21);
		\path (w3) edge[bend right=45] node {} (v23);
		
		\path[very thick,dotted] (v21) edge node {} (w2k1);
		\path[very thick,dotted] (v23) edge node {} (w3k1);
		
		\path (w2k1) edge[bend left=45] node {} (vv1);
		\path (w3k1) edge[bend right=45] node {} (vv3);
		
		\path[dashed] (vv1) edge node {} (vv4);
		\path[dashed] (vv3) edge node {} (vv2);
		\path[dashed] (ww1) edge node {} (ww2);
		\path[dashed] (ww3) edge node {} (ww4);
		\path (vv1) edge node {} (vv2);
		\path (vv3) edge node {} (vv4);
		\path (ww1) edge node {} (ww4);
		\path (ww2) edge node {} (ww3);
		\path (vv4) edge node {} (ww1);
		\path (vv2) edge node {} (ww4);
		
	\end{scope}
	\begin{scope}[every node/.style={draw=none,rectangle}]
    	\node (e1) at (-1,1.25) {$e_1$};
    	\node (f1) at (-1,-3.25) {$f_1$};
	    \node (g1) at (3,-0.375) {$g_1$};
	    \node (h1) at (3,-2.375) {$h_1$};
	    
	    \node (e2) at (7.25,1) {$e_2$};
    	\node (f2) at (7.25,-3) {$f_2$};
	    \node (ek) at (13.75,1) {$e_k$};
    	\node (fk) at (13.75,-3) {$f_k$};
    	\node (gk) at (18,-0.375) {$g_k$};
	    \node (hk) at (18,-2.375) {$h_k$};
	    
	    \node (B0) at (-3,-1) {$B_0$};
	    \node (J1) at (3,-3.5) {$J_1$};
	    \node (J2) at (10.5,-3.5) {$J_2,\ldots,J_{k-1}$};
	    \node (Jk) at (17.5,-3.5) {$J_k$};
	\end{scope}
	\end{tikzpicture}
	\caption{\protect\centering
	The graph $G_k$ has copies $J_1,\dots,J_k$ of the gadget.
	All edges have cost~one, except the zero-edges of $J_1,\dots,J_k$.}
	\label{f:algo:Gk}
\end{figure}
} 

Note that the cost of any 2-edge~cover is $\ge 4k+6$, since it contains
all edges of $B_0$, as well as (at least) one unit-edge incident to
each of the eight nodes of each gadget. W.l.o.g, $\wh{\DTWO}$ consists
of $B_0$ and the two 4-cycles $v_1,v_2,v_3,v_4,v_1$ and
$w_1,w_2,w_3,w_4,w_1$ of each gadget. Note that our choice of
$\wh{\DTWO}$ is a bridgeless 2-edge~cover.

Consider the working of the algorithm on $G_k$.
Since $\wh{\DTWO}$ has no bridges, the algorithm proceeds to the
gluing~step.  We use the notation of Section~\ref{s:algo-last} to
describe the working of the gluing~step on $G_k$.

In each iteration of the gluing~step, we choose the 2ec-block
containing $B_0$ to be the root 2ec-block. In the first iteration,
$R=B_0$ is the root 2ec-block, and the block $B$ is the cycle $v_1,
v_2, v_3, v_4, v_1$ of $J_1$, i.e., the ear-augmentation step picks the
``ear'' $R,B,R$ and takes the edges $e,f$ (see Figure~\ref{f:gluing:1})
to be the edges $e_1,f_1$ of $G_k$.  Since the end nodes of $e_1$ and
$f_1$ are non-adjacent in $B$, we apply case~2 of the gluing~step by
taking $B'$ to be the 4-cycle $w_1, w_2, w_3, w_4, w_1$ of $J_1$ (see
Figure~\ref{f:gluing:3}); moreover, we take $\tilde{e}$ to be the edge
$g_1$ (of $G_k$), and we take $\tilde{P}$ to consist of $h_1$ and its
two end~nodes (in $G_k$). The algorithm buys the four edges $e_1, f_1,
g_1, h_1$ and ``sells'' one unit-edge (say $v_3v_4$), so the algorithm
incurs a cost of $7$ for $J_1$. In subsequent iterations, the same case
of the gluing~step is applied to each of the copies $J_2,\dots,J_k$ of
the gadget, hence, $\ALGO(G_k)$ incurs a cost of $7$ for each copy of
the gadget; thus, we have $\cost(\ALGO(G_k)) = 7k+6$.

On the other hand, the subgraph $G^*$ of $G_k$ (described below) is a
2-ECSS of cost $4k+7$; $E(G^*)$ consists of the union of $k+3$ sets of
edges, namely,
the set of edges $\{e_i, f_i, g_i, h_i\}$ for each $i \in \{1,\ldots,k\}$, 
the set of zero-edges of $G_k$, 
$E(B_0)$,
and the singleton $\{w_2^kw_3^k\}$ (the edge $w_2^kw_3^k$ is indicated
by the right-most vertical line in Figure~\ref{f:algo:Gk}).  Hence,
$\opt(G_k) \leq 4k+7$.
\end{proof}
}

\bigskip
\noi\textbf{Acknowledgments}: We are grateful to several colleagues for
their careful reading of preliminary drafts and for their comments.
We thank an anonymous reviewer for a thorough review.

{\small
\bibliographystyle{abbrv}
\bibliography{map}
}

\end{document}